\documentclass[journal]{IEEEtran}

\usepackage{amsmath, amssymb,bm}
\usepackage{url}
\usepackage{enumerate,mathtools}
\usepackage{amsthm}
\usepackage{cite}
\usepackage{microtype}

\usepackage{topcapt} 
\usepackage{booktabs}

\usepackage{tikz}
\usetikzlibrary{shapes.geometric, arrows}

\newtheorem{theorem}{Theorem}
\newtheorem{lemma}[theorem]{Lemma}

\theoremstyle{definition}
\newtheorem{definition}{Definition}
\newtheorem{remark}{Remark}
\newtheorem{assumption}{Assumption}

\newcommand{\cD}{\mathcal{D}}
\newcommand{\cE}{\mathcal{E}}
\newcommand{\cF}{\mathcal{F}}
\newcommand{\cG}{\mathcal{G}}
\newcommand{\cH}{\mathcal{H}}
\newcommand{\cI}{\mathcal{I}}
\newcommand{\cJ}{\mathcal{J}}

\newcommand{\cL}{\mathcal{L}}
\newcommand{\cM}{\mathcal{M}}

\newcommand{\cV}{\mathcal{V}}

\newcommand{\goto}{\rightarrow} 

\newcommand{\bEx}{\ensuremath{\mathbb{E}}}

\newcommand{\ex}[1]{\ensuremath{\mathbb{E}\left[ #1\right]}}

\newcommand{\gmid}{\! \mid \!}
\newcommand{\DKL}[2]{\ensuremath{D_\mathrm{KL}\left( #1 \, \middle \| \, #2 \right)}}

\DeclareMathOperator{\cov}{\sf Cov}
\DeclareMathOperator{\mmse}{\sf mmse}

\DeclareMathOperator{\var}{\sf Var}

\DeclareMathOperator{\sgn}{sgn}

\DeclareMathOperator{\gtr}{tr}

\newcommand{\reals}{\mathbb{R}}
\newcommand{\integers}{\mathbb{N}}

\newcommand{\normal}{\mathcal{N}}

\newcommand{\SE}{\cE}
\newcommand{\PSE}{V}

\newcommand{\IMI}{Z}
\newcommand{\IMID}{\cJ}
\newcommand{\PMID}{J}
\newcommand{\Ybar}{\bar{Y}}
\newcommand{\Zbar}{\bar{Z}}

\newcommand{\IR}{\cI_{\mathrm{RS}}}
\newcommand{\MR}{\cM_{\mathrm{RS}}}

\let\originalleft\left
\let\originalright\right
\renewcommand{\left}{\mathopen{}\mathclose\bgroup\originalleft}
\renewcommand{\right}{\aftergroup\egroup\originalright}

\newcommand{\fp}{\mathrm{FP}}
\newcommand{\rs}{\mathrm{RS}}

\allowdisplaybreaks

\global\long\def\dd{\mathrm{d}}

\global\long\def\rgz{\mathbb{R}_{+}}

\global\long\def\nifunB{\mathcal{D}_{0}}

\global\long\def\mmse{\mathsf{mmse}}

\global\long\def\fp{\mathrm{FP}}

\title{The Replica-Symmetric Prediction for Compressed Sensing with Gaussian Matrices is Exact} 
\author{
\IEEEauthorblockN{Galen Reeves\IEEEauthorrefmark{1}\IEEEauthorrefmark{2} and Henry D. Pfister\IEEEauthorrefmark{1}}\\
\thanks{The work of G.\ Reeves was supported in part by funding from the Laboratory for Analytic Sciences (LAS).  The work of H. Pfister was supported part by the NSF under Grant No.~1545143.  Any opinions, findings, conclusions, and recommendations expressed in this material are those of the authors and do not necessarily reflect the views of the sponsors.}
\IEEEauthorblockA{\IEEEauthorrefmark{1}Department of Electrical Engineering, Duke University\\ 
\IEEEauthorrefmark{2}Department of Statistical Science, Duke University}}

\begin{document}


\maketitle

\begin{abstract}
This paper considers the fundamental limit of compressed sensing for i.i.d.\ signal distributions and i.i.d.\ Gaussian measurement matrices. Its main contribution is a rigorous  characterization of the asymptotic mutual information (MI) and  minimum mean-square error (MMSE) in this setting. Under mild technical conditions, our results show that the limiting MI and MMSE are equal to the values predicted by the replica method from statistical physics. This resolves a well-known problem that has remained open for over a decade. 
\end{abstract}

\section{Introduction}

The canonical compressed sensing problem can be formulated as follows.  The signal is a random $n$-dimensional vector $X^n = (X_1,\ldots,X_n)$ whose entries are drawn independently from a common distribution $P_X$ with finite variance. The signal is observed using noisy linear measurements of the form
\begin{align}
Y_k = \langle A_k, X^n \rangle  + W_k,\notag 
\end{align}
where $\{A_k\}$ is a sequence of $n$-dimensional measurement vectors, $\{W_k\}$ is a sequence of standard Gaussian random variables, and $\langle \cdot, \cdot \rangle$ denotes the Euclidean inner product between vectors. The primary goal is to reconstruct $X^n$ from the set of $m$ measurements $\{(Y_k, A_k)\}_{k=1}^m$. Since the reconstruction problem is symmetric under simultaneous scaling of $X^n$ and $\{W_k\}$, the unit-variance assumption on $\{W_k\}$ incurs no loss of generality. In matrix form, the relationship between the signal and a set of $m$ measurements is given by
\begin{align}
Y^m = A^m X^n  + W^m \label{eq:model}
\end{align}
where $A^m$ is an $m \times n $ measurement matrix whose $k$-th row is $A_k$.

This paper analyzes the minimum mean-square error (MMSE) reconstruction in the asymptotic setting where the number of measurements $m$ and the signal length $n$ increase to infinity. The focus is on scaling regimes in which the measurement ratio $\delta_n = m/n$ converges to a number $\delta \in (0,\infty)$. The objective is to show that the normalized  mutual information (MI) and MMSE converge to limits, 
\begin{align*}
\cI_n (\delta_n) & \triangleq \frac{1}{n} I\left(X^n ;  Y^{m} \gmid A^{m}  \right)    \to \cI(\delta) \\
\cM_n (\delta_n) & \triangleq \frac{1}{n} \mmse\left(X^n \gmid Y^{m} , A^m \right)     \to \cM(\delta), 
\end{align*}
almost everywhere and to characterize these limits in terms of the measurement ratio $\delta$ and the signal distribution $P_X$. 
We note that all mutual informations are computed in nats

Using the replica method from statistical physics, Guo and Verd\'{u}~\cite{guo:2005} provide an elegant characterization of these limits in the setting of i.i.d.\ measurement matrices.  Their result was stated originally as a generalization of Tanaka's replica analysis of code-division multiple-access (CDMA) with binary signaling \cite{tanaka:2002}.  The replica method was also applied specifically to compressed sensing in~\cite{kabashima:2009,guo:2009, rangan:2012, reeves:2012,reeves:2012a, tulino:2013}.  The main issue, however, is that  the replica method is not rigorous. It requires an exchange of limits that is unjustified, and it requires the assumption of replica symmetry, which is unproven in the context of the compressed sensing problem. 

The main result of this paper is that replica prediction is correct for i.i.d.\ Gaussian measurement matrices provided that the signal distribution, $P_X$, has bounded fourth moment and satisfies a certain `single-crossing' property. The proof differs from previous approaches in that we first establish some properties of the finite-length MMSE and MI sequences, and then use these properties to uniquely characterize their limits.

\subsection{The Replica-Symmetric Prediction}
We now describe the results predicted by the replica method. For a  signal distribution $P_X$, the function $R : \reals_+^2 \goto \reals_+$  is defined as
\begin{align}
R(\delta, z) &= I_X\left( \frac{\delta}{1+z} \right)  + \frac{\delta}{2}\left[ \log\left( 1+ z\right) - \frac{ z}{ 1 +  z} \right], \notag
\end{align}
where $I_X(s)  = I\left( X  ; \sqrt{s} X + N \right)$  is the scalar mutual information function (in nats) of $X \sim P_X$  under independent Gaussian noise $N \sim \normal(0, 1)$ with signal-to-noise ratio $s \in \reals_+$~\cite{guo:2005,guo:2009}.
\begin{definition}\label{def:rs_limits} 
The replica-MI function $\IR\colon \rgz \to \rgz$ and the replica-MMSE function $\MR\colon \rgz \to \rgz$ are defined as
\begin{align}
\IR (\delta) &= \min_{z \ge 0} R(\delta, z) \notag\\
\MR (\delta) & \in \arg\min_{z \ge 0 } R(\delta, z). \notag
\end{align}
\end{definition}
\vspace{-0.5mm}

The function $\IR(\delta)$ is increasing because $R(\delta,z)$ is increasing in $\delta$ and it is concave because it is the pointwise minimum of concave functions.  The concavity implies that $\IR'(\delta)$ exists almost everywhere and is decreasing.  It can also be shown that $\MR(\delta)$ is also decreasing and, thus, continuous almost everywhere.  If the minimizer is not unique, then $\MR(\delta)$ may have jump discontinuities and may not be uniquely defined at those points; see Figure~\ref{fig:fp_curve}.

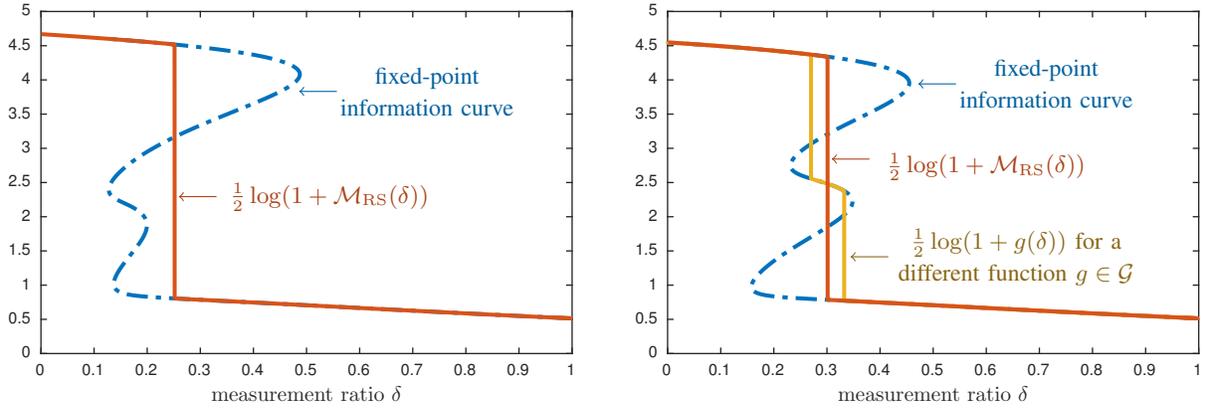
\begin{figure*}[!ht]
\centering
\begin{tikzpicture}
    	\node[anchor=south west,inner sep=0] at (0,0) {\includegraphics{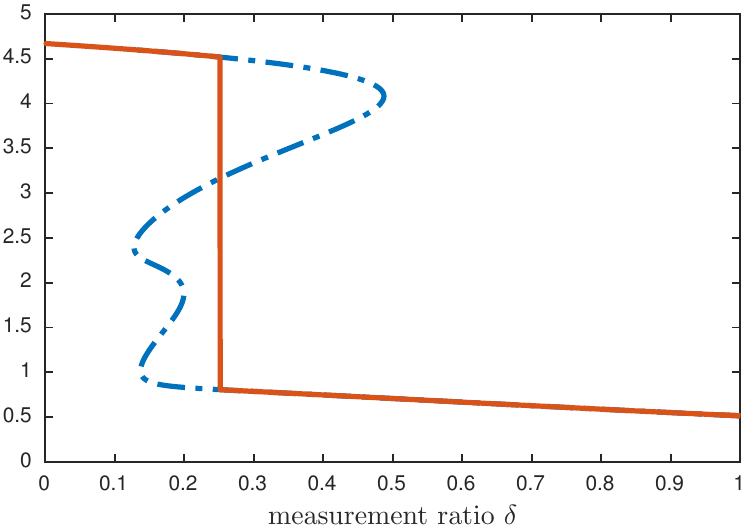}};
        \draw [<-,color = {rgb:red,0;green,0.447;blue,0.741}] (3.9, 4.2) -- (4.4,4.2);
        \node[ align =center, text = {rgb:red,0;green,0.447;blue,0.741}] at (5.6,4.2) {\small fixed-point \\ \small information curve};
        \draw [<-,color = {rgb:red,0.85;green,0.325;blue,0.098}] (2.3, 2.8) -- (2.8,2.8);
 	\node[ align =center, text = {rgb:red,0.85;green,0.325;blue,0.098}] at (4.3,2.8) { \small $\frac{1}{2}\log(1 +\MR(\delta))$};
\end{tikzpicture}
\hspace{.5cm}
\begin{tikzpicture}
        \node[anchor=south west,inner sep=0] at (0,0) {\includegraphics{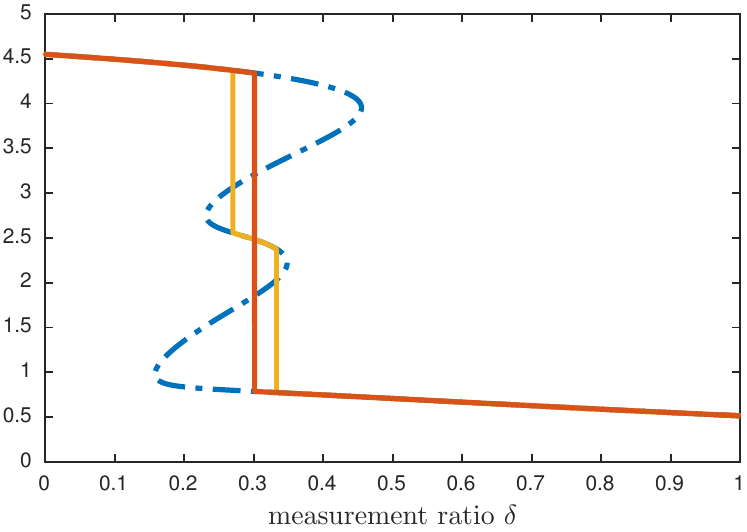}};
        \draw [<-,color = {rgb:red,0;green,0.447;blue,0.741}] (3.8, 4.3) -- (4.3,4.3);
        \node[ align =center, text = {rgb:red,0;green,0.447;blue,0.741}] at (5.5,4.3) {\small fixed-point \\ \small information curve};
        \draw [<-,color = {rgb:red,0.85;green,0.325;blue,0.098}] (2.65, 3.3) -- (3.15,3.3);
 	\node[ align =center, text = {rgb:red,0.85;green,0.325;blue,0.098}] at (4.7,3.2) {\small $\frac{1}{2}\log(1 +\MR(\delta))$};
         \draw [<-,color = {rgb:red,0.929;green,0.694;blue,0.125}] (2.9 , 2) -- (3.4,2);
 	\node[align =center, text ={rgb:red,0.929;green,0.694;blue,0.125}] at (5.1,2) {\small $\frac{1}{2}\log(1 +g(\delta))$ for a \\ {\small different function $g \in \cG$}};
\end{tikzpicture}

\caption{\label{fig:fp_curve} 
Plot of the replica-MMSE as a function of the measurement ratio $\delta$. The signal distribution is given by a three-component Gaussian mixture of the form $P_X = 0.4 \normal(0, 5) + \alpha \normal(40,5) + (0.6-\alpha) \normal(220,5)$. In the left panel, $\alpha = 0.1$ and the distribution satisfies the single-crossing property. In the right panel,  $\alpha = 0.3$ and the distribution does not satisfy the single-crossing property. The fixed-point information curve (dashed blue line) is given by $\frac{1}{2} \log(1 + z)$ where $z$ satisfies the fixed-point equation $R_z(\delta,z) = 0$.}
\vspace*{-.25cm}
\end{figure*}

\vspace{-0.5mm}
\subsection{Statement of Main Result} 

In order to state our results, we need some further definitions. Let $R_z(\delta,z) = \frac{\partial}{\partial z} R(\delta, z)$ denote the partial derivative of $R(\delta, z)$ with respect to $z$. The \emph{fixed-point curve} FP is the set of $(\delta, z)$ pairs where $z$ is a stationary point of $R(\delta, z)$, i.e.,
\begin{align}
\mathrm{FP} =\left \{ (\delta, z) \in \reals_+^{2} \, : \, R_z(\delta, z) = 0 \right\}. \notag
\end{align}
To emphasize the connection with mutual information, we often plot this curve using the change of variables $z\mapsto \frac{1}{2}\log(1+z)$.
The resulting curve, $(\delta,\frac{1}{2}\log(1+z))$, is called the 
 \emph{fixed-point information curve}; see Figure~\ref{fig:fp_curve}.

\begin{definition}[Single-Crossing Property] Informally, a signal distribution $P_X$ has the single-crossing property if the replica-MMSE crosses the fixed-point curve FP at most once. A formal definition of the single-crossing property is given in Section~\ref{sec:single_crossing}. 
\end{definition}

\begin{assumption}[IID Gaussian Measurements]
The rows of the measurement matrix $\{A_k\}$ are independent Gaussian vectors with mean zero and covariance $n^{-1} I_n$. Furthermore, the noise $\{W_k\}$ is i.i.d.\ Gaussian with mean zero and variance one. 
\end{assumption} 

\begin{assumption}[IID Signal Entries] 
The signal entries $\{X_i\}$ are independent copies of a random variable $X \sim P_X$ with bounded fourth moment $\ex{ X^4} \le B$.
\end{assumption}

\begin{assumption}[Single-Crossing Property]
The signal distribution $P_X$ satisfies the single-crossing property.
\end{assumption}

\begin{theorem} \label{thm:limits} Under Assumptions 1-3, we have
\begin{enumerate}[(i)]

\item The sequence of MI functions $\cI_n (\delta)$ converges to the replica prediction.  In other words, for all $\delta \in \rgz$,
 \begin{align}
 \lim_{n \goto \infty}  \cI_n (\delta ) = \IR(\delta). \notag
\end{align}

\item The sequence of MMSE functions $\cM_n (\delta)$ converges almost everywhere to the replica prediction.  In other words, for all continuity points of $\MR(\delta)$, 
 \begin{align}
  \lim_{n \goto \infty}  \cM_n(\delta )=  \MR(\delta). \notag
\end{align}

\end{enumerate}

\end{theorem}

\begin{remark}
The primary contribution of Theorem~\ref{thm:limits} is for the case where $\MR(\delta)$ has a discontinuity.  This occurs, for example, in applications such as compressed sensing with sparse priors and CDMA with finite alphabet signaling.
For the special case where $\MR(\delta)$ is continuous, the validity of the replica prediction can also be established by combining the AMP analysis with the I-MMSE relationship~\cite{guo:2006, donoho:2009a,donoho:2011,  bayati:2011,bayati:2012a}.
\end{remark}

\begin{remark}
For a given signal distribution $P_X$ the single-crossing property can be verified by numerically evaluating the replica-MMSE and checking whether it crosses the fixed-point curve more than once. 
\end{remark}

\subsection{Related Work} 
The replica method was developed originally to study mean-field approximations in spin glasses~\cite{edwards:1975,mezard:2009}.
It was first applied to linear estimation problems in the context of CDMA wireless communication \cite{tanaka:2002, muller:2003, guo:2005}, with subsequent work focusing on the compressed sensing directly \cite{guo:2009, kabashima:2009, rangan:2012, reeves:2012,reeves:2012a,  tulino:2013}. 

Within the context of compressed sensing, the results of the replica method have been proven rigorously in a number of settings. One example is given by message passing on matrices with special structure, such as sparsity \cite{montanari:2006, guo:2006, baron:2010} or spatial coupling~\cite{kudekar:2010, krzakala:2012, donoho:2013}. 
However, in the case of i.i.d.\ matrices, the results are limited to signal distributions with a unique fixed point \cite{donoho:2009a, bayati:2011} (e.g., Gaussian inputs~\cite{verdu:1999,tse:1999}). For the special case of i.i.d.\ matrices with binary inputs, it has also been shown that the replica prediction provides an upper bound  for the asymptotic mutual information \cite{korada:2010}. Bounds on the locations of discontinuities in the MMSE with sparse priors have also been obtain by analyzing the problem of approximate support recovery~\cite{reeves:2012, reeves:2012a,  tulino:2013}. 

Recent work by Huleihel and Merhav \cite{huleihel:2016} addresses the validity of the replica MMSE directly in the case of Gaussian mixture models, using tools from statistical physics and random matrix theory \cite{merhav:2010,merhav:2011a}.

\subsection{Notation} 

We use $C$ to denote an absolute constant and $C_\theta$ to denote a number that depends on a parameter $\theta$. In all cases, the numbers $C$ and $C_\theta$ are positive and finite, although their values change from place to place. The Euclidean norm is denoted by $\|\cdot \|$.  The positive part of a real number $x$ is denoted by $(x)_+ \triangleq \max(x,0)$. The nonnegative real line $[0,\infty)$ is denoted by $\reals_+$ and the positive integers $\{1,2,\cdots\}$ are denoted by $\integers$. For each  $n \in \integers$ the set $\{1,2,\cdots , n\}$ is denoted by $[n]$.  

The  joint distribution of the random variables $X,Y$ is denoted by $P_{X,Y}$ and the marginal distributions are denoted by $P_X$ and $P_Y$. The conditional distribution of $X$ given $Y= y$ is denoted by $P_{X\mid Y=y}$ and the conditional distribution of $X$ corresponding to a random realization of $Y$ is denoted by $P_{X| Y}$.  The expectation over a single random variable is denoted by $\bEx_{X}$.  For example, this implies that $\bEx \left[ f(X,Y) | Y \right] = \bEx_X \left[ f(X,Y) \right]$.

Using this notation, the mutual information between $X$ and $Y$ can be expressed in terms of the expected Kullback-Leibler divergence as follows: 
\begin{align}
I(X;Y) &= \DKL{P_{X,Y}}{P_X  \times P_Y} \notag \\
& =  \bEx\left[ \DKL{P_{X\mid Y}}{ P_X} \right] \notag\\
&  =  \bEx\left[ \DKL{P_{Y\mid X}}{ P_Y} \right],\notag
\end{align}
where the expectation in the second line is with respect to $Y$ and the expectation in the third line is with respect to $X$.

The conditional variance of a random variable $X$ given $Y$ is denoted by
\begin{align}
\var(X \gmid Y) &= \ex{ ( X - \ex{X  \gmid Y})^2\; \middle |  \;  Y}, \notag
\end{align}
and the conditional covariance matrix of a random vector $X$ given $Y$ is denoted by
\begin{align}
\cov(X \gmid Y) &= \ex{ ( X - \ex{X  \gmid Y}) ( X - \ex{X \gmid Y})^T \; \middle |  \;  Y}. \notag
\end{align}
The conditional variance and conditional covariance matrix are random because they are functions of the random conditional distribution $P_{X\mid Y}$. 

The minimum mean-square error (MMSE) of $X$ given $Y$ is defined to be the expected squared difference between $X$ and its conditional expectation and is denoted by
\begin{align}
\mmse(X \gmid Y) & = \ex{ \| X - \ex{X \gmid Y}\|^2}.  \notag
\end{align}
Since the expectation is taken with respect to both $X$ and $Y$, the MMSE is a deterministic functions of the joint distribution $P_{X,Y}$. The MMSE can also be expressed in terms of the expected trace of the conditional covariance matrix:
\begin{align}
\mmse(X \gmid Y) & =  \bEx\left[ \gtr\left( \cov(X \gmid Y)\right) \right]. \notag
\end{align}


\section{Overview of Main Steps in Proof}


We begin with some additional definitions.  The finite-length MI sequence $I : \integers^2 \to \reals_+$ and MMSE sequence $M : \integers^2 \to \reals_+$ are defined according to
\begin{align}
I_{m,n} &\triangleq I(X^n; Y^m  \! \mid \! A^m), \notag\\
M_{m,n} &\triangleq  \frac{1}{n} \mmse(X^n \! \mid  \! Y^m, A^m ),\notag
\end{align}
where the relationship between the $n$-dimensional signal $X^n$, $m$-dimensional measurements $Y^m$, and $m \times n$ measurement matrix $A^m$ is given by the measurement model \eqref{eq:model}. Furthermore, the first and second order MI difference sequences are defined according to
\begin{align}
I'_{m,n} &\triangleq I_{m+1,n} - I_{m,n} \notag\\
I''_{m,n} &\triangleq   I'_{m+1,n} - I'_{m,n}. \notag
\end{align}
To simplify notation we will often drop the explicit dependence on the signal length $n$ and simply write  $I_m$, $M_m$, $I'_m$, and $I''_m$. 

\subsection{Properties of Finite-Length Sequences} 

At its foundation, our proof relies on certain relationships between the MI and MMSE sequences.  Observe that by the chain rule for mutual information, 
\begin{align}
\underbrace{I(X^n;  Y^m  \gmid  A^m)}_{I_{m,n}}
& = \sum_{k=1}^{m-1} \underbrace{I(X^n ; Y_{k+1} \gmid Y^k, A^{k+1})}_{I'_{k,n}}. \notag 
\end{align}
Here, the conditioning in the mutual information on the right-hand side depends only on $A^{k+1}$ because the measurement vectors are generated independently of the signal.

The above decomposition shows that the MI difference is given by  $I'_{m,n} = I(X^n ; Y_{m+1} \gmid Y^m, A^{m+1})$. In other words, it is the mutual information between the signal and a new measurement $Y_{m+1}$,  conditioned on the previous data $(Y^m, A^{m+1})$.  One of the key steps in our proof is to show that the MI difference and MMSE satisfy the following relationship almost everywhere
\begin{align}
I'_{m,n} \approx \frac{1}{2} \log\left(1 + M_{m,n}\right). \label{eq:Ip_MMSE_inq}
\end{align}
Our approach relies on the fact that the gap between the right and left sides of \eqref{eq:Ip_MMSE_inq} can be viewed as a measure of the non-Gaussianness of the conditional distribution of the new measurement. By relating this non-Gaussianness to certain properties of the posterior distribution, we are able to show that \eqref{eq:Ip_MMSE_inq} is tight whenever the second order MI difference sequence is small. The details of these steps are given in Section~\ref{proof:thm:I_MMSE_relationship}.

Another important relationship that is used in our proof is the following fixed-point identity for the MMSE
\begin{align}
M_{m,n} \approx \mmse_X\left(  \frac{m/n}{1 + M_{m,n}} \right) . \label{eq:M_fixed_point} 
\end{align}
In words, this says that the MMSE of the compressed sensing problem is approximately equal to that of a scalar problem whose signal-to-noise ratio  is a function of the measurement rate. In Section~\ref{proof:thm:MMSE_fixed_point} it is shown that the tightness in \eqref{eq:M_fixed_point}  can be bounded in terms of the tightness of  \eqref{eq:Ip_MMSE_inq}.

\subsection{Asymptotic Constraints} 

The previous subsection focused on relationships between the finite-length MI and MMSE sequences. To characterize these relationships in the asymptotic setting of interest, the finite-length sequences are extended to functions of a continuous parameter $\delta \in \reals_+$ according to
\begin{align}
\cI'_n(\delta) & = I'_{\lfloor \delta n \rfloor, n}   \notag\\
\cI_n(\delta) & = \int_0^\delta \cI'_n(\gamma)\,  \dd \gamma \notag\\
\cM_n(\delta) &= M_{\lfloor \delta n \rfloor,n}. \notag
\end{align}
This choice of interpolation has the convenient property that the MI function $\cI_n$ is continuous and differentiable almost everywhere. Furthermore, by construction, $\cI_n$ corresponds to the normalized mutual information and satisfies $\cI_n\left( \tfrac{m}{n} \right) = \frac{1}{n} I_{m, n}$ for all integers $m$ and $n$. 

With this notation in hand, we are now ready to state two of the main theorems in this paper. These theorems provide precise bounds on the relationships given in \eqref{eq:Ip_MMSE_inq} and \eqref{eq:M_fixed_point}. The proofs are given in Section~\ref{proof:thm:I_MMSE_relationship} and Section~\ref{proof:thm:MMSE_fixed_point}.

\begin{theorem}\label{thm:I_MMSE_relationship}
Under Assumptions 1 and 2, the MI and MMSE functions satisfy
\begin{align}
   \int_{0}^{\delta} \left| \cI_{n}'(\gamma)   - \frac{1}{2} \log\left(1 + \cM_{n}(\gamma)  \right) \right| \dd \gamma &  \le C_{B, \delta}  \cdot n^{-r}, \notag
\end{align}
for all $n \in \integers$ and $\delta \in \reals_+$ where $r \in (0,1)$ is a universal constant. 
\end{theorem}

\begin{theorem}\label{thm:MMSE_fixed_point} 
Under Assumptions 1 and 2, the MMSE function satisfies
\begin{align}
  \int_{0}^{\delta} \left| \cM_{n}(\gamma)    - \mmse_X\left( \frac{\gamma  }{1 + \cM_n(\gamma) }    \right) \right| \dd \gamma&  \le C_{B,\delta} \cdot n^{-r}, \notag
\end{align}
for all $n \in \integers$ and $\delta \in \reals_+$ where $r \in (0,1)$ is a universal constant. 
\end{theorem}

The bounds given in Theorems~\ref{thm:I_MMSE_relationship} and \ref{thm:MMSE_fixed_point} are with respect the integrals over $\cI'_n$ and $\cM_n$, and thus prove convergence in $L_1$ over bounded intervals. This is sufficient to show that the relationships hold almost everywhere in the limit. Importantly, though, these bounds still allow for the possibility that the relationships do not hold at countably many points, and thus allow for the possibility of phase transitions.

For distributions that have a phase transition, our proof technique requires a boundary condition for the mutual information. This boundary condition is used to determine the location of the phase transition. The next result shows that the replica-MI is equal to the MI-function in the limit as the measurement rate increases to infinity, and thus the replica-MI can be used as a boundary condition. The proof is given in Section~\ref{proof:thm:I_m_boundary}.

\begin{theorem}\label{thm:I_m_boundary}  Under Assumptions 1 and 2, the MI function satisfies
\begin{align}
 \left |   \cI_{n}(\delta) -  \IR\left( \delta \right) \right |  \le C \cdot \delta^{-\frac{1}{2}}, \notag
\end{align}
for all $n \in \integers$ and $\delta \ge 4$.
\end{theorem}

At first glance, it may seem surprising that Theorem~\ref{thm:I_m_boundary} holds for all signal lengths. However, this bound is tight in the regime where the number of measurements is much larger than the number of signal entries. From the rotational invariance of the Gaussian noise and monotonicity of the mutual information with respect to the signal-to-noise ratio,  one can obtain the sandwiching relation
\begin{align}
n\,  \ex{  I_X\left(  \sigma^2_\text{min}(A^m) \right)} \le I_{m,n} \le n\,  \ex{  I_X\left(  \sigma^2_\text{max}(A^m) \right)},\notag
\end{align}
where the upper and lower bounds depend only on the minimum and maximum singular values of the random $m \times n$ measurement matrix. For fixed $n$, it is well known that the ratio of these singular values converges to one almost surely in the limit as $m$ increases to infinity. Our proof of Theorem~\ref{thm:I_m_boundary} uses a more refined analysis, based on the QR decomposition, but the basic idea is the same.

\subsection{Uniqueness of Limit}\label{sec:properties_replica_prediction}

The final step of the proof is to make the connection between the asymptotic constraints on the MI and MMSE described in the previous subsection and the replica-MI and replica-MMSE functions given in Definition~\ref{def:rs_limits}. 

We begin by describing two functional properties of the replica limits. The first property follows from the fact that the MMSE is a global minimizer of the function $R(\delta, z)$ with respect to its second argument. Since $R(\delta,z)$ is differentiable for all $\delta,z \in \reals_+$, any minimizer $z^*$ of $R(\delta, z)$ must satisfy the equation 
  $R_z (\delta, z^*)  = 0$, where $R_z (\delta,z) = \frac{\partial}{\partial z} R(\delta,z)$. Using the I-MMSE relationship \cite{guo:2005a}, it can be shown that, for each $\delta \in \reals_+$, the replica-MMSE $\MR(\delta)$ satisfies the fixed-point equation
  \begin{align}
\MR(\delta)  = \mmse_X\left( \frac{\delta}{1 + \MR(\delta) }   \right) . \label{eq:fixed_point_replica} 
\end{align}
Note that if this fixed-point equation has a unique solution, then it provides an equivalent definition of the replica-MMSE. However, in the presence of multiple solutions, only the solutions that correspond to global minimizers of $R(\delta, z)$ are valid. 

The second property relates the derivative of the replica-MI to  the replica-MMSE. Specifically,  as a consequence of the envelope theorem~\cite{milgrom:2002} and  \eqref{eq:fixed_point_replica}, it can be shown that the derivative of the replica-MI satisfies
\begin{align}
\IR'(\delta) = \frac{1}{2} \log\left( 1 + \MR(\delta) \right), \label{eq:derivative_replica} 
\end{align}
almost everywhere.

These properties show that the replica limits $\IR$ and $\MR$ satisfy the relationships given in Theorems~\ref{thm:I_MMSE_relationship} and \ref{thm:MMSE_fixed_point} with equality. 
In order to complete proof we need to show that, in conjunction with a boundary condition imposed by the large measurement rate limit, the constraints \eqref{eq:fixed_point_replica} and \eqref{eq:derivative_replica} provide an equivalent characterization of the replica limits. 

\begin{definition}
\label{lem:replica_mmse_set} For a given signal distribution $P_X$, let $\cV$ be the subset of non-increasing functions from $\rgz \to \rgz$ such that all $g\in \cV$ satisfy the fixed-point condition 
\begin{align} 
g(\delta) = \mmse_X \left(\frac{\delta}{1+g(\delta)} \right) \label{eq:fixed_point_g}. 
\end{align}
Furthermore, let $\cG \subseteq \cV$ be the subset such that all $g \in \cG$ also satisfy the large measurement rate boundary condition
\begin{align}
\lim_{\delta \to \infty} \left|  \int_{0}^{\delta}\frac{1}{2} \log \left(1 + g(\gamma) \right)  \dd \gamma  - \IR(\delta)   \right| = 0. \notag
\end{align}

\end{definition}

In Section~\ref{sec:single_crossing}, it is shown that if the signal distribution $P_{X}$ has the single-crossing property, then $\MR(\delta)$ has at most one discontinuity and all $g\in \cG$ satisfy $g(\delta)=\MR(\delta)$ almost everywhere. In other words, the single-crossing property provides a sufficient condition under which the replica limits can be obtained uniquely from \eqref{eq:derivative_replica} and \eqref{eq:fixed_point_g}.   A graphical illustration is provided in Figure~\ref{fig:fp_curve}.

%
%

%
%
%
%

\section{Properties of MI and MMSE}\label{sec:MI_MMSE_bounds}

\subsection{Single-Letter Functions}  

The single-letter MI and MMSE functions corresponding to a real valued input distribution $P_X$ under additive Gaussian noise are defined by
\begin{align}
I_X(s) &\triangleq I(X; \sqrt{s} X + N) \notag \\
\mmse_X(s) & \triangleq \mmse(X\mid \sqrt{s} X + N), \notag
\end{align}
where $X \sim P_X$ and $N \sim \normal(0,1)$ are independent and $s \in \reals_+$ parametrizes the signal-to-noise ratio. Many  properties of these functions have been studied in the literature \cite{guo:2005a,guo:2011,wu:2011,wu:2012a}. The function $I_X(s)$ is concave and non-decreasing with $I_X(0) = 0$. If $I_X(s)$ is finite for some $s > 0$ then it is finite for all $s \in \reals_+$ \cite[Theorem~6]{wu:2012a}. The MMSE function is non-increasing with $\mmse_X(0) = \var(X)$. Both $I_X(s)$ and $\mmse_X(s)$ are infinitely differentiable on $(0,\infty)$ \cite[Proposition~7]{guo:2011}.

The so-called I-MMSE relationship \cite{guo:2005a} states that
\begin{align}
\frac{\dd}{ \dd s} I_X(s) = \frac{1}{2} \mmse_X(s).  \label{eq:I-MMSE}
\end{align}
This relationship was originally stated for input distributions with finite second moments \cite[Theorem~1]{guo:2005a}, and was later shown to hold for any input distributions with finite mutual information \cite[Theorem~6]{wu:2012a}. This relationship can also be viewed as a consequence of De-Bruijn's identity \cite{stam:1959}. 

Furthermore, it is well known that under a second moment constraint, the MI and MMSE are maximized when the input distribution is Gaussian. This yields the following inequalities 
\begin{align}
I_X(s) &\le \frac{1}{2} \log(1 + s \var(X))  \label{eq:IX_Gbound}\\
\mmse_X(s) &\le  \frac{ \var(X)}{ 1 + s \var(X)}. \label{eq:mmseX_Gbound}
\end{align}
More generally, the MMSE function satisfies the upper bound $\mmse_X(s) \le 1/s$, for every input distribution $P_X$ and  $s > 0$ \cite[Proposition~4]{guo:2011}. Combining this inequality with \eqref{eq:I-MMSE} leads to 
\begin{align}
I_X(t ) - I_X(s) \le \frac{1}{2} \log\left( \frac{t}{s}\right), \quad  0 < s \le t, \label{eq:IX_smooth}
\end{align}
which holds for any input distribution with finite mutual information. 

Finally, the derivative of the MMSE with respect to $s$ is given by the second moment of the conditional variance \cite[Proposition~9]{guo:2011},
\begin{align}
\frac{\dd}{\dd s} \mmse_X(s) = - \ex{ \left(\var(X \gmid Y) \right)^2}. \label{eq:mmseXp}
\end{align}
This relationship shows that the slope of $\mmse_X(s)$ at $s=0$ is finite if and only if $X$ has bounded forth moment. It also leads to the following result, which is proved in Appendix~\ref{proof:lem:mmseX_bound}.

\begin{lemma}\label{lem:mmseX_bound} The single-letter MMSE function satisfies the following bounds:
\begin{enumerate}[(i)]
\item For any input distribution $P_X$ with finite fourth moment and $s,t \in \reals_+$, 
\begin{align}
 \left | \mmse_X\left(s  \right ) - \mmse_X\left(t \right)  \right| &\le 4 \ex{X^4} |s - t|. \label{eq:mmseX_smooth_a}
\end{align}
\item For every input distribution $P_X$ and $s,t \in (0,\infty)$, 
\begin{align}
\left | \mmse_X\left(s  \right ) - \mmse_X\left(t \right)  \right| &\le  12\,  \left | \frac{1}{s} - \frac{1}{t} \right| .  \label{eq:mmseX_smooth_b}
\end{align}
\end{enumerate}
\end{lemma}

\subsection{Multivariate MI and MMSE}\label{sec:multivariate_MI_MMSE}

From the chain rule for mutual information, we see that the MI difference sequence is given by
\begin{align}
I'_{m,n} & = I(X^n ; Y_{m+1} \gmid Y^m, A^{m+1}). \label{eq:Ip_alt}
\end{align}
By the non-negativity of mutual information, this establishes that the MI sequence is non-decreasing in $m$. Alternatively, by the data-processing inequality for MMSE \cite[Proposition 4]{rioul:2011}, we also see that $M_{m,n}$ is non-increasing in $m$, and also
\begin{align}
M_{m,n} \le   \var(X). \notag
\end{align}

The next result shows that the second order MI difference can also be expressed in terms of the mutual information. The proof is given in Appendix~\ref{proof:lem:Ip_and_Ipp}.

\begin{lemma}\label{lem:Ip_and_Ipp} 
Under Assumption 1, the second order MI difference sequence satisfies 
\begin{align}
I''_{m,n} & = - I(Y_{m+1}  ; Y_{m+2} \mid Y^m, A^{m+2}). \label{eq:Ipp_alt}
\end{align}
\end{lemma}

One consequence of Lemma~\ref{lem:Ip_and_Ipp}  is that the first order MI difference sequence is non-increasing in $m$, and thus
\begin{align}
I'_{m,n} \le I'_{1,n} = I_{1,n}. \label{eq:Ip_bound}
\end{align}
This inequality plays an important role later on in our proof, when we show that certain terms of interest are bounded by the magnitude of the second order MI difference.

The next result provides non-asymptotic bounds in terms of the single-letter MI and MMSE functions corresponding to the signal distribution $P_X$. The proof is given in Appendix~\ref{proof:lem:Im_bounds}

\begin{lemma}\label{lem:Im_bounds}  Under Assumptions 1 and 2, the MI and MMSE sequences satisfy
\begin{align}
& \sum_{k=1}^{\min(n,m)} \hspace{-.2cm} \ex{ I_X\left( \tfrac{1}{n} \chi^2_{m-k+1} \right)} \le I_{m,n}  \le n \, \ex{ I_X\left( \tfrac{1}{n} \chi^2_{m} \right)}  \label{eq:Im_bounds} \\
& \ex{ \mmse_X\left( \tfrac{1}{n} \chi^2_m \right)}  \le  M_{m,n} \le \ex{ \mmse_X\left( \tfrac{1}{n}  \chi^2_{m - n + 1} \right) }, \label{eq:Mm_bounds}
\end{align}
where $\chi_k^2$ denotes a chi-squared random variable with $k$ degrees of freedom and the upper bound on $M_{m,n}$ is valid for all $m \ge n$. 
\end{lemma}

\begin{remark}
The proof of Lemma~\ref{lem:Im_bounds} does not require the assumption that the signal distribution has bounded fourth moment. In fact, \eqref{eq:Im_bounds} holds for any signal distribution with finite mutual information and \eqref{eq:Mm_bounds}  holds for any  signal distribution. 
\end{remark} 

\begin{remark}
The upper bound in \eqref{eq:Im_bounds} and lower bound in \eqref{eq:Mm_bounds} are not new and are special cases of results given in \cite{tulino:2013}.
\end{remark} 

Combining Lemma~\ref{lem:Im_bounds} with Inequalities  \eqref{eq:IX_Gbound} and \eqref{eq:mmseX_Gbound}, leads to upper bounds on the MI and MMSE that depend only on the variance of the signal distribution. Alternatively, combining Lemma~\ref{lem:Im_bounds} with the smoothness of the single-letter functions given in \eqref{eq:IX_smooth} and \eqref{eq:mmseX_smooth_b} leads to the following characterization, which is tight whenever $m$ is much larger than $n$.  The proof is given in Appendix~\ref{proof:lem:Im_bounds_gap}

\begin{lemma}\label{lem:Im_bounds_gap}  Under Assumptions 1 and 2, the MI and MMSE sequences satisfy, for all $m \ge n + 2$, 
\begin{align}
\left| \tfrac{1}{n} I_{m,n} -  I_X(\tfrac{m}{n}) \right| & \le  \tfrac{1}{2} \left[  \tfrac{n + 1}{m - n -  1} +   \sqrt{ \tfrac{2}{m\!-\! 2 } }  \right] \label{eq:Im_bounds_gap}\\
\left| M_{m,n} - \mmse_X\left( \tfrac{m}{n} \right) \right| & \le  \tfrac{ 12\, n}{m}  \left[  \tfrac{ n+1}{m - n - 1} +  \sqrt{ \tfrac{2}{m -2}} \right].
\label{eq:Mm_bounds_gap}
\end{align}
\end{lemma}

For any fixed $n$, the right-hand sides of \eqref{eq:Im_bounds_gap} and \eqref{eq:Mm_bounds_gap} converge to zero as $m$ increases to infinity. Consequently,  the large $m$ behavior of the MI sequence is given by
\begin{align}
\lim_{m \to \infty} I_{m,n} & = \begin{cases}
H(X^n) ,  & \text{if $P_X$ has finite entropy}\\
+ \infty, &\text{otherwise}.
\end{cases}  \notag
\end{align}

\subsection{Properties of Replica Prediction} 

Using the I-MMSE relationship, the partial derivative of $R(\delta, z)$ with respect to its second argument is given by
\begin{align}
R_z(\delta, z) & =  \frac{\delta}{ 2 (1 + z)^2 } \left[  z - \mmse_X\left(  \frac{ \delta}{ 1 + z} \right)  \right]. \label{eq:R_z} 
\end{align}
From this expression, we see that the $R_z(\delta, z) = 0$ is equivalent to the fixed-point condition
\begin{align}
z = \mmse\left( \frac{ \delta}{1 + z} \right) . \notag
\end{align}

Furthermore, since $\IR(\delta)$ is concave, it is differentiable almost everywhere. For all $\delta$ where $\IR'(\delta)$ exists, it follows from the envelope theorem~\cite{milgrom:2002} that $\IR'(\delta)  = R_{\delta}\left(\delta, \MR(\delta)\right)$, where $R_\delta(z,\delta)$ is the partial derivative of $R(\delta, z)$ with respect to its first argument. Direct computation yields 
\begin{align}
R_\delta(\delta, z) & = \frac{1}{2} \log(1 \!+\! z) +  \frac{1}{ 2 (1 \!+\! z) } \left[  \mmse_X\left(  \frac{ \delta}{ 1 \!+\! z} \right) - z \right]. \notag
\end{align}
Finally, noting that the second term on the right-hand side is equal to zero whenever $z = \MR(\delta)$, leads to 
\begin{align}
\IR'(\delta) & = \frac{1}{2} \log\left(1 + \MR(\delta) \right). \notag
\end{align}

The proof of the next result is given in Appendix~\ref{proof:lem:IR_boundary}. 
\begin{lemma} \label{lem:IR_boundary} The Replica-MI and Replica-MMSE functions satisfy, for all $\delta \ge 1$, 
\begin{align}
I_X(\delta-1) &\le   \IR(\delta)  \le I_X(\delta), \label{eq:IRS_bounds}\\
\mmse_X(\delta) &\le   \MR(\delta)  \le \mmse_X(\delta-1).  \label{eq:MRS_bounds}
\end{align}
\end{lemma}

It is interesting to note the parallels between the bounds on the MI and MMSE sequences in Lemma~\ref{lem:Im_bounds}  and the bounds on the replica functions in Lemma~\ref{lem:IR_boundary}. Combining Lemma~\ref{lem:IR_boundary} with the smoothness of the single-letter functions given in \eqref{eq:IX_smooth} and \eqref{eq:mmseX_smooth_b} leads to 
\begin{align}
\left| \IR(\delta) -  I_X(\delta) \right| &\le \frac{ 1} {2(\delta -1)}   \notag\\
\left| \MR(\delta) -  \mmse_X(\delta) \right| &\le  \frac{1}{ \delta( \delta  -1)} .\notag
\end{align}

\subsection{Proof of Theorem~\ref{thm:I_m_boundary}}\label{proof:thm:I_m_boundary}

This proof follows from combining Lemmas~\ref{lem:Im_bounds_gap} and \ref{lem:IR_boundary}. Fix any $n \in \reals_+$ and $\delta > 4$ and let $m = \lfloor \delta n \rfloor$ and  $\lambda = m + 1 - \delta n$. The MI function obeys the upper bound 
\begin{align}
\cI_n(\delta) &=  \frac{1}{n} \left[ \lambda  I_{m,n}  + (1-\lambda) I_{m+1,n} \right] \notag\\
& \overset{(a)}{\le} \lambda \ex{  I_X\left(\tfrac{1}{n} \chi^2_m \right)}  + (1-\lambda) \ex{  I_X\left(\tfrac{1}{n} \chi^2_{m+1} \right)} \notag\\
& \overset{(b)}{\le}
I_X(\delta), \label{eq:I_m_boundary_b}
\end{align}
where: (a) follows from \eqref{eq:Im_bounds}; and (b) follows from Jensen's inequality and the concavity of $I_X(s)$. 
The MI function also obeys the lower bound
\begin{align}
\cI_n(\delta) &\overset{(a)}{\ge}
\frac{1}{n}  I_{m,n}  \notag\\
& \overset{(b)}{\ge } I_X\left(\delta \right)  - \tfrac{1}{2 (\delta  - 1)  } -  \tfrac{1}{2} \left[  \tfrac{n+1}{(\delta -1) n  - 2}  -  \sqrt{ \tfrac{2}{\delta-3} } \right] ,\label{eq:I_m_boundary_c}
\end{align}
where (a) follows from the fact that $I_{m,n} $ is non-decreasing in $m$ and (b) follows from \eqref{eq:Im_bounds_gap}, \eqref{eq:IX_smooth}, and the fact that  $m \ge \delta n - 1 \ge \delta - 1$. Finally, we have
\begin{align}
\left| \cI_n(\delta) - \IR(\delta) \right| &\overset{(a)}{\le} \left| \cI_n(\delta) - I_X(\delta) \right|   + \left| I_X(\delta) - \IR(\delta) \right| \notag\\
& \overset{(b)}{\le} \tfrac{1}{(\delta -1) }   + \tfrac{1}{2}  \tfrac{n+1}{(\delta -1) n  - 2}  + \tfrac{1}{2} \sqrt{ \tfrac{2}{\delta-3} } \notag\\
&\overset{(c)}{\le}  \left( 4 + \sqrt{2} \right) \delta^{-\frac{1}{2}},  \notag
\end{align}
where: (a) follows from the triangle inequality;  (b) follows from \eqref{eq:IRS_bounds},  \eqref{eq:I_m_boundary_b}, and \eqref{eq:I_m_boundary_c}; and (c) follows from the assumption $\delta \ge 4$. This completes the proof of Theorem~\ref{thm:I_m_boundary}.

\begin{figure*}[!ht]
\centering

\tikzstyle{gblock} = [rectangle, minimum width=2.75cm, minimum height=1cm, text centered, text width=2.75cm,  draw=black, fill=gray!05]
\tikzstyle{arrow} = [thick,->,>=stealth]

\begin{tikzpicture}[node distance=1.6cm]
\footnotesize
\node (dist_ident) [gblock]{Posterior distribution identities \\ (Lemma~\ref{lem:post_dist_to_id})};
\node (var_imi) [gblock, below of  = dist_ident, yshift = -0.8cm] {Concentration of MI density\\ (Lemmas~\ref{lem:Im_var} and \ref{lem:IMI_var})};

\node (post_dist) [gblock, right of  = dist_ident,  xshift=1. 9cm, yshift =.8cm] {Weak decoupling \\ (Lemma~\ref{lem:SE_bound_1})};
\node (smooth) [gblock, below of = post_dist] {Smoothness of posterior variance\\ (Lemma~\ref{lem:post_var_smoothness})};
\node (pmid) [gblock, below of = smooth] {Smoothness of posterior MI difference \\ (Lemma~\ref{lem:PMID_var})};

\node (cclt) [gblock, right of = smooth, xshift=1.9 cm, yshift  = 0cm] {Posterior Gaussianness of new measurements \\ (Lemma~\ref{lem:DeltaP_bound})};
\node (cclt0) [gblock, above of = cclt] {Conditional CLT \cite{reeves:2016b} \\ (Lemma~\ref{lem:cclt})};

\node (dev_var) [gblock, below of= cclt] {Concentration of posterior variance\\ (Lemma~\ref{lem:post_var_dev_bound})};

\node (cclt2) [gblock, right of = cclt, xshift=1.9cm, yshift  = -0cm] {Gaussianness of new measurements \\ (Lemma~\ref{lem:Delta_bound})};
\node (fixed_point) [gblock, below of = cclt2] {MMSE fixed-point constraint \\ (Lemmas~\ref{lem:M_to_M_aug} and \ref{lem:MMSE_aug_bound})};

\node (derivative) [gblock, right of  = cclt2, xshift = 1.9 cm, yshift = .8cm] {Asymptotic derivative constraint\\ (Theorem~\ref{thm:I_MMSE_relationship})};
\node (mmse) [gblock, below of = derivative] {Asymptotic  fixed-point constraint\\ (Theorem~\ref{thm:MMSE_fixed_point})};

%


\draw [arrow] (dist_ident) -- (post_dist);
\draw [arrow] (dist_ident) -- (smooth);
\draw [arrow] (post_dist) -- (cclt);
\draw [arrow] (smooth) -- (dev_var);
\draw [arrow] (var_imi) -- (pmid);
\draw [arrow] (pmid) -- (dev_var);
\draw [arrow] (cclt) -- (dev_var);
\draw [arrow] (cclt0) -- (cclt);
\draw [arrow] (cclt) -- (cclt2);
\draw [arrow] (dev_var) -- (cclt2);
\draw [arrow] (cclt2) -- (derivative);
\draw [arrow] (cclt2) -- (mmse);
\draw [arrow] (cclt2) -- (mmse);
\draw [arrow] (fixed_point) -- (mmse);
\end{tikzpicture}
\caption{Outline of the main steps in the proofs of Theorem~\ref{thm:I_MMSE_relationship}  and Theorem~\ref{thm:MMSE_fixed_point}  \label{fig:proof_outline_a}. } 
\end{figure*}
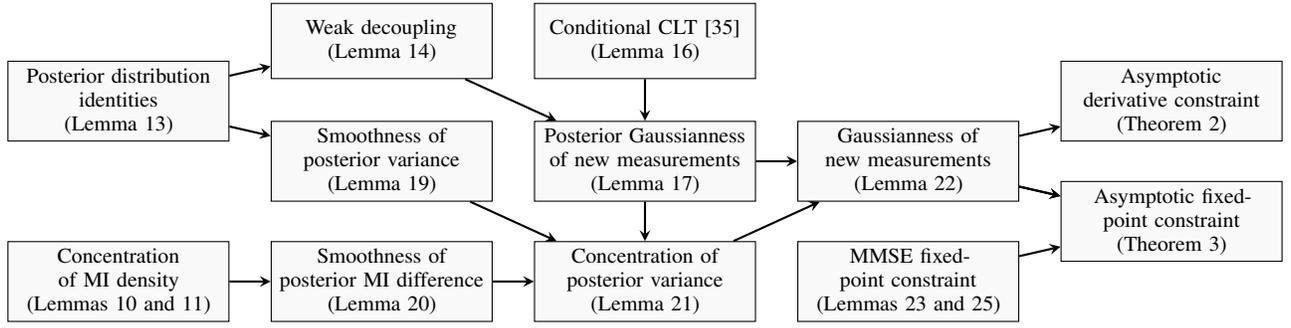

\subsection{Concentration of MI Density}\label{sec:IMI_concentration} 

In order to establish the MI and MMSE relationships used our proof, we need to show that certain functions of the random tuple $(X^n, W^n, A^n)$ concentrate about their expectations. 

Our first result bounds the variation in the mutual information corresponding to the measurement matrix. Let $I_{m,n}(A^m)$ denote the MI sequence as a function of the random matrix $A^m$. The following result follows from the Gaussian Poincar\'e inequality and the multivariate I-MMSE relationship. The proof is given in Appendix~\ref{proof:lem:Im_var}. 

\begin{lemma}\label{lem:Im_var} 
Under Assumptions 1 and 2, the variance of the MI with respect to the measurement matrix satisfies
\begin{align}
\var\left(I_{m,n}(A^m)\right) & \le \min\left\{ \left(\var(X)\right)^2\, m , \frac{n^2}{(m-n-1)_+} \right\}.  \notag
\end{align}
\end{lemma}

An important consequence of Lemma~\ref{lem:Im_var} is  that the variance of the normalized mutual information $\frac{1}{n} I_{m,n}(A^m)$ converges to zero as $\max(m,n)$ increases to infinity. 

Next, we focus on the concentration of the mutual information density, which is a random variable whose expectation is equal to the mutual information.

\begin{definition}\label{def:intant_MI} Given a distribution $P_{X,Y,Z}$, the \textit{conditional mutual information density} between $X$ and $Y$ given $Z$ is defined as 
\begin{align}
\imath(X; Y\gmid Z) & \triangleq\log\left(\frac{  \dd P_{X, Y \mid Z} (X, Y \gmid Z)}{ \dd  \left( P_{X \mid Z}(X \gmid Z)  \times   P_{Y \mid Z}(Y \gmid Z) \right)} \right),  \notag
\end{align}
where $(X,Y,Z) \sim P_{X,Y,Z}$.  This is well-defined because a joint distribution is absolutely continuous with respect to the product of its marginals.
\end{definition}

The mutual information density satisfies many of the same properties as mutual information, such as the chain rule and invariance to one-to-one transformations; see \cite[Chapter~5.5]{gray:2013}. For this compressed sensing problem, the mutual information density can be expressed in terms of the density functions $f_{Y^m|X^n, A^m}$ and $f_{Y^m|A^m}$, which are guaranteed to exist because of the additive Gaussian noise: 
\begin{align}
\imath(X^n ; Y^m \gmid  A^m)  = \log\left( \frac{f_{Y^m \mid X^n, A^m}( Y^m  \gmid  X^n,  A^m )}{f_{Y^m \mid A^m}(Y^m \gmid  A^m )}   \right).\notag
\end{align}

The  next result bounds the variance of the mutual information density in terms of the fourth moment of the signal distribution and the problem dimensions. The proof is given in Appendix~\ref{proof:lem:IMI_var}.

\begin{lemma}\label{lem:IMI_var} Under Assumptions 1 and 2, the variance of the MI density satisfies
\begin{align}
 \var\left(\imath(X^n ; Y^m \mid A^m) \right)  \le C_B \cdot \left( 1+ \frac{m}{n} \right)^2 n.\notag
\end{align}
\end{lemma}

\section{Proof of Theorem~\ref{thm:I_MMSE_relationship}}\label{proof:thm:I_MMSE_relationship}

This section describes the proof of Theorem~\ref{thm:I_MMSE_relationship}. An outline of the dependences  between various steps is provided in Figure~\ref{fig:proof_outline_a}.

\subsection{Further Definitions}

The conditional distribution induced by the data $(Y^m, A^m)$  plays an important role and is referred to throughout as the \textit{posterior distribution}.  The optimal signal estimate with respect to squared error is given by the mean of the posterior distribution, and the squared error associated with this estimate is denoted by
\begin{align}
\SE_{m,n} & =  \frac{1}{n} \left\| X^n - \ex{ X^n \gmid Y^m, A^m} \right\|^2. \notag
\end{align}
The conditional expectation of the squared error with respect to the posterior distribution is referred to as the \textit{posterior variance} and is denoted by
\begin{align}
\PSE_{m, n} & = \ex{ \cE_{m,n} \gmid Y^m, A^m } .  \notag
\end{align}
Both $\SE_{m,n}$ and $\PSE_{m,n}$ are random variables. By construction, their expectations are equal to the MMSE, that is 
\begin{align}
M_{m,n} = \ex{ \PSE_{m,n}} = \ex{\SE_{m,n}}.\notag
\end{align}

Next, recall that the MI difference sequence can be expressed in terms of the mutual information between the signal and a new measurement: 
\begin{align}
I'_{m,n} & = I(X^n; Y_{m+1} \gmid Y^m, A^{m+1}).\notag
\end{align}
The \textit{MI difference density} is defined to be the random variable 
\begin{align}
\IMID_{m,n} & = \imath(X^n; Y_{m+1} \gmid Y^m, A^{m+1}),\notag
\end{align}
where the mutual information density is defined in Definition~\ref{def:intant_MI}. The conditional expectation of the MI difference density with respect to the posterior distribution is referred to as the \textit{posterior MI difference} and is denoted by
\begin{align}
\PMID_{m,n} & = \ex{ \imath(X^n; Y_{m+1} \gmid Y^m, A^{m+1}) \; \middle  | \; Y^m, A^m}.\notag
\end{align}
Both $\IMID_{m,n} $ and $\PMID_{m,n}$ are random variables.  By construction, their expectations are equal to the MI difference, that is
\begin{align}
\ex{\IMID_{m,n}}= \ex{\PMID_{m,n}} = I'_{m,n}. \notag
\end{align}
A summary of this notation is provided in Table~\ref{tab:notation}. 

The next result shows that the moments of the square error and MI difference density can be bounded in terms of the fourth moment of the signal distribution. The proof is given in Appendix~\ref{proof:lem:moment_bounds}.
\begin{lemma}\label{lem:moment_bounds}  Under Assumptions 1 and 2,
\begin{align}
\ex{\left| \SE_{m,n} \right|^2} & \le C \cdot B \label{eq:SE_moment_2}\\ 
\ex{\left| \IMID_{m,n} \right|^2} & \le C\cdot (1 + B)  \label{eq:IMID_moment_2}\\ 
\ex{\left| Y_{m+1} \right|^4 } & \le C \cdot (1 +B) \label{eq:Y_moment_4} \\
\ex{\left|Y_{m+1} - \widehat{Y}_{m+1}  \right|^4} & \le  C \cdot (1 + B) , \label{eq:Ybar_moment_4} 
\end{align}
where $\widehat{Y}_{m+1} =  \ex{ Y_{m+1} \gmid Y^m, A^m}$.
\end{lemma}

\begin{table*}[htbp]
   \centering
   \caption{\label{tab:notation} Summary of notation used in the proofs of Theorem~\ref{thm:I_MMSE_relationship}  and Theorem~\ref{thm:MMSE_fixed_point}  } 
   \begin{tabular}{@{} clll @{}} 
      \toprule
& Random Variable  & Posterior Expectation & Expectation \\
 \midrule
Squared Error &  $ \SE_{m} = \frac{1}{n}  \| X^n-\ex{X^n \gmid Y^m, A^m }\|^2$ & $ \PSE_{m} = \ex{\SE_{m} \gmid Y^m, A^m}$ & $M_{m} = \ex{\SE_m}$   \\[3pt]
MI Difference & $\IMID_m =\imath\left(X^n ; Y_{m+1} \gmid Y^{m}, A^{m+1}\right)   $ & $\PMID_{m} = \ex{ \IMID_m \gmid Y^m, A^m}$ & $I'_{m} = \ex{ \IMID_m }$ \\
      \bottomrule
   \end{tabular}
   \label{tab:booktabs}
\end{table*}

\subsection{Weak Decoupling}\label{sec:posterior_distribution} 

The posterior distribution of the signal cannot, in general, be expressed as the product of its marginals since the measurements introduce dependence between the signal entries. Nevertheless, it has been observed that in some cases, the posterior distribution satisfies a  \textit{decoupling principle} \cite{guo:2005}, in which the posterior distribution on a small subset of the signal entries is well approximated by the product of the marginals on that subset. One way to express this decoupling is say that, for any fixed index set $\{i_1, \cdots i_L\}$ with $L \ll n$, the random posterior distribution satisfies
\begin{align}
  P_{X_{i_1}, \cdots, X_{i_L} \mid Y^m, A^m}\approx  \prod_{\ell=1}^L  P_{X_{i_\ell} \mid Y^m, A^m}, \notag
\end{align}
with high probability with respect to the data $(Y^m, A^m)$.

One of the main ideas  in our proof is to use decoupling to show that the MI and MMSE sequences satisfy certain relationships. For the purposes of our proof, it is sufficient to work with a weaker notation of decoupling that depends only on the statistics of the pairwise  marginals of the posterior distribution.

\begin{definition}
The posterior signal distribution $P_{X^n \mid Y^m, A^m}$ satisfies \textit{weak decoupling} if 
\begin{gather}
  \ex{ \left | \SE_{m,n} - \PSE_{m,n}  \right |}  \to 0 \notag\\
 \frac{1}{n} \ex{ \left \| \cov(X^n \gmid Y^{m}, A^{m} ) \right\|_F } \to 0. \notag
\end{gather}
as $m$ and $n$ increase to infinity. 
\end{definition}

The first condition in the definition of weak decoupling says that the magnitude of the squared error must concentrate about its conditional expectation. The second condition says that the average correlation between the signal entries under the posterior distribution is converging to zero. Note that both of these conditions are satisfied a priori in the case of $m=0$ measurements, because that  prior signal distribution is i.i.d.\ with finite fourth moments.

The next result provides several key identities, which show that certain properties of the posterior distribution can be expressed in terms of the second order statistics of new measurements. This result does not require Assumption~2, and thus holds generally for any prior distribution  on $X^n$. The proof is given in Appendix~\ref{proof:lem:post_dist_to_id}.

\begin{lemma}\label{lem:post_dist_to_id} 
Under Assumption 1, the following identities hold for all integers  $m < i < j$:
\begin{enumerate}[(i)] \item  The posterior variance satisfies  
\begin{align}
\PSE_{m,n}  = \bEx_{A_{i}}\left[  \var\left(Y_{i}  \gmid Y^m, A^{m} , A_i    \right)  \right]  - 1. \label{eq:id_var} 
\end{align}

\item  The posterior covariance matrix satisfies 
\begin{multline}
  \frac{1}{n^2}  \left \| \cov(X^n \gmid Y^{m}, A^{m} ) \right\|^2_F\\
 = \bEx_{A_{i}, A_{j}}\left[  \left| \cov\left(Y_{i}, Y_{j} \gmid Y^m\!, A^{m}\!, A_i, A_j   \right) \right|^2  \right] .\label{eq:id_cov} 
\end{multline}

\item The conditional variance of $\sqrt{1 + \SE_{m,n}}$ satisfies 
\begin{multline}
\var\left( \sqrt{1 + \SE_{m,n}}  \;  \middle | \;  Y^m ,A^m \right) \\
= \frac{\pi}{2}  \cov\left( \big| Y_{i} - \widehat{Y}_{i}  \big | , \big | Y_{j} - \widehat{Y}_{j} \big| \,  \middle| \,  Y^m\!, A^{m}  \right), \label{eq:id_post_var} 
\end{multline}
where $\widehat{Y}_{i}    =  \ex{ Y_{i} \mid Y^m, A^{m}, A_i} $.

\end{enumerate}
\end{lemma}

Identity  \eqref{eq:id_cov} relates the correlation of the signal entries under the posterior distribution to the correlation of new measurements. Identity \eqref{eq:id_post_var} relates the deviation of the squared error under the posterior distribution to the correlation of between new measurements.   Combining these identities with the bounds on the relationship between covariance and mutual information given in Appendix~\ref{sec:cov_to_MI}  leads to the following result. The proof is given in Appendix~\ref{proof:lem:SE_bound_1}.

\begin{lemma}\label{lem:SE_bound_1}  Under Assumptions 1 and 2, the posterior variance and the posterior covariance matrix satisfy
\begin{align}
\ex{ \left| \SE_{m,n} - \PSE_{m,n} \right|}  &  \le  C_B \cdot  \left| I''_{m,n} \right|^{\frac{1}{4}} \label{eq:weak_dec_1}\\
\frac{1}{n} \ex{ \left \| \cov(X^n \mid Y^{m}, A^{m} ) \right\|_F }  &  \le  C_B \cdot  \left| I''_{m,n} \right|^{\frac{1}{4}}. \label{eq:weak_dec_2}
\end{align}
\end{lemma}

\subsection{Gaussiannness of New Measurements}\label{sec:Gaussiannness}

The \textit{centered measurement} $\Ybar_{m+1}$ is defined to be the difference between a new measurement and its conditional expectation given the previous data:
\begin{align}
\Ybar_{m+1} & \triangleq Y_{m+1} - \ex{Y_{m+1} \mid Y^m, A^{m+1}}. \notag 
\end{align}
Conditioned on the data $(Y^m, A^{m+1})$, the centered measurement provides the same information as $Y_{m+1}$, and thus the posterior MI difference  and the MI difference can be expressed equivalently as 
\begin{align}
\PMID_m &= \ex{  \imath(X^n; \Ybar_{m+1} \mid Y^m, A^{m+1} ) \; \middle |  \; Y^m ,A^m } \notag\\
I'_m & = I(X^n; \Ybar_{m+1} \mid Y^m, A^{m+1} ).  \notag
\end{align}

Furthermore, by the linearity of expectation, the centered measurement can be viewed as a noisy linear projection of the  signal error:
\begin{align}
\Ybar_{m+1} & = \langle A_{m+1}, \bar{X}^n  \rangle  + W_{m+1}, \label{eq:Ybar_alt}
\end{align}
where $\bar{X}^n = X^n -  \ex{X^n \gmid Y^m, A^m}$. Since the measurement vector $A_{m+1}$ and noise term $W_{m+1}$ are independent of everything else, the variance of the centered measurement can be related directly to the posterior variance $V_m$ and the MMSE $M_m$ via the following identities: 
\begin{align}
\var(\Ybar_{m+1} \gmid Y^m, A^m) &= 1 +\PSE_{m,n} \label{eq:Ybar_var_cond} \\
\var(\Ybar_{m+1} ) &= 1 +M_{m,n}.\label{eq:Ybar_var} 
\end{align}
Identity \eqref{eq:Ybar_var_cond} follows immediately from Lemma~\ref{lem:post_dist_to_id}.  Identity \eqref{eq:Ybar_var} follows from the fact that the centered measurement has zero mean, by construction, and thus its variance is equal to the expectation of  \eqref{eq:Ybar_var_cond}.

At this point, the key question for our analysis is the extent to which the conditional distribution of the centered measurement can be approximated by a zero-mean Gaussian distribution. We focus on two different measures of non-Gaussianness.  The first measure, which is referred to as the \textit{posterior non-Gaussianness}, is defined by the random variable
\begin{align}
\Delta^P_{m,n}  \triangleq  \bEx_{A_{m+1}} \left[  D_\mathrm{KL}\left( P_{\Ybar_{m+1} \mid Y^m, A^{m+1}} \,  \middle \|\, \normal(0, 1  + V_m ) \right)\right].  \notag
\end{align}
This is the Kullback--Leibler divergence with respect to the Gaussian distribution whose variance is matched to the conditional variance of $\Ybar_{m+1}$ given the data $(Y^m, A^m)$.

The second measure, which is referred to simply as the \textit{non-Gaussianness},  is defined by
\begin{align}
\Delta_{m,n} & \triangleq  \ex{ D_\mathrm{KL}\left( P_{\Ybar_{m+1} \mid Y^m, A^{m+1}} \,  \middle \|\, \normal(0, 1 + M_m ) \right)}.\notag
\end{align}
Here, the expectation is taken with respect to the tuple $(Y^m, A^{m+1})$ and the comparison is with respect to the Gaussian distribution whose variance is matched to the marginal variance of $\Ybar_{m+1}$.


The connection between the non-Gaussianness of the centered measurement and the relationship between the mutual information and MMSE sequences is given by the following result. The proof  is given in Appendix~\ref{proof:lem:Delta_alt}.
 
\begin{lemma}\label{lem:Delta_alt}  Under Assumption 1, the posterior non-Gaussianness  and the non-Gaussianness satisfy the following identities: 
\begin{align}
\Delta^P_{m,n} & = \frac{1}{2} \log(1 + \PSE_{m,n})  -  \PMID_{m,n} \label{eq:DeltaP_alt}\\
\Delta_{m,n} & = \frac{1}{2} \log(1 + M_{m,n})  -  I'_{m,n} \label{eq:Delta_alt} .
\end{align}
\end{lemma} 

Identity~\eqref{eq:Delta_alt} shows the integral relationship between mutual information and MMSE in Theorem~\ref{thm:I_MMSE_relationship} can be stated equivalently in terms of the non-Gaussianness of the centered measurements.  Furthermore, by combining \eqref{eq:Delta_alt} with \eqref{eq:DeltaP_alt}, we see that the non-Gaussianness can be related to the expected posterior non-Gaussianness using the following decomposition:
\begin{align}
\Delta_{m,n} &=  \ex{ \Delta^P_{m,n}}   + \frac{1}{2}  \ex{\log\left( \frac{  1  +  M_{m,n}}{ 1 + V_{m,n}} \right)}.\label{eq:Delta_decomp} 
\end{align}

The rest of this subsection is focused on bounding the expected posterior non-Gaussianness. The second term on the right-hand side of \eqref{eq:Delta_decomp} corresponds to the deviation of the posterior variance and is considered in the next subsection. 

The key step in bounding the posterior non-Gaussianness is provided by the following result, which bounds the expected Kullback--Leibler divergence between the conditional distribution of a random projection and a Gaussian approximation~\cite{reeves:2016b}. 

\begin{lemma}[{\!\cite{reeves:2016b}}]\label{lem:cclt} 
Let $U$ be an $n$-dimensional random vector with mean zero and $\ex{ \|U\|^4 } < \infty$, and let $Y = \langle A, U \rangle + W$, 
where $A \sim \normal(0, \frac{1}{n} I_n)$ and $W \sim \normal(0,1)$ are independent. Then, the expected KL divergence between $P_{Y|A}$ and the Gaussian distribution with the same mean and variance as $P_Y$ satisfies
\begin{align}
\MoveEqLeft 
\ex{ \DKL{P_{Y\mid A} }{\normal(0, \var(Y) }} \notag\\
&\le   \frac{1}{2} \ex{ \left|  \tfrac{1}{n}  \|U\| -\tfrac{1}{n} \ex{ \|U\|^2} \right|  } \notag\\ 
& \quad C \cdot \left|  \tfrac{1 }{n}   \|\! \cov(U  )\|_F \left( 1 + \tfrac{1}{n}\sqrt{\ex{ \|U\|^4 }} \right) \right|^\frac{2}{5} .\notag
\end{align}
 \end{lemma}

Combining Lemma~\ref{lem:cclt} with Lemma~\ref{lem:SE_bound_1} leads to the following result, which bounds the expected posterior non-Gaussianness in terms of the second order MI difference. The proof is given in Appendix~\ref{proof:lem:DeltaP_bound}.

 \begin{lemma}\label{lem:DeltaP_bound} Under Assumptions 1 and 2, the expected posterior non-Gaussianness satisfies
  \begin{align}
\ex{ \Delta^P_{m,n}} &\le  C_B\cdot  \left| I''_{m,n} \right|^{\frac{1}{10}} .\notag
\end{align}
\end{lemma}

\subsection{Concentration of Posterior Variance} 

We now turn our attention to the second term on the right-hand side of \eqref{eq:Delta_decomp}. By the concavity of the logarithm, this term is nonnegative and measures the deviation of the posterior variance  about its expectation.

We begin with the following result, which  provides useful bounds on the deviation of the posterior variance. The proof is given in Appendix~\ref{proof:lem:Vdev_to_logVdev}. 

\begin{lemma}\label{lem:Vdev_to_logVdev} Under Assumption 2, the posterior variance satisfies the following inequalities:
\begin{align}
  \ex{\log\left( \frac{  1  +  M_{m,n}}{ 1 + V_{m,n}} \right)} & \le    \ex{\left| V_{m,n} - M_{m,n} \right|} \notag \\
 & \le  C_B \cdot  \sqrt{ \inf_{t \in \reals} \ex{\tfrac{1}{2} \log(1 + V_{m,n}) - t } }  \label{eq:Vdev_to_logVdev} .
\end{align}
\end{lemma}

The next step is to bound the right-hand side of \eqref{eq:Vdev_to_logVdev}. Observe that by Lemma~\ref{lem:Delta_alt}, the term $\frac{1}{2} \log(1 + V_{m,n})$ can be expressed in terms of the posterior non-Gaussianness and the posterior MI difference. Accordingly, the main idea behind our approach is to show that the deviation of this term can be upper bounded in terms of deviation of the posterior MI difference. 

Rather than working with the sequences $V_{m,n}$ and $\PMID_{m,n}$ directly, however, we bound the averages of these terms corresponding to a sequence of $\ell$ measurements, where $\ell$ is an integer that that is chosen at the end of the proof to yield the tightest bounds. The main reason that we  introduce this averaging is so that we can take advantage of the bound on the variance of the mutual information density given in Lemma~\ref{lem:IMI_var}.

%


The key technical results that we need are given below.  Their proof are given in Appendices~\ref{proof:lem:post_var_smoothness} and \ref{proof:lem:PMID_var}.

\begin{lemma}\label{lem:post_var_smoothness}  Under Assumptions 1 and 2, the posterior variance satisfies 
\begin{align}
\frac{1}{\ell} \sum_{k=m}^{m+\ell-1} \ex{ \left| \PSE_{m} - \PSE_{k} \right|} & \le C_B \cdot \left| I'_{m,n} - I'_{m+\ell -1,n} \right|^{\frac{1}{2}}, \notag
\end{align}
for all $(\ell, m,n)\in \integers^3$.
\end{lemma}


\begin{lemma}\label{lem:PMID_var} Under Assumptions 1 and 2, the posterior MI difference satisfies
\begin{align}
  \inf_{t \in \reals} \ex{  \left|  \frac{1}{\ell} \! \sum_{k =m}^{m+ \ell-1}\!\! \PMID_k - t \right|} &  \le C_B \cdot \left[ \left( 1 + \frac{m}{n}  \right)  \frac{\sqrt{n}}{\ell}   + \frac{1}{\sqrt{n}}  \right], \notag
\end{align}
for all $(\ell,m,n) \in \integers^3$. 
\end{lemma}

Finally, combining Lemmas~\ref{lem:DeltaP_bound}, \ref{lem:post_var_smoothness}, and \ref{lem:PMID_var} leads to the following result, which bounds the deviation of the posterior variance in terms of the MI difference difference sequence. The proof is given in Appendix~\ref{proof:lem:post_var_dev_bound}.

 \begin{lemma}\label{lem:post_var_dev_bound} Under Assumptions 1 and 2, the posterior variance satisfies, 
\begin{multline}
\ex{ |V_{m,n} - M_{m,n}|}  
\le  C_B \cdot    \Big[   \left| I'_{m,n}  -  I'_{m+\ell-1,n} \right|^\frac{1}{4}   + \ell^{-\frac{1}{20}} \\ + \left( 1 +  \tfrac{m}{n}  \right)^\frac{1}{2}    n^{\frac{1}{4}} \ell^{-\frac{1}{2}} + n^{-\frac{1}{4}}\Big].  \notag
\end{multline}
for all $(\ell,m,n) \in \integers^3$. 
\end{lemma}

\subsection{Final Steps in Proof of Theorem~\ref{thm:I_MMSE_relationship}}

The following result is a straightforward consequence of Identity~\eqref{eq:Delta_decomp}  and  Lemmas~\ref{lem:DeltaP_bound} and \ref{lem:post_var_dev_bound}. The proof is given in Appendix~\ref{proof:lem:post_var_dev_bound}.

\begin{lemma}\label{lem:Delta_bound} 
Under Assumptions 1 and 2, the non-Gaussianness of new measurements satisfies the upper bound
\begin{multline}
\Delta_{m,n}   \le  C_B \cdot   \Big[  \left|I''_{m,n} \right|^{\frac{1}{10}} +  \left| I'_{m,n} -  I'_{m+\ell_n,  n}   \right|^\frac{1}{4} \\     +\left(1 + \tfrac{m}{n} \right)^\frac{1}{2}  n^{-\frac{1}{24}}   \Big], \notag
\end{multline}
where $\ell_n = \lceil n^\frac{5}{6} \rceil$.
\end{lemma}

We now show how the proof of Theorem~\ref{thm:I_MMSE_relationship} follows as a consequence of Identity~\eqref{eq:Delta_alt} and Lemma~\ref{lem:Delta_bound}. Fix any $n\in \integers$ and $\delta \in \reals_+$ and let $m = \lceil \delta n \rceil$ and $\ell = \lceil n^\frac{5}{6} \rceil$. Then, we can write
\begin{align}
\MoveEqLeft
 \int_{0}^{\delta} \left| \cI_{n}'(\gamma)   - \frac{1}{2} \log\left(1 + \cM_{n}(\gamma)  \right) \right| \dd \gamma  \notag\\
& \le \sum_{k=0}^{m- 1}  \int_{\frac{k}{n} }^{\frac{k+1}{n}} \left| \cI_{n}'(\gamma)   - \frac{1}{2} \log\left(1 + \cM_{n}(\gamma)  \right) \right| \dd \gamma \notag\\
&\overset{(a)}{=}   \frac{1}{n}\sum_{k=0}^{m- 1}  \left| I'_{k,n}    - \frac{1}{2} \log\left(1 + M_{k,n}   \right) \right| \notag\\
& \overset{(b)}{=}   \frac{1}{n}\sum_{k=0}^{m-1}  \Delta_{k,n}  \notag \\
& \!\begin{multlined}[b]
\overset{(c)}{\le}     \frac{  C_B }{n}\sum_{k=0}^{m-1} \Big[ \left| I''_{k,n} \right|^\frac{1}{10}  +  \left|  I'_{k+\ell, n} - I'_{k,n} \right|^\frac{1}{4}  \\
  +  \left( 1 + \tfrac{k}{n}\right)^\frac{1}{2} n^{-\frac{1}{24}}  \Big] ,    \label{eq:I_MMSE_relationship_b}
     \end{multlined} 
\end{align}
where: (a) follows from the definitions of $\cI'_n(\delta)$ and $\cM_n(\delta)$; (b) follows from Identity~\eqref{eq:Delta_alt};  and (c) follows from Lemma~\ref{lem:Delta_bound}. To further bound the right-hand side of \eqref{eq:I_MMSE_relationship_b}, observe that
\begin{align}
\frac{1}{n}  \sum_{k=0}^{m-1}  \left| I''_{k,n}  \right|^\frac{1}{10}  &\overset{(a)}{\le} \frac{m^\frac{9}{10}}{n}  \left(  \sum_{k=0}^{m- 1}\left|  I''_{k,n}  \right|\right)^\frac{1}{10} \notag \\
& \overset{(b)}{=} \frac{m^\frac{9}{10}}{n} \left|  I'_{1,n} -I'_{m,n}   \right|^\frac{1}{10} \notag  \\
& \overset{(c)}{\le} \frac{m^\frac{9}{10}}{n} \left( I_{1,n}\right)^\frac{1}{10} \notag  \\
& \overset{(d)}{\le} C_B \cdot  \left( 1+ \delta\right)^\frac{9}{10} \,  n^{-\frac{1}{10}} ,  \label{eq:I_MMSE_relationship_c}
\end{align}
where: (a) follows from H\"olders inequality; (b) follows from the fact that $I''_{m,n}$ is non-positive; (c) follows from the fact that $I'_{m,n}$ is non-increasing in $m$; and  (d) follows from the fact that $I_{1,n}$ is upper bounded by a constant that depends only on $B$. Along similar lines,  
\begin{align}
\frac{1}{n}  \sum_{k=0}^{m-1}  \left| I'_{k,n} -  I'_{k+\ell, n}  \right|^\frac{1}{4}  &\overset{(a)}{\le} \frac{m^\frac{3}{4}}{n}  \left(  \sum_{k=0}^{m- 1}\left|  I'_{k,n}   -  I'_{k + \ell, n} \right|\right)^\frac{1}{4} \notag  \\
& \overset{(b)}{=} \frac{m^\frac{3}{4}}{n} \left(   \sum_{k=0}^{\ell-1} ( I'_{k,n} - I'_{k + m, n})  \right)^\frac{1}{4} \notag \\
& \overset{(c)}{\le} \frac{m^\frac{3}{4}}{n} \left( \ell \cdot I_{1,n}\right)^\frac{1}{4} \notag  \\
& \overset{(d)}{\le} C'_B \cdot  \left( 1+ \delta\right)^\frac{3}{4} \,  n^{-\frac{1}{24}} ,  \label{eq:I_MMSE_relationship_d}
\end{align}
where: (a) follows from H\"olders inequality; (b) and (c)  follow from the fact that $I'_{m,n}$ is non-increasing in $m$; and  (d) follows from the fact that $I_{1,n}$ is upper bounded by a constant that depends only on $B$. Finally, 
\begin{align}
\frac{1}{n}  \sum_{k=0}^{m-1}  \left( 1 + \tfrac{k}{n}\right)^\frac{1}{2} n^{-\frac{1}{24}} & \le  \tfrac{m}{n} \left(1 + \tfrac{m-1}{n} \right)^\frac{1}{2}\, n^{-\frac{1}{24}}  \notag \\
& \le  \left( 1+ \delta\right)^\frac{3}{2}\, n^{-\frac{1}{24}}  .  \label{eq:I_MMSE_relationship_e}
\end{align}
Plugging \eqref{eq:I_MMSE_relationship_c}, \eqref{eq:I_MMSE_relationship_d}, and \eqref{eq:I_MMSE_relationship_e} back into \eqref{eq:I_MMSE_relationship_b} and retaining only the dominant terms yields
\begin{align}
  \int_{0}^{\delta} \left| \cI_{n}'(\gamma)   - \frac{1}{2} \log\left(1 \!+\! \cM_{n}(\gamma)  \right) \right| \dd \gamma
\le  C_B \cdot  (1\! + \! \delta)^\frac{3}{2} \,n^{- \frac{1}{24}} . \notag
\end{align}
This completes the proof of Theorem~\ref{thm:I_MMSE_relationship}.


\section{Proof of Theorem~\ref{thm:MMSE_fixed_point}}\label{proof:thm:MMSE_fixed_point}

\subsection{MMSE Fixed-Point Relationship}\label{sec:MMSE_fixed_point} 

This section shows how the MMSE can be bounded in terms of a fixed-point equation defined by the single-letter MMSE function of the signal distribution. At a high level, our approach focuses on the MMSE of an augmented measurement model,  which contains an extra measurement and extra signal entry, and shows that this augmented MMSE can be related to $M_n$ in two different ways. 

For a fixed signal length $n$ and measurement number $m$, the augmented measurement model consists of the measurements $(Y^m, A^m)$ plus an additional measurement given by
\begin{align}
Z_{m+1} =  Y_{m+1} + \sqrt{G_{m+1}}\,  X_{n+1} , \notag
\end{align}
where  $ G_m \sim \frac{1}{n}  \chi^2_m$ has a scaled chi-square distribution and is independent of everything else.  The observed data is given by the tuple $(Y^m, A^m, \cD_{m+1})$ where
\begin{align}
\cD_{m+1} = (Z_{m+1}, A_{m+1}, G_{m+1}). \notag
\end{align}
The augmented MMSE $\widetilde{M}_{m,n}$ is defined to be the average MMSE of the first $n$ signal entries given this data:
\begin{align}
\widetilde{M}_{m,n} & \triangleq  \frac{1}{n}  \mmse(X^n    \mid Y^{m},  A^{m}, \cD_{m+1}).  \notag
\end{align}

The augmented measurement $Z_{m+1}$ is a noisy version of  the measurement $Y_{m+1}$. Therefore, as far as estimation of the signal $X^n$ is concerned, the augmented measurements are more informative than $(Y^m, A^m)$, but less informative than $(Y^{m+1}, A^{m+1})$.  An immediate consequence of the data processing inequality for MMSE \cite[Proposition 4]{rioul:2011}, is that the augmented MMSE is sandwiched between the MMSE sequence: 
\begin{align}
M_{m+1,n} \le \widetilde{M}_{m+1,n}  \le M_{m,n}. \label{eq:tildeM_sandwich} 
\end{align}
The following result follows immediately from \eqref{eq:tildeM_sandwich} and the smoothness of the posterior variance given in Lemma~\ref{lem:post_var_smoothness}. The proof is given in Appendix~\ref{proof:lem:M_to_M_aug}. 

\begin{lemma}\label{lem:M_to_M_aug} 
Under Assumptions 1 and 2, the augmented MMSE $\widetilde{M}_{m,n}$ and the MMSE $M_{m,n}$ satisfy
\begin{align}
\left| \widetilde{M}_{m,n} -  M_{m,n}  \right| \le C_B \cdot \left| I''_{m,n} \right|^\frac{1}{2}. \label{eq:M_to_M_aug} 
\end{align}
\end{lemma} 

The next step in the proof is to show that $\widetilde{M}_m$ can also be expressed in terms of the single-letter MMSE function $\mmse_X(s)$. The key property of the augmented measurement model that allows us to make this connection is given by the following result. The proof is given in Appendix~\ref{proof:lem:MMSE_aug_alt}. 

\begin{lemma}\label{lem:MMSE_aug_alt}
Under Assumptions 1 and 2, the augmented MMSE  can be expressed equivalently in terms of the last signal entry:
\begin{align}
\widetilde{M}_{m,n} = \mmse(X_{n+1} \mid Y^m, A^m, \cD_{m+1}). \label{eq:M_aug_alt} 
\end{align} 
\end{lemma}

To see why the characterization in \eqref{eq:M_aug_alt} is useful note that the first $m$ measurements $(Y^m, A^m)$ are independent of $X_{n+1}$. Thus, as far as estimation of $X_{n+1}$ is concerned, the relevant information provided by these measurements is summarized by the conditional distribution of $Y_{m+1}$ given $(Y^m,A^{m+1})$. This observation allows us to leverage results from Section~\ref{sec:Gaussiannness}, which focused on the non-Gaussianness of this distribution. The proof of the following result is given in  Section~\ref{proof:lem:MMSE_aug_bound}.

\begin{lemma}\label{lem:MMSE_aug_bound}
Under Assumptions 1 and 2, the augmented MMSE  and  the MMSE satisfy
\begin{align}
 \left| \widetilde{M}_{m,n}  - \mmse_X\left( \frac{m/n}{1 + M_{m,n}} \right)  \right|  & \le C_B \cdot  \left( \sqrt{ \Delta_{m,n} }  +  \frac{ \sqrt{m}}{n} \right),  \notag
\end{align}
where $\Delta_{m,n}$ is the non-Gaussianness of new measurements.
\end{lemma}

\subsection{Final Steps in Proof of Theorem~\ref{thm:MMSE_fixed_point}}
Fix any $n \in \integers$ and $\delta\in \reals_+$ and let $m = \lceil \delta n \rceil$. Then, we can write
\begin{align}
\MoveEqLeft \int_{0}^{\delta} \left| \cM_{n}(\gamma)    - \mmse_X\left( \frac{\gamma  }{1 + \cM_n(\gamma) }    \right) \right| \dd \gamma  \notag\\
 & \le \sum_{k=0}^{m-1} \int_{ \frac{k}{n}}^{\frac{k+1}{n}}  \left| \cM_{n}(\gamma)    - \mmse_X\left( \frac{\gamma  }{1 + \cM_n(\gamma) }    \right) \right| \dd \gamma  \notag\\
& \overset{(a)}{=} \sum_{k=0}^{m-1} \int_{ \frac{k}{n}}^{\frac{k+1}{n}}  \left| M_{k,n}  - \mmse_X\left( \frac{\gamma  }{1 + M_{k,n}}    \right) \right| \dd \gamma  \notag\\
&\overset{(b)}{\le}  \frac{1}{n} \sum_{k=0}^{m-1} \left(  \left| M_{k,n}  - \mmse_X\left( \frac{k/n  }{1 + M_{k,n}}    \right) \right|  + \frac{4 B}{n} \right)   \notag\\
& \overset{(c)}{\le } \frac{C_B}{n} \sum_{k=0}^{m-1} \left( \left| I''_{k,n}\right|^\frac{1}{2}   +  \left|  \Delta_{k,n} \right|^\frac{1}{2}   +  \frac{ \sqrt{k}}{n}  + \frac{1}{n} \right),  \label{eq:MMSE_fixed_point_b}
\end{align}
where: (a) follows from definition of $\cM_n(\delta)$; (b) follows from the triangle inequality and Lemma~\ref{lem:mmseX_bound}; and  (c) follows from Lemmas~\ref{lem:M_to_M_aug}  and \ref{lem:MMSE_aug_bound}. To further bound the right-hand side of \eqref{eq:MMSE_fixed_point_b}, observe that, by the same steps that let to Inequality~\eqref{eq:I_MMSE_relationship_c}, 
\begin{align}
\frac{1}{n} \sum_{k=0}^{m-1}  \left| I''_{k,n}\right|^\frac{1}{2} 
 & \le   C_B \cdot (1 + \delta)^\frac{1}{2} n^{-\frac{1}{2}}.   \label{eq:MMSE_fixed_point_c}
\end{align}
Furthermore,
\begin{align}
\frac{1}{n} \sum_{k=0}^{m-1}  \left|  \Delta_{k,n} \right|^\frac{1}{2} & \overset{(a)}{ \le}  \sqrt{ \frac{m}{n}}  \sqrt{ \frac{1}{n} \sum_{k=0}^{m-1} \Delta_{k,n}  }  \notag \\
&  \overset{(b)}{\le} C_B \cdot  (1 + \delta)^\frac{1}{2} \sqrt{ (1\! + \! \delta)^\frac{3}{2} \,n^{- \frac{1}{24}}   } \notag \\
& =  C_B \cdot (1 + \delta)^\frac{7}{4} \, n^{-\frac{1}{48}} \label{eq:MMSE_fixed_point_d}
\end{align}
where (a) follows from the Cauchy-Schwarz inequality, and (b) follows from proof of Theorem~\ref{thm:I_MMSE_relationship}. Finally, 
\begin{align}
\frac{1}{n} \sum_{k=0}^{m-1}  \left(\frac{ \sqrt{k}}{n}  + \frac{1}{n} \right)  & \le \frac{ m( 1 + \sqrt{m-1})}{n^2}  \le   (1 + \delta)^\frac{3}{2} n^{-\frac{1}{2}}. \label{eq:MMSE_fixed_point_e}
\end{align}

Plugging \eqref{eq:MMSE_fixed_point_c}, \eqref{eq:MMSE_fixed_point_d}, and \eqref{eq:MMSE_fixed_point_e}, back into \eqref{eq:MMSE_fixed_point_b} and keeping only the dominant terms leads to 
\begin{multline}
\int_{0}^{\delta} \left| \cM_{n}(\gamma)    - \mmse_X\left( \frac{\gamma  }{1 + \cM_n(\gamma) }    \right) \right| \dd \gamma\\
  \le C_B \cdot  (1 + \delta)^\frac{7}{4} \, n^{-\frac{1}{48}}.  \notag
\end{multline}
This completes the proof of Theorem~\ref{thm:MMSE_fixed_point}.


\section{Proof of Theorem~\ref{thm:limits}} 

The proof of Theorem~\ref{thm:limits} is established by combining implications of the single-crossing property, the constraints on the MI and MMSE given in Theorems~\ref{thm:I_MMSE_relationship}, \ref{thm:MMSE_fixed_point}, and \ref{thm:I_m_boundary},  and standard results from functional analysis.

\subsection{The Single-Crossing Property}\label{sec:single_crossing}

The fixed point curve is given by the graph of the function
\begin{align}
\delta_\fp(z) = (1 + z) \mmse_X^{-1}(z),  \notag
\end{align}
where $\mmse_X^{-1}(z)$ is the functional inverse of $\mmse_X(s)$. The function $\delta_\fp(z)$ is continuously differentiable over its domain because $\mmse_X(s)$ is smooth on $(0,\infty)$  \cite[Proposition~7]{guo:2011}.

The function $\delta_\rs(z)$ is defined to be the functional inverse of the replica-MMSE function
\begin{align}
\delta_\rs(z)  &= \MR^{-1}(z).  \notag
\end{align}
The function $\delta_\rs$ is continuous and non-increasing because $\MR$ is strictly decreasing. Note that jump discontinuities in $\MR(\delta)$ correspond to flat sections in $\delta_\rs(z)$. 

Using these definitions we can now provide a formal definition of the single-crossing property. 

\begin{definition}[Single-Crossing Property] 
A signal distribution $P_X$ has the single-crossing property if  $\delta_\rs - \delta_\fp$ has at most one zero-crossing. In other words, there exists $z_* \in \reals_+$ such that $\delta_\rs - \delta_\fp$ is nonpositive or nonnegative on $[0,z_*]$ and nonpositive or nonnegative on $[z_*, \infty)$.  
\end{definition}

\begin{lemma}\label{lem:delta_rs_alt}
If the signal distribution $P_X$ has the single-crossing property there exists a point $(z_*,\delta_*) \in \fp$ such that
\begin{align}
\delta_\rs(z) & = \begin{cases}
\max( \delta_\fp(z), \delta_*), & \text{if $z \in [0,z_*]$}  \notag\\
\min( \delta_\fp(z), \delta_*), & \text{if $z\in [z_*,\infty)$}.  \notag
\end{cases}
\end{align}
\end{lemma}
\begin{proof}
If $\delta_\rs(z) = \delta_\fp(z)$ for all $z \in \reals_+$ then this representation holds for every point in $\fp$ because $\delta_\rs$ is non-increasing. Alternatively, if there exists $(u_*, \delta_*) \in \reals^2_+$ such that  $\delta_* = \delta_\rs(u_*) \ne \delta_\fp(u_*) $, then it must be the case that the global minimum of $Q_*(z) \triangleq R(\delta_*, z)$  is attained at more than one point. More precisely,  there exists $z_1 <u_* <  z_2$ such that
\begin{align}
& Q_*( z_1)  = Q_*(z_2)  = \min_{z} Q_*(z)  \label{eq:delta_rs_alt_b}\\
& Q_*( u )  >  \min_{z} Q_*(z)   \quad \text{for some $u \in (z_1, z_2)$}. \label{eq:delta_rs_alt_c}
\end{align}
To see why the second constraint follows from the assumption $\delta_\rs(u_*) \ne \delta_\fp(u_*)$, note that if $Q_*(z)$ were constant over the interval $[z_1, z_2]$,  that would mean that $Q'_*(z) = R_z(\delta_* , z) = 0 $ for all $z \in [z_1, z_2]$. This is equivalent to saying that every point on the line from $(\delta_*, z_1)$ to $(\delta_*, z_2)$ is on the fixed-point curve, which is a contradiction. 

Now, since $\delta_\rs(z)$ is non-increasing and equal to $\delta_*$ at both $z_1$ and $z_2$ we  know that $\delta_\rs(z) = \delta_*$ for all $z \in [z_1, z_2]$. Furthermore, since $z_1$ and $z_2$ are minimizers of $R(\delta_*, z)$, we also know that  $\delta_\fp(z_1) = \delta_\fp(z_2) = \delta_*$. 

Next, we will show that the the function $\delta_\fp(z)  - \delta_*$ must have at least one negative-to-positive zero-crossing on  $(z_1, z_2)$. Recall that the function $Q_*(z)$ is continuous, has global minima at $z_1$ and $z_2$, and it not constant over $[z_1,z_2]$. Therefore, it must attain a local maximum on the open interval $(z_1,z_2)$. Since it is continuously differentiable, this means that there  exists $u_1, u_2 \in (z_1, z_2)$ with $u_1 < u_2$ such that $Q'_*(u_1) > 0$ and $Q'_*(u_2) < 0$. The sign changes in $Q_*(z)$ can be related to the sign changes in $\delta_\fp(z) - \delta_*$ be noting that
\begin{align}
\sgn( Q'_*(z) ) &   \overset{(a)}{=} \sgn\left( z- \mmse_X\left( \frac{\delta_*}{1 + z} \right)  \right) \notag\\
&  \overset{(b)}{=}  - \sgn\left( \mmse^{-1}_X(z)-  \frac{\delta_*}{1 + z} \right)  \notag\\
& =  - \sgn\left( \delta_\fp(z) - \delta_*\right),  \notag
\end{align}
where (a) follows from \eqref{eq:R_z} the fact that $\delta_*$ can be taken to be strictly positive and (b)  follows from the fact that $\mmse_X(s)$ is strictly decreasing. As a consequence, we see that $\delta_\fp(u_1) < \delta_*$ and $\delta_\fp(u_2) > \delta_*$, and thus $\delta_\fp(z) - \delta_*$ has at least one negative-to-positive zero-crossing on the interval $(z_1, z_2)$.

At this point, we have shown that every tuple $(\delta_*, z_1, z_2)$ satisfying \eqref{eq:delta_rs_alt_b} and \eqref{eq:delta_rs_alt_c}  leads to at least one negative-to-positive zero-crossing of $\delta_\fp - \delta_\rs$. Therefore, if the signal distribution has the single-crossing property, there can be at most one  such tuple. This implies that  $\delta_\fp(z)  = \delta_\rs(z)$ for all $z \in [0,z_1] \cup [z_2, \infty)$. 
Furthermore, by the continuity of $\delta_\fp$, there exists a point $z_* \in (u_1, z_2)$ such that $\delta_\fp(z_*) = \delta_*$ and
\begin{align}
z \le z_* &\implies \delta_\rs(z) \ge \delta_\fp(z) \notag\\
z \ge z_* &\implies \delta_\rs(z) \le \delta_\fp(z).  \notag
\end{align}
Combining these observations leads to the stated result. 
\end{proof}


Next, for each $g \in \cV$, we use $\delta_g(z) = g^{-1}(z)$ to denote the functional inverse. The function $\delta_g$ is continuous and non-increasing because $g$ is strictly decreasing.

\begin{lemma}\label{lem:delta_g_bounds}
 If the signal distribution $P_X$ has the single-crossing property then, for every $g \in \cV$, the function $\delta_g$ is either an upper bound or lower bound on $\delta_\rs$. 
\end{lemma} 
\begin{proof}
Let $(\delta_*,z_*)$ be the point described in Lemma~\ref{lem:delta_rs_alt}.  Since $\delta_g$ is non-increasing, we have $\delta_g(z) \ge \delta_g(z_*)$ for all $z \in [0,z_*]$ and $\delta_g(z) \le \delta_g(z_*)$ for all $z \in [z_*, \infty)$. Combining these inequalities with the fact that $\delta_g$ is lower bounded by the lower envelop of the fixed-point curve leads to
\begin{align}
\delta_g(z) & \ge 
\begin{dcases}
\min_{u \in [0,z]} \max( \delta_\fp(u) ,  \delta_g(z_*)), &\text{if $z \in [0, z_*]$}\\
\min_{u \in [z_*, z]} \min( \delta_\fp(u) , \delta_g(z_*)), &\text{if $z \in [z_*,\infty)$}.
\end{dcases}. \notag
\end{align}
Therefore, if $\delta_g(z_*) \ge \delta_*$, we see that
\begin{align}
\delta_g(z) & \ge
\begin{dcases}
\min_{u \in [0,z]} \max( \delta_\fp(u) ,  \delta_*), &\text{if $z \in [0, z_*]$}\\
\min_{u \in [z_*, z]} \min( \delta_\fp(u) , \delta_*), &\text{if $z \in [z_*,\infty)$}. 
\end{dcases} \notag\\
&\overset{(a)}{ =} \begin{dcases}
\min_{u \in [0,z]} \delta_\rs(u) , &\text{if $z \in [0, z_*]$}\\
\min_{u \in [z_*, z]} \delta_\rs(u), &\text{if $z \in [z_*,\infty)$}.
\end{dcases} \notag\\
& \overset{(b)}{=} \delta_\rs(z) , \notag
\end{align}
where (a) follows from Lemma~\ref{lem:delta_rs_alt} and (b) follows from the fact that $\delta_\rs$ is non-increasing. Alternatively, if $\delta_g(z_*) \le \delta_*$ then a similar argument can be used to show that $\delta_g(z) \le \delta_\rs(z)$ for all $z \in \reals_+$. 
\end{proof}

\begin{lemma}\label{lem:G_unique}
 If the signal distribution $P_X$ has the single-crossing property, then $\cG$ is equal to the equivalence class of functions in $\cV$ that are equal to $\MR$ almost everywhere. 
\end{lemma} 
\begin{proof}
Recall that $\cG$ is the set of all functions $g \in \cV$ that satisfy the boundary condition
\begin{align}
\lim_{\delta \to \infty}  \left| \int_0^\delta \frac{1}{2} \log(1 + g(\gamma) ) \dd \gamma - \IR(\delta) \right| = 0. \label{eq:G_unique_b}
\end{align}
Furthermore, for each $g \in \cV$ and $\delta \in \reals_+$, we can write
\begin{align}
\MoveEqLeft \left| \int_0^\delta \tfrac{1}{2} \log(1 + g(\gamma) ) \dd \gamma - \IR(\delta) \right|  \notag\\
& \overset{(a)}{=}  \left| \int_0^\delta \! \tfrac{1}{2} \log(1\! +\! \MR(\gamma) ) \dd \gamma - \int_0^\delta \! \tfrac{1}{2} \log(1 \!+\! g(\gamma) ) \dd \gamma  \right|  \notag\\
& \overset{(b)}{=}  \int_0^\delta  \left|\tfrac{1}{2} \log(1\! +\! \MR(\gamma) )  -\tfrac{1}{2} \log(1\! +\! g(\gamma) )\right| \dd \gamma,    \label{eq:G_unique_c}
\end{align}
where (a) follows from \eqref{eq:derivative_replica} and (b) follows from the monotonicity of $g$ and $\MR$ and Lemma~\ref{lem:delta_g_bounds}. Combining \eqref{eq:G_unique_b} and \eqref{eq:G_unique_c},  we see that, for all $g \in \cG$, 
\begin{align}
  \int_0^\infty  \left|\frac{1}{2} \log(1 + \MR(\gamma) )  -\frac{1}{2} \log(1 + g(\gamma) )\right| \dd \gamma   = 0, \notag
\end{align}
and thus $\MR$ and $g$ are equal almost everywhere. 
\end{proof}

\subsection{Convergence of Subsequences}

For each $n\in \mathbb{N}$, the function $\cM_n$ is a non-increasing function from $\rgz$ to $\rgz$.
Convergence of the sequence $\cM_n$ can be treated in a few different ways.
In our original approach~\cite{reeves:2016a}, we focused on the L\'{e}vy metric~\cite[Ch.~2]{Gnedenko-1968}.
Here, we present a more direct argument based on the Helly Selection Theorem~\cite[Thm.~12]{Hanche-em10}.

First, we let $L^1 ([0,S])$ represent the standard Banach space of Lebesgue integrable functions from $[0,S]$ to $\mathbb{R}$ with norm
\[ \int_0^S |f(\delta)| \dd \delta. \]
In this space, two functions $f,g$ are called \emph{equivalent} if they are equal almost-everywhere (i.e., $\int_0^S|f(\delta)-g(\delta)|\dd \delta = 0$).
Next, we recall that monotone functions are continuous almost everywhere (e.g., except for a countable set of jump discontinuities) \cite{Rudin-1976}.
Thus, $f,g$ are equivalent if and only if they are equal at all points of continuity.

The following lemmas outline our approach to convergence.

\begin{lemma}
\label{lem:subseq_limit} Under Assumptions 1 and 2, for any $S>0$ and any subsequence of 
$(\mathcal{M}_{n},\mathcal{I}_{n})$, there is a further subsequence (whose index is denoted by $n'$) and some $g\in \cG$ such that
\begin{align}
\lim_{n'} & \int_0^S \big| \mathcal{M}_{n'}(\delta) - g(\delta) \big| \dd \delta =0  \notag \\
\lim_{n'} & \int_0^S \left| \mathcal{I}_{n'}'(\delta) - \frac{1}{2}\log\left(1+g(\delta)\right)\right|  \dd \delta = 0.  \notag
\end{align}
\end{lemma}

\begin{proof}
For any $S>0$ and each $n\in \mathbb{N}$, the restriction of $\cM_n (\delta)$ to $\delta\in [0,S]$ is non-increasing and uniformly bounded by $\cM_n(0) =\var(X)$.  Since $\cM_n (\delta)$ is nonnegative and non-increasing, its total variation on $[0,T]$ equals $\cM_n (0) - \cM_n(S) \leq \var(X)$~\cite[Section~6.3]{Royden-2010}.

Based on this, the Helly Selection Theorem~\cite[Thm.~12]{Hanche-em10} shows that any subsequence of $\cM_n$ contains a further subsequence that converges in $L^1 ([0,S])$.
Let $\cM_{n'}$ denote this further subsequence and $\mathcal{M}_*$ denote its limit so that
\[
\lim_{n'} \int_0^S \left| \mathcal{M}_{n'} (\delta) - \mathcal{M}_* (\delta) \right| \dd \delta = 0.
\]

To simplify notation, we define the operator $T\colon L^1 ([0,S]) \to L^1 ([0,S])$ via $(Tf)(\delta) \mapsto \mmse_X \left(\delta/(1+f(\delta))\right)$.
To analyze $\mathcal{M}_* (\delta)$, we observe that, for all $n$, one has
\begin{align}
\int_0^S \!\!\! & \; \big| \cM_{*}  (\delta) - T\cM_{*} (\delta) \big|  \dd \delta \leq  \notag
\int_0^S \big|\cM_{*}(\delta)- \cM_{n}(\delta) \big| \dd \delta  \notag\\
&\quad +  \int_0^S \big|\cM_{n}(\delta) -T\cM_{n}(\delta) + T\cM_{n}(\delta)-T\cM_{*}(\delta) \big| \dd \delta  \notag\\
&\leq 
(1+L_T) \int_0^S\big|\cM_{*}(\delta)- \cM_{n}(\delta) \big| \dd \delta  \notag\\
& \quad + \int_0^S\big|\cM_{n}(\delta) -T\cM_{n}(\delta) \big|\dd \delta,  \notag
\end{align}
where $L_T$ is the Lipschitz constant of $T$.
Under Assumption 2, one can use Lemma~\ref{lem:mmseX_bound}, to show that $L_T \leq 4 B S$.
Since $\int_0^S \big|\cM_{*} (\delta)- \cM_{n'} (\delta) \big| \dd \delta \to 0$ by construction and $\int_0^S \big| \cM_{n'}(\delta) -T\cM_{n'} (\delta) \big| \dd \delta \to 0$ by Theorem~\ref{thm:MMSE_fixed_point}, taking the limit along this subsequnce shows that $\cM_{*}$ equals $T\cM_{*}$ almost everywhere on $[0,S]$.
As $S$ was arbitrary, we see that $\cM_*$ satisfies the first condition of Definition~\ref{lem:replica_mmse_set}.

To establish the second condition of Definition~\ref{lem:replica_mmse_set}, we focus on the sequence $\cI_{n}$. Recall that each $\cI_{n}$ is concave and differentiable, with derivative $\cI_{n}'$. 
Also, the set $\{\cI_{n}\}$ is uniformly bounded on $[0,S]$ and uniformly Lipschitz by \eqref{eq:IX_Gbound}, \eqref{eq:Ip_bound}, and \eqref{eq:Im_bounds}.
By the Arzel\`{a}-Ascoli theorem~\cite[Section~10.1]{Royden-2010}, this implies that any subsequence of $\cI_{n}$ contains a further subsequence that converges uniformly on $[0,S]$.
Moreover, the limiting function is concave and the further subsequence of derivatives also converges to the derivative of the limit function at each point where it is differentiable~\cite[Corollary~1.3.8]{Niculescu-2009}.

Thus, from any subsequence of $(\cM_{n},\cI_{n})$, we can choose a further subsequence (whose index is denoted by $n'$) such that $\int_0^S \big| \mathcal{M}_{n'} (\delta) - \mathcal{M}_* (\delta) \big| \dd \delta \to 0$ and $\cI_{n'}$ converges uniformly on $[0,S]$
to a concave limit function $\cI_* $. Moreover, the
sequence of derivatives $\mathcal{I}_{n'}'$ also converges
to $\cI_* '$ at each point where $\cI_*$
is differentiable.
Since $\cI$ is concave, it is
differentiable almost everywhere and we have
\[
\lim_{n'}\mathcal{I}_{n'}'(\delta)=\mathcal{I}_* '(\delta)
\]
almost everywhere on $[0,S]$.
Since $\left| \cI_{n'} ' (\delta) - \cI_{*}' (\delta) \right|$ is bounded and converges to zero almost everywhere on $[0,S]$, we can apply the dominated convergence theorem to see also that $\int_0^S \big|\cI_{n'}' (\delta) - \cI_{*}' (\delta) \big| \dd \delta \to 0$.
Next, we can apply Theorem~\ref{thm:I_MMSE_relationship} to see that
\[
\lim_{n'}\int_{0}^{S}\left|\cI'_{n'}(\delta)-\frac{1}{2}\log\left(1+\cM_{n'}(\delta)\right)\right|\dd\delta=0.
\]
Since $\frac{1}{2}\log\left(1+z\right)$ is Lipschitz in $z$, one finds that
\[ \int_0^S \left| \frac{1}{2} \log \left(1 + \mathcal{M}_{n'} (\delta) \right) - \frac{1}{2} \log \left(1+ \mathcal{M}_* (\delta) \right) \right| \dd \delta \to 0 \]
follows from the fact that
$\int_0^S \big| \mathcal{M}_{n'} (\delta) - \mathcal{M}_* (\delta) \big| \dd \delta \to 0$.
Along with the triangle inequality, this shows that $\cI_{n'}'(\delta)$ converges to $\frac{1}{2} \log (1+\cM_* (\delta))$ almost everywhere on $[0,S]$. Since $\left| \cI_{n'}'(\delta) - \frac{1}{2} \log (1+\cM_* (\delta)) \right|$ is bounded and converges to zero almost everywhere on $[0,S]$, we can apply the dominated convergence theorem to see that
\begin{equation} \label{eq:Istar_int_Mstar}
\cI_* (S)  =\lim_{n'}  \int_0^S \cI_{n'}' (\delta) \dd \delta
 = \int_{0}^{S} \frac{1}{2} \log (1+\cM_* (\delta)) \dd \delta.
\end{equation}

Next, we observe that Theorem~\ref{thm:I_m_boundary} implies 
\[
\lim_{n'} \; \left|\IR \left(S\right)-\cI_{n'}(S)\right|\leq\frac{C}{\sqrt{S}},
\]
for all $S \ge 4$.  With \eqref{eq:Istar_int_Mstar}, this implies that 
\begin{align} 
&\left| \IR(S) - \int_{0}^{S} \frac{1}{2}\log\big( 1+\cM_*(\delta) \big) \dd \delta \right|  \notag\\
& \leq \left| \IR(S) - \cI_{n'} (S) \right| + \left| \cI_{n'} (S) - \int_{0}^{S} \frac{1}{2}\log\big( 1+\cM_*(\delta) \big) \dd \delta \right|  \notag\\
& \leq \frac{C}{\sqrt{S}} + \epsilon_{n'},  \notag
\end{align}
where $\lim_{n'} \epsilon_{n'} =0$.
Taking the limit $n' \to \infty$ followed by the limit $S\to \infty$, we see
\begin{equation} 
\lim_{S\to \infty} \left| \IR(S) - \int_{0}^{S} \frac{1}{2}\log\big( 1+\cM_*(\delta) \big) \dd \delta \right| = 0  \notag
\end{equation}
and, thus, that $\cM_* \in \cG$. 
Notice that we focus first on finite $S \in \reals_+$ and then take the limit $S\to \infty$.
This is valid because the functions $\cI(\delta)$ and $\cM_n (\delta)$ are defined for all $\delta \in \reals_+$ but restricted to $[0,S]$ for the convergence proof.
\end{proof}

Now, we can complete the proof of Theorem~\ref{thm:limits}.
The key idea is to combine Lemma~\ref{lem:G_unique} with Lemma~\ref{lem:subseq_limit}.
From these two results, it follows that, for any $S>0$, every subsequence of $\cM_n (\delta)$ has a further subsequence that converges to $\MR(\delta)$.
This holds because the further subsequence must converge to some function in $\cG$ (by Lemma~\ref{lem:subseq_limit}) but there is only one function up to almost everywhere equivalence (by Lemma~\ref{lem:G_unique}).

The final step is to realize that this is sufficient to prove that, for all $S>0$, we have
\begin{equation}
\lim_{n\to\infty} \int_0^S \left| \MR(\delta) - \cM_n (\delta) \right| \dd \delta = 0.  \notag
\end{equation}
To see this, suppose that $\cM_n (\delta)$ does not converge to $\MR(\delta)$ in $L^1 ([0,S])$.
In this case, there is an $\epsilon >0$ and an infinite subsequence $n(i)$ such that
\begin{equation}
\int_0^S \left| \MR(\delta) - \cM_{n(i)} (\delta) \right| \dd \delta > \epsilon  \notag
\end{equation}
for all $i\in \mathbb{N}$.
But, applying Lemma~\ref{lem:subseq_limit} shows that $n(i)$ has a further subsequence (denoted by $n'$) such that
\begin{equation}
\lim_{n'} \int_0^S \left| \MR(\delta) - \cM_{n'} (\delta) \right| \dd \delta = 0.  \notag
\end{equation}
From this contradiction, one must conclude that $\cM_n (\delta)$ converges to $\MR(\delta)$ in $L^1 ([0,S])$ for any $S>0$.

\section{Conclusion}

In this paper, we present a rigorous derivation of the fundamental limits of compressed sensing for i.i.d.\ signal distributions and i.i.d.\ Gaussian measurement matrices. We show that the limiting MI and MMSE are equal to the values predicted by the replica method from statistical physics. This resolves a well-known open problem.

\appendices

\section{Useful Results} 
\subsection{Basic Inequalities} 

We begin by reviewing a  number of basic inequalities. For numbers $x_1, \cdots x_n$ and $p\ge 1$,  Jensen's inequality combined with the convexity of $|\cdot|^p$ yields, 
\begin{align}
\left| \sum_{i =1}^n x_i \right|^p\
& \le  n^{p-1}  \sum_{i =1}^n \left| x_i \right|^p, \qquad p \ge 1 \label{eq:LpBound_1}.
\end{align}
In the special  case $n =2$ and $p \in \{2,4\}$, we obtain
\begin{align}
(a+b)^2 & \le 2 (a^2 + b^2) \label{eq:ab2}\\
(a+b)^4 & \le 8 (a^4 + b^4) \label{eq:ab4}.
\end{align}

For random variables $X_1, \cdots , X_n$ and $p \ge1$, a consequence of Minkowski's inequality  \cite[Theorem 2.16]{boucheron:2013} is that
\begin{align}
\ex{ \left| \sum_{i =1}^n X_i \right|^p} & \le \left( \sum_{i =1}^n \left( \ex{ \left| X_i \right|^p } \right)^\frac{1}{p}  \right)^p, \qquad p \ge 1. \label{eq:Minkowski} 
\end{align}
Also, for random variables $X$ and $Y$, an immediate consequence of Jensen's inequality is that expectation of the absolute difference $|X-Y|$ can be upper bounded  in terms of higher moments, i.e.,
\begin{align}
\ex{|X-Y|} \le  \left | \ex{|X - Y|^p}\right|^{\frac{1}{p}}, \qquad p \ge 1. \notag
\end{align} 
Sometimes, we need to bound the differnence in terms of weaker measures of deviation between $X$ and $Y$. The following Lemma provides two such bounds that also depend on the moments of $X$ and $Y$. 

\begin{lemma}\label{lem:L1_to_Log} 
For nonnegative random variables $X$ and $Y$, the expectation of the absolute difference $|X-Y|$ obeys the following upper bounds:
\begin{align}
\ex{ \left | X - Y \right|} &\le \sqrt{ \frac{1}{2} (\ex{X^2} + \ex{Y^2})  \ex{ \left | \log(X/Y) \right|}} \label{eq:abs_dev_bnd_log}\\
\ex{ \left | X - Y \right|} &\le \sqrt{ 2 (\ex{X} + \ex{Y})   \ex{ \left | \sqrt{X} - \sqrt{Y} \right|}}. \label{eq:abs_dev_bnd_sqrt}
\end{align}
\end{lemma}
\begin{proof}
We begin with \eqref{eq:abs_dev_bnd_log}. For any numbers $0 < x < y$, the difference $y-x$ can be upper bounded as follows:
\begin{align}
y-x  
&= \int_x^y \sqrt{u} \frac{1}{\sqrt{u}} du  \notag\\
& \le \sqrt{ \int_x^y u\,  du}  \sqrt{ \int_x^y \frac{1}{u} du}  \notag\\
&  = \sqrt{  \frac{1}{2} (y^2 - x^2) }  \sqrt{  \log(y/x) }  \notag\\
&  \le \sqrt{  \frac{1}{2} (y^2 + x^2) }  \sqrt{  \log(y/x) },  \notag
\end{align}
where the first inequality is due to the Cauchy-Schwarz inequality.  Thus, the absolute difference between $X$ and $Y$ obeys
\begin{align}
|X-Y|  &  \le \sqrt{  \frac{1}{2} (X^2 + Y^2)}  \sqrt{ \left | \log(X/Y)  \right|}.  \notag
\end{align}
Taking the expectation of both sides and using the  Cauchy-Schwarz inequality leads to \eqref{eq:abs_dev_bnd_log}.

To prove \eqref{eq:abs_dev_bnd_sqrt}, observe that the difference between between $X$ and $Y$ can be decomposed as
\begin{align}
X - Y  = (\sqrt{X} + \sqrt{Y}) (\sqrt{X} - \sqrt{Y}).  \notag
\end{align}
Thus, by the Cauchy-Schwarz inequality, 
\begin{align}
\ex{ |X - Y|}   &= \sqrt{ \ex{ (\sqrt{X} + \sqrt{Y})^2}} \sqrt{ \ex{  (\sqrt{X} - \sqrt{Y})^2}}  \notag\\
& \le \sqrt{ 2( \ex{X} + \ex{Y})} \sqrt{ \ex{  (\sqrt{X} - \sqrt{Y})^2}},   \notag
\end{align}
where the last step is due to \eqref{eq:ab2}.
\end{proof}
%
\subsection{Variance Decompositions} 

This section reviews some useful decompositions and bounds on the variance. As a starting point, observe that the variance of a random variable $X$ can be expressed in terms of an independent copy $X'$ according to 
\begin{align}
\var(X) = \frac{1}{2} \ex{ (X - X')^2}.   \notag
\end{align}
This representation can  extended to  conditional variance, by letting $X_y$ and $X'_y$ denote independent draws from the conditional distribution $P_{X|Y= y}$, so that
\begin{align}
\var(X \gmid Y = y) = \frac{1}{2} \ex{ (X_y - X'_y)^2}. \notag
\end{align}
For a random draw of $Y$, it then follows that the random conditional variance of $X$ given $Y$ can be expressed as
\begin{align}
\var(X\gmid Y) = \frac{1}{2} \ex{ (X_Y - X'_Y)^2 \gmid Y }, \label{eq:cond_var_decomp}
\end{align}
where $X_Y$ and $X'_Y$ are conditionally independent draws from the random conditional distribution $P_{X|Y}(\cdot | Y)$. 

Using this representation, the moments of the conditional variance can be bounded straightforwardly.  For all $p \ge 1$, 
\begin{align}
\ex{ \left|\var(X | Y) \right|^p} & = \frac{1}{2^p}\ex{ \left| \ex{ (X_Y - X'_Y)^2 \mid Y} \right|^p}  \notag\\
& \overset{(a)}{\le} \frac{1}{2^p}\ex{ \left|  X_Y - X'_Y \right|^{2p}}  \notag\\
& \overset{(b)}{\le} 2^{ p -1} \left(  \ex{ \left| X_Y\right|^{2 p}} +  \ex{ \left| X'_Y \right|^{2 p}} \right)  \notag\\
& \overset{(c)}{=}  2^{ p} \ex{\left|X\right|^{2p}},  \notag
\end{align}
where (a) follows from Jensen's inequality and the convexity of $|\cdot|^p$, (b) follows from \eqref{eq:LpBound_1}, and (c) follows from the fact that $X_Y$ and $X'_Y$ both have the same distribution as $X$. 

The law of total variance gives
\begin{align}
\var(X) = \ex{ \var(X \gmid Y)} + \var( \ex{ X \gmid Y}). \label{eq:law_tot_var}
\end{align}
As an immediate consequence, we obtain the data processing inequality for MMSE (see e.g.\ \cite[Proposition 4]{rioul:2011}) , which states that conditioning cannot increase the MMSE on average. In particular, if $X \to Y \to Z$ form a Markov chain, then, 
\begin{align}
\mmse(X \gmid Y) \le \mmse(X \gmid Z).  \notag
\end{align}


\subsection{Bounds using KL Divergence} \label{sec:cov_to_MI} 

This section provides a number results that that allow us to bound  differences  in  expectations in terms of Kullback--Leibler divergence. One of the consequences of Lemma~\ref{lem:integral_bound} (given below) is that random variables $X \sim P_X$ and $Y \sim P_Y$ with positive and finite second moments satisfy 
\begin{align}
\frac{ \left| \ex{X} - \ex{Y}\right|}{  \sqrt{2 \ex{X^2} +2\ex{Y^2}}} \le  \sqrt{  \DKL{P_X}{P_Y}}.  \notag
\end{align}


We begin by reviewing some basic definitions (see e.g., \cite[Section 3.3]{pollard:2002}). Let $P$ and $Q$ be probability measures with densities $p$ and $q$ with respect to a dominating measure $\lambda$. The Hellinger distance $d_H( P, Q)$  is defined as the $\cL_2$ distance between the square roots of the densities $\sqrt{p}$ and $\sqrt{q}$, and the squared Helliger distance is given by
\begin{align}
d^2_H( P, Q) & =  \int  \left( \sqrt{p}  - \sqrt{q} \right )^2   \dd \lambda.  \notag
\end{align}
The Kullback--Leibler divergence (also known as relative entropy) is defined as
\begin{align}
D_\mathrm{KL}( P\, \| \,  Q) & =   \int   p  \log\left(\frac{ p}{q }\right)  \dd \lambda  .  \notag
\end{align}
The squared  Hellinger distance is upper bounded by the  KL divergence \cite[pg. 62]{pollard:2002}, 
\begin{align}
d^2_H(P, Q) \le D_\mathrm{KL}( P\, \| \,  Q) \label{eq:H_to_KL}. 
\end{align}

\begin{lemma}\label{lem:integral_bound} Let $f$ be a function that is measurable with respect to $P$ and $Q$. Then
\begin{align}
\MoveEqLeft
 \left| \int f\left(  \dd P -\dd Q  \right)  \right|  \notag\\
& \le \sqrt{  2  \int f^2  \left( \dd P + \dd Q \right)  }  \min\left\{  \sqrt{  D_\mathrm{KL}\left( P \, \middle \| \, Q \right)}, 1 \right\}. \notag
\end{align}
\end{lemma}
\begin{proof}
Let $p$ and $q$ be the densities of $P$ and $Q$ with respect to a dominating measures $\lambda$. Then, we can write
\begin{align}
\MoveEqLeft \left| \int f ( \dd P -\dd Q)    \right|  \notag\\
& = \left| \int  f ( p - q )  \dd \lambda  \right| \notag \\
& \overset{(a)}{=} \left| \int  f  (\sqrt{p} + \sqrt{q}) ( \sqrt{p} -\sqrt{q}) \dd \lambda  \right|  \notag \\
& \overset{(b)}{\le}  \sqrt{ \int  f^2  (\sqrt{p} + \sqrt{q})^2 \dd \lambda  \,    d^2_H(P, Q) }  \notag \\ 
& \overset{(c)}{\le}  \sqrt{ \int  f^2  (\sqrt{p} + \sqrt{q})^2 \dd \lambda  \,    \DKL{P}{Q} }  \notag\\
& \overset{(d)}{\le}  \sqrt{2 \int  f^2  (p + q) \dd \lambda  \,    \DKL{P}{Q} }, \label{eq:integral_bound_b}
\end{align}
where (a) is justified by the non-negativity of the densities, (b) follows from the Cauchy-Schwarz inequality, (c) follows from \eqref{eq:H_to_KL}, and (d) follows from \eqref{eq:ab2}. 

Alternatively, we also have the upper bound
\begin{align}
 \left| \int f \left( \dd P -\dd Q  \right)  \right|  \notag
& \le \left| \int f \dd P  \right| + \left|  \int f \dd Q   \right|  \notag\\
& \le \sqrt{ \int f^2 \dd P}  + \sqrt{  \int f^2 \dd Q}  \notag \\
& \le   \sqrt{2 \int f^2 \left(  \dd P +\dd Q \right)} ,  \label{eq:integral_bound_c}
\end{align}
where (a) follows from the triangle inequality, (b) follows from Jensen's inequality, and (c) follows from \eqref{eq:ab2}. Taking the minimum of \eqref{eq:integral_bound_b} and \eqref{eq:integral_bound_c} leads to the stated result.
\end{proof}

\begin{lemma}\label{lem:var_bnd}
For any distribution $P_{X, Y, Z}$ and $p \ge 1$, 
\begin{align}
\MoveEqLeft \ex{ \left| \var(X  \gmid   Y) - \var(X   \gmid  Z) \right|^p}   \notag \\
 & \le  2^{2p + \frac{1}{2}  } \sqrt{ \ex{ |X|^{4p} }\ex{  \DKL{ P_{X|Y}}{ P_{X|Z} }}}.  \notag
\end{align}
\end{lemma}
\begin{proof}
Let $P$ and $Q$ be the random probability measures on $\reals^2$ defined by
\begin{align}
P & =P_{X \mid Y} \times P_{X  \mid Y}  \notag\\
Q & =P_{X \mid Z} \times P_{X  \mid Z},  \notag
\end{align}
and let $f : \reals^2 \to \reals$ be defined by $f(x_1, x_2) = \frac{1}{2} (x_1 - x_2)^2$. Then, by the variance decomposition  \eqref{eq:cond_var_decomp}, we can write
\begin{align}
\var(X \gmid  Y)  & = \int f \dd P , \quad \var(X  \gmid Z)   = \int f \dd Q.  \notag
\end{align}
Furthermore, by the upper bound $f^2(x_1, x_2) \le 2  (x_1^4  + x_2^4)$, the expectation of the $f^2$ satisfies 
\begin{align}
\int f^2 \left(  \dd P  + \dd Q \right)   \le 4  \ex{ X^4 \!\mid\! Y} +  4  \ex{ X^4 \! \mid \! Z} .\label{eq:var_bnd_a}
\end{align}
Therefore, by Lemma~\ref{lem:integral_bound}, the  difference between the conditional variances satisfies
\begin{align}
\MoveEqLeft \left|  \var(X  \gmid   Y) - \var(X  \gmid  Z)  \right|  \notag \\
& \le \sqrt{  2 \int f^2  \left( \dd P + \dd Q \right) }  \sqrt{  \min\left\{  D_\mathrm{KL}\left( P \, \middle \| \, Q \right), 1 \right\}}  \notag\\
& \le  \sqrt{   8 \ex{ X^4  \gmid Y} +   8  \ex{ X^4   \gmid   Z}  }  \sqrt{  \min\left\{  D_\mathrm{KL}\left( P \, \middle \| \, Q \right), 1 \right\}}, \label{eq:var_bnd_b}
 \end{align}
 where the second inequality follows from \eqref{eq:var_bnd_a}. 
 
The next step is to bound the expected $p$-th power of the right-hand side of \eqref{eq:var_bnd_b}. Starting with the Cauchy-Schwarz inequality, we have 
\begin{align}
\MoveEqLeft \ex{  \left|  \var(X \mid  Y) - \var(X \mid Z)  \right|^p}  \notag \\
& \le  \sqrt{ \ex{   \left| 8  \ex{ X^4 \! \mid \! Y} +  8 \ex{ X^4  \! \mid  \! Z}     \right|^p}}  \notag\\
& \quad \times \sqrt{ \ex{ \left| \min\left\{ D_\mathrm{KL}\left( P \, \middle \| \, Q \right) , 1 \right\} \right|^p } }. \label{eq:var_bnd_c}
 \end{align}
For the first term on the right-hand side of \eqref{eq:var_bnd_c}, observe that, by Jensen's inequality, 
\begin{align}
\sqrt{ \ex{   \left| 8  \ex{ X^4 \! \mid \! Y} +  8 \ex{ X^4  \! \mid  \! Z}     \right|^p} } 
& \le  \sqrt{ 8^p \ex{ X^{4p} }  + 8^p \ex{ X^{4p}} }  \notag\\
& = 4^p \sqrt{  \ex{X^{4p} }} .  \label{eq:var_bnd_d}
\end{align}
Meanwhile, the expectation in the second term on the right-hand side of \eqref{eq:var_bnd_c} satisfies 
\begin{align}
 \ex{  \left| \min\left\{  D_\mathrm{KL}\left( P \, \middle \| \, Q \right), 1 \right\}\right|^{p}}  
 & \le \ex{ D_\mathrm{KL}\left( P \, \middle \| \, Q \right)} \notag \\
 & = 2    \ex{D_\mathrm{KL}\left( P_{X|Y} \, \middle \| \, P_{X|Z} \right)} , \label{eq:var_bnd_e}
\end{align}
where the second step follows from the definition of $P$ and $Q$.  Plugging \eqref{eq:var_bnd_d} and \eqref{eq:var_bnd_e} back into \eqref{eq:var_bnd_c} leads to the stated result. 
\end{proof}

\begin{lemma}\label{lem:cov_bnd}
For any distribution $P_{X,Y,Z}$ and $p \ge 1$, 
\begin{align}
\MoveEqLeft  \ex{ \left| \cov(X, Y  \gmid   Z )  \right|^p }  \notag \\
 &\le 2^{p} \sqrt{ \ex{  \left|   \ex{X^4  \gmid Z} \ex{Y^4 \gmid Z} \right|^\frac{p}{2}} I(X; Y\gmid  Z) }.  \notag
\end{align}
\end{lemma}
\begin{proof}
Let $P$ and $Q$ be the random probability measures on $\reals^2$ defined by
\begin{align}
P & =P_{X, Y \mid Z}  \notag\\
Q & =P_{X \mid Z} \times  P_{Y  \mid Z},  \notag
\end{align}
and let $f : \reals^2 \to \reals$ be defined by $f(x,y) = x y$. Then, the conditional covariance between $X$ and $Y$ can be expressed as 
\begin{align}
\cov(X, Y \gmid Z ) & = \int f \left( \dd P - \dd Q \right). \notag
\end{align}
Furthermore
\begin{align}
\int f^2 \left(  \dd P  + \dd Q \right)  & = \ex{ |X Y|^2 \mid Z}  + \ex{X^2 \! \mid \! Z} \ex{ Y^2 \! \mid  \! Z}  \notag\\
& \le 2  \sqrt{ \ex{X^4 \gmid  Z} \ex{ Y^4 \gmid  Z}  },  \notag
\end{align}
where the second step follows from the  Cauchy-Schwarz inequality and Jensen's inequality.  Therefore, by Lemma~\ref{lem:integral_bound}, the  magnitude of the covariance  satisfies 
\begin{align}
\MoveEqLeft  \left| \cov(X, Y \gmid Z )  \right|  \notag\\
& \le \sqrt{  2 \int f^2  \left( \dd P + \dd Q \right) }  \sqrt{  \min\left\{  D_\mathrm{KL}\left( P \, \middle \| \, Q \right), 1 \right\}}  \notag\\
 & \le 2  \left|\ex{X^4 \! \mid \! Z} \ex{ Y^4 \!\mid  \!Z}  \right|^\frac{1}{4}  \left| \min\left\{   D_\mathrm{KL}\left( P \, \middle \| \, Q \right),1 \right\}\right|^{\frac{1}{2}}. \label{eq:cov_bnd_b}
\end{align}

The next step is to bound the expected $p$-th power of the right-hand side of \eqref{eq:cov_bnd_b}. Starting with the Cauchy-Schwarz inequality, we have 
\begin{align}
\ex{ \left| \cov(X, Y \mid Z )  \right|^p }  \notag 
 &\le 2^p \sqrt{ \ex{  \left|   \ex{X^4 \!  \mid \! Z} \ex{Y^4 \! \mid \! Z} \right|^\frac{p}{2}}}     \notag\\
& \quad \times \sqrt{ \ex{ \left| \min\left\{ D_\mathrm{KL}\left( P \, \middle \| \, Q \right) , 1 \right\} \right|^p } }. \label{eq:cov_bnd_c}
\end{align}
Note that the expectation in the second term on the right-hand side of \eqref{eq:cov_bnd_c} satisfies 
\begin{align}
\ex{ \left| \min\left\{   D_\mathrm{KL}\left( P \, \middle \| \, Q \right), 1 \right\}\right|^{p} }
 & \le  \ex{ D_\mathrm{KL}\left( P \, \middle \| \, Q \right)} \notag \\
 & =  I(X; Y \! \mid \! Z) , \label{eq:cov_bnd_e}
\end{align}
where the second step follows from the definition of $P$ and $Q$. Plugging \eqref{eq:cov_bnd_e} back into \eqref{eq:cov_bnd_c} leads to the stated result. 
\end{proof}

\begin{lemma}\label{lem:mmse_diff}
For any distributions $P_{X,Y}$ and $P_{X,Z}$,
\begin{align}
\MoveEqLeft \left| \mmse(X \gmid  Y) - \mmse(X   \gmid   Z) \right|  \notag\\
& \le  2^{\frac{5}{2}} \sqrt{ \ex{ |X|^{4}} \,  D_\mathrm{KL}(P_{X,Y} \, \| \,  P_{X,Z}) }. \notag
\end{align}
\end{lemma}
\begin{proof}
Let $P$ and $Q$ be the distributions given by
\begin{align}
P(x_1,x_2,y) &= P_{X\mid Y}(x_1 \mid y) P_{X\mid Y}(x_2 \mid y) P_{Y}(y)  \notag\\
Q(x_1,x_2,y) &= P_{X\mid Z}(x_1 \mid y) P_{X\mid Z}(x_2 \mid y) P_{Z}(y)  \notag
\end{align}
Then, 
\begin{align}
\mmse(X \gmid Y) = \frac{1}{2} \int (x_1 - x_2)^2  \dd P(x_1,x_2,y)  \notag\\
\mmse(X  \gmid  Z) = \frac{1}{2} \int (x_1 - x_2)^2 \dd Q(x_1,x_2,y)  \notag
\end{align}
and so, by  Lemma~\ref{lem:integral_bound}, 
\begin{align}
\MoveEqLeft \left| \mmse(X \gmid Y) - \mmse(X \gmid Z) \right|  \notag\\
& \le  \sqrt{ \frac{1}{2} \! \int (x_1 \!-\! x_2)^4 \left(\dd P(x_1,x_2,y)  +  \dd Q(x_1, x_2, y) \right)  } \notag  \\
& \quad \times   \sqrt{ D_\mathrm{KL}(P \, \| \,  Q )} . \label{eq:mmse_diff_b}
\end{align}
For the first term on the right-hand side of \eqref{eq:mmse_diff_b}, observe that
\begin{align}
 \MoveEqLeft \int (x_1 \!-\! x_2)^4 \left(\dd P(x_1,x_2,y)  +  \dd Q(x_1, x_2, y) \right)  \notag\\
& \overset{(a)}{\le} 8 \int (x^4_1 +  x^4_2) \left(\dd P(x_1,x_2,y)  +  \dd Q(x_1, x_2, y) \right) \notag \\
& \overset{(b)}{=}  32 \, \ex{ |X|^4},   \label{eq:mmse_diff_c}
\end{align}
where (a) follows from \eqref{eq:ab4} and (b) follows from the fact that the marginal distributions of $X_1$ and $X_2$ are identical under $P$ and $Q$. 

For the second term on the right-hand side of \eqref{eq:mmse_diff_b}, observe that $P$ and $Q$ can be expressed as
\begin{align}
P(x_1,x_2,y) &= \frac{ P_{X , Y}(x_1,y) P_{X, Y}(x_2, y) }{ P_{Y}(y) } \notag\\
Q(x_1,x_2,y) &= \frac{ P_{X , Z}(x_1,y) P_{X, Z}(x_2, y) }{ P_{Z}(y) }. \notag
\end{align}
Letting $(X_1, X_2, Y) \sim P$, we see that the Kullback-Leibler divergence satisfies
\begin{align}
D_\mathrm{KL}( P \| Q) & = \ex{ \log\left(  \frac{ P_{X , Y}(X_1,Y) P_{X, Y}(X_2, Y) P_{Z}(Y)}{ P_{X , Z}(X_1,Y) P_{X, Z}(X_2, Y)  P_{Y}(Y) }   \right)}  \notag\\
& =2 D_\mathrm{KL}\left( P_{X,Y} \| P_{X,Z} \right)  -  D_\mathrm{KL}( P_Y \| P_Z) \notag \\
& \le 2 D_\mathrm{KL}\left( P_{X,Y} \| P_{X,Z} \right) . \label{eq:mmse_diff_d}
\end{align}
Plugging \eqref{eq:mmse_diff_c} and \eqref{eq:mmse_diff_d} back into \eqref{eq:mmse_diff_b} completes the proof of Lemma~\ref{lem:mmse_diff}.
\end{proof}

\section{Proofs of Results in Section~\ref{sec:MI_MMSE_bounds}}

\subsection{Proof of Lemma~\ref{lem:mmseX_bound}}
\label{proof:lem:mmseX_bound}

Let $Y = \sqrt{s} X+ W$ where $X \sim P_X$ and $W \sim \normal(0,1)$  are independent. 
%
Letting $X_Y$ and $X'_Y$ denote conditionally independent draws from $P_{X|Y}$ the conditional variance can be expressed as $\var(X \gmid Y) =  \frac{1}{2} \ex{ (X_Y - X'_Y)^2 \gmid Y}$. Therefore, 
\begin{align}
\left | \frac{\dd}{\dd s} \mmse_X(s)  \right| &  \overset{(a)}{=}  \ex{ \left(\var(X \gmid Y) \right)^2}  \notag\\
&= \frac{1}{4}   \ex{ \left( \ex{ (X_Y -X'_Y)^2 \gmid Y } \right)^2}  \notag\\
&\overset{(b)}{\le} \frac{1}{4}   \ex{ \left( X_Y -X'_Y \right)^4}  \notag\\
& \overset{(c)}{\le} 4\ex{ X^4},  \notag
\end{align}
where (a) follows from \eqref{eq:mmseXp}, (b) follows from Jensen's inequality and (c) follows from \eqref{eq:ab4} and the fact that $X_Y$ and $X'_Y$ have the same distribution as $X$. This completes the proof of \eqref{eq:mmseX_smooth_a}.

Next, since $\ex{W \gmid Y} = Y - \sqrt{s} \ex{ X\gmid Y}$, the conditional variance can also be expressed as
\begin{align}
\var(X \gmid Y) = \frac{1}{s} \var(W \gmid Y).  \notag
\end{align}
Using the same argument as above leads to 
\begin{align}
\left | \frac{\dd}{\dd s} \mmse_X(s)  \right| & \le  \frac{  4\ex{ W^4}}{s^2}= \frac{12}{s^2} . \notag
\end{align}
This completes the proof of \eqref{eq:mmseX_smooth_b}. 

%

\subsection{Proof of Lemma~\ref{lem:Ip_and_Ipp}} \label{proof:lem:Ip_and_Ipp}

The first order MI difference can be decomposed as
\begin{align}
I'_{m,n} & = h(Y_{m+1} \gmid Y^m, A^{m+1}) -  h(Y_{m+1} \gmid Y^m, A^{m+1}, X^n), \notag
\end{align}
where the differential entropies are guaranteed to exist because of the additive Gaussian noise. The second term  is given by the entropy of the noise,
\begin{align}
h(Y_{m+1} \gmid Y^m, A^{m+1}, X^n) = h(W_{m+1}) = \frac{1}{2} \log(2 \pi e), \notag
\end{align}
and thus does not depend on $m$. Using this decomposition, we can now write
\begin{align}
I''_{m,n} & = h(Y_{m+2} \gmid Y^{m+1}, A^{m+2})   - h(Y_{m+1} \gmid Y^m, A^{m+1}) \notag\\
& = h(Y_{m+2} \gmid Y^{m+1}, A^{m+2})   - h(Y_{m+2} \gmid Y^m, A^{m+2})   \notag\\
& =  - I(Y_{m+1}  ; Y_{m+2} \gmid Y^m, A^{m+2}), \notag
\end{align}
where the second step follows from the fact that, conditioned on $(Y^m, A^m)$, the new measurements pairs $(Y_{m+1},A_{m+1})$ and $(Y_{m+2}, A_{m+2})$ are identically distributed.

\subsection{Proof of Lemma~\ref{lem:Im_bounds}}\label{proof:lem:Im_bounds}

We first consider the upper bound in \eqref{eq:Im_bounds}. Starting with the chain rule for mutual information, we have
\begin{align}
 I(X^n ; Y^m  \mid A^m)
 &=  \sum_{i=1}^n I(X_i  ; Y^m  \mid A^m, X^{i-1} ). \label{eq:Im_UB_b}
\end{align}
Next, observe that each summand satisfies
\begin{align}
\MoveEqLeft  I(X_i  ; Y^m  \mid A^m, X^{i-1} ) \notag \\
& \overset{(a)}{\le}  I( X_{i}  ; X_{i+1}^n ,  Y^m  \mid A^m, X^{i-1}  )  \notag\\
& \overset{(b)}{=} I( X_{i}  ; Y^m  \mid A^m, X^{i-1} , X_{i+1}^n   ), \label{eq:Im_UB_c}
\end{align}
where (a) follows from the data processing inequality and (b) follows from expanding the mutual information using the chain rule and noting that  $I( X_{i}  ; X_{i+1}^n   \mid A^m,  X^{i-1})$ is equal to zero because the signal entries are independent. 

Conditioned on data  $(A^m, X^{i-1}, X_{i+1}^n)$, the mutual information provided by $Y^m$ is equivalent to the mutual information provided by the  measurement vector
\begin{align}
Y^m - \ex{ Y^m \mid A^m, X^{i-1}, X_{i+1}^n} & = A^m(i) X_i + W^m,  \notag
\end{align}
where $A^m(i)$ is the $i$-th column of $A^m$.  Moreover, by the rotational invariance of the Gaussian distribution of the noise, the linear projection of this vector in the direction of $A^m(i)$ contains all of the information about $X_i$. This projection can be expressed as
\begin{align}
\frac{ \langle A^m(i),A^m(i) X_i + W^m  \rangle  }{ \| A^m(i)\|}  =  \|A^m(i) \|   X_i +  W, \notag
\end{align}
where $W =  \langle A^m(i), W^m  \rangle  /  \| A^m(i)\|$ is Gaussian $\normal(0,1)$ and independent of $X_i$ and $A^m(i)$, by the Gaussian distribution of $W^m$. Therefore, the mutual information obeys
\begin{align}
\MoveEqLeft 
 I( X_{i}  ; Y^m  \gmid A^m, X^{i-1} , X_{i+1}^n   ) \notag\\
& = I\left(X_i ;  \|A^m(i) \|   X_i +  W \; \middle | \; \| A^m(i) \|  \right ) \notag \\
& = \ex{ I_X\left( \|A^m(i) \|^2 \right ) } \notag\\
& = \ex{ I_X\left( \tfrac{1}{n} \chi^2_m \right ) } , \label{eq:Im_UB_d}
\end{align}
where the last step follows from the fact that the entries of $A^m(i)$ are i.i.d.\ Gaussian $\normal(0, 1/n)$.  
Combining \eqref{eq:Im_UB_b}, \eqref{eq:Im_UB_c}, and \eqref{eq:Im_UB_d}  gives the upper bound in \eqref{eq:Im_bounds}.

Along similar lines, the lower bound in \eqref{eq:Mm_bounds} follows from 
\begin{align}
M_{m,n} & \triangleq \frac{1}{n} \mmse( X^{n} \gmid Y^m, A^m)   \notag\\
& \overset{(a)}{=} \mmse( X_{n} \gmid Y^m, A^m) \notag \\
& \overset{(b)}{\ge }  \mmse(X_n \gmid Y^m, A^m, X^{n-1} ) \notag\\
& \overset{(c)}{= }   \mmse(X_n \gmid A^m(n) X_n + W^m , A^m(n)  ) \notag\\
& \overset{(d)}{= }   \mmse\left(X_n \; \middle | \;  \|A^m(n) \|   X_i +  W,  \| A^m(n) \|  \right )  \notag\\
& \overset{(e)}{=}  \ex{ \mmse_X\left( \tfrac{1}{n} \chi^2_{m} \right)},  \notag
\end{align}
where (a) follows from the fact that the distributions of the columns of $A^m$ and entries of $X^n$ are permutation invariant, (b) follows from the data processing inequality for MMSE, (c) follows from the independence of the signal entries, (d) follows from the Gaussian distribution of the noise $W^m$, and (e) follows from the distribution of $A^m(n)$.

We now turn our attention to the lower bound in  \eqref{eq:Im_bounds}. Let the QR decomposition of $A^m$ be given by 
\begin{align}
 A^m  = Q R , \notag
\end{align}
where $Q$ is an $m \times n$ orthogonal matrix and $R$ is an $m \times n$ upper triangular matrix. Under the assumed Gaussian distribution on $A^m$, the nonzero entries of $R$ are independent random variables with \cite[Theorem~2.3.18]{gupta:1999}:
\begin{align}
R_{i,j} \sim \begin{cases}
\frac{1}{n} \chi^2_{m-i+1}, & \text{if $i = j$}\\
\normal\left(0, \frac{1}{n} \right),  & \text{if $i < j$}. 
\end{cases} \label{eq:R_dist} 
\end{align}
Let the rotated measurements and noise be defined by
\begin{align}
\widetilde{Y}^m & = Q^T Y^ m, \qquad \widetilde{W}^m = Q^T W^m  \notag
\end{align}
and observe that
\begin{align}
\widetilde{Y}^m & =R X^n  + \widetilde{W}^m.  \notag
\end{align}
By the rotational invariance of the Gaussian distribution of the noise $W^m$,  the rotated noise $\widetilde{W}^m$ is Gaussian $\normal(0,I_{m \times m} )$ and independent of everything else. Therefore, only the first $d\triangleq \min(m,n)$ measurements provide any information about the signal. Using this notation, the mutual information can be expressed equivalently as
\begin{align}
  I(X^n ; Y^m  \mid A^m) & = I(X^n ; \widetilde{Y}^d  \mid R )  \notag\\
 & = \sum_{i=1}^d I(X^n ;  \widetilde{Y}_k \mid  \widetilde{Y}_{k+1}^d  , R),  \label{eq:Im_LB_b}
\end{align}
where the second step follows from the chain rule for mutual information.

To proceed, note that the measurements $\widetilde{Y}_{k}^d$ are independent of the first part of the signal  $X^{k-1}$, because $R$ is upper triangular. Therefore, for all $1 \le k \le d$, 
\begin{align}
 I(X^n ;  \widetilde{Y}_k \mid  \widetilde{Y}_{k+1}^d  , R)  \notag
 & =  I(X_{k}^n  ;  \widetilde{Y}_k \mid  \widetilde{Y}_{k+1}^d  , R) \notag \\
 & \overset{(a)}{\ge}  I(X_{k}  ;  \widetilde{Y}_k \mid  \widetilde{Y}_{k+1}^d, X_{k+1}^n   , R)  \notag\\
 & \overset{(b)}{=}  I(X_{k}  ;  R_{k,k} X_k  + \widetilde{W}_k  \mid   R)  \notag\\
 & \overset{(c)}{=}  \ex{ I_X( \tfrac{1}{n} \chi^2_{m-k+1})},\label{eq:Im_LB_c}
\end{align}
where (a)  follows from the data processing inequality, (b) follows from the fact that $R$ is upper triangular and the independence of the  signal entries, and  (c) follows from \eqref{eq:R_dist}.  Combining \eqref{eq:Im_LB_b} and \eqref{eq:Im_LB_c} gives the upper bound in \eqref{eq:Im_bounds}

For the upper bound in \eqref{eq:Mm_bounds},  suppose that $m \ge n$ and let $\widetilde{Y}^n$ be the rotated measurements defined above. Then we have
\begin{align}
M_{m,n}& \triangleq \frac{1}{n} \mmse( X_{n} \mid Y^m, A^m)  \notag \\
& \overset{(a)}{=} \mmse( X_{n} \gmid Y^m, A^m)  \notag\\
& \overset{(b)}{=}  \mmse(X_n \gmid \widetilde{Y}^m , R) \notag\\
& \overset{(c)}{\le}   \mmse(X_n \gmid \widetilde{Y}_n , {R}_{n,n} ) \notag\\
& \overset{(d)}{=}    \mmse(X_n \gmid  R_{n,n} X_n + \widetilde{W}_n  , {R}_{n,n} ) \notag\\
& \overset{(e)}{=}  \ex{ \mmse_X\left( \tfrac{1}{n} \chi^2_{m-n+1} \right)} , \notag
\end{align}
where (a) follows from the fact that the distributions of the columns of $A^m$ and entries of $X^n$ are permutation invariant, (b) follows from the fact that multiplication by $Q$ is a one-to-one transformation, (c) follows from the data processing inequality for MMSE, (d)  follows from the fact that $R$ is upper triangular with $m \ge n$, and (e) follows from \eqref{eq:R_dist}.

\subsection{Proof of Lemma~\ref{lem:Im_bounds_gap} }\label{proof:lem:Im_bounds_gap}

Let $U = \frac{1}{n} \chi^2_{m-n+1}$ and $V= \frac{1}{n} \chi^2_{n+1}$ be independent scaled chi-square random variables and let $Z = U+V$. Using this notation, the lower bound in \eqref{eq:Im_bounds} satisfies
\begin{align}
\frac{1}{n} \sum_{k=1}^{n} \ex{ I_X\left( \tfrac{1}{n} \chi^2_{m-k+1} \right)}  & \ge \ex{ I_X(U)},  \label{eq:Im_bounds_gap_b} 
\end{align}
where we have used the fact that the mutual information function is non-decreasing. Moreover, by \eqref{eq:IX_smooth}, we have
\begin{align}
\ex{  I_X(Z)}  - \ex{ I_X(U)} & \le \frac{1}{2} \ex{ \left( Z/U -1 \right)_+} \notag \\
& = \frac{1}{2} \ex{V/U } \notag\\
& = \frac{1}{2} \frac{n + 1}{m - n - 1 }. \label{eq:Im_bounds_gap_c} 
\end{align}

Next, observe that $Z$ has a scaled chi-squared distribution $Z \sim \frac{1}{n} \chi^2_{m+2}$, whose inverse moments are given by
\begin{align}
\ex{ Z^{-1}}  = \frac{n}{m}, \qquad  \var( Z^{-1})  = \frac{2 n^2  }{m^2 (m-2)} . \notag
\end{align}
Therefore, by \eqref{eq:IX_smooth},  we have
\begin{align}
 I_X(\tfrac{m}{n} ) - \ex{ I_X(Z)}
& \le \frac{1}{2} \ex{ \left( \frac{m/n}{Z}  -1 \right)_+} \notag \\
& \le \frac{m}{2 n} \ex{ \left| \frac{1}{Z}  -\frac{n}{m}  \right|}  \notag\\
& \le \frac{m}{2 n } \sqrt{ \var(Z^{-1})} \notag\\
& = \frac{1}{2  } \sqrt{ \frac{2}{ m-2}}, \label{eq:Im_bounds_gap_d} 
\end{align}
where the second and third steps follow from Jensen's inequality. Combining \eqref{eq:Im_bounds_gap_b}, \eqref{eq:Im_bounds_gap_c} and  \eqref{eq:Im_bounds_gap_d} completes the proof of Inequality \eqref{eq:Im_bounds_gap}. 

We use a similar approach for the MMSE. Note that
\begin{align}
 \ex{  \mmse_X(U)}  & = \ex{ \mmse_X\left( \tfrac{1}{n} \chi^2_{m - n + 1}\right) }  \notag \\
 \ex{\mmse_X(Z)}  & \le \ex{ \mmse_X\left( \tfrac{1}{n} \chi^2_{m}\right) } , \label{eq:Mm_bounds_gap_b} 
\end{align}
where the second inequality follows from the monotonicity of the MMSE function. By Lemma~\ref{lem:mmseX_bound}, the MMSE obeys
\begin{align}
\ex{  \mmse_X(U)}  - \ex{ \mmse_X(Z)} 
& \le 12 \, \ex{ \left| \frac{ 1}{U} -\frac{1}{Z}  \right|}  \notag\\
& = 12 \left( \ex{U^{-1}  } - \ex{ Z^{-1} } \right)  \notag\\
& =12 \frac{n + 1}{m - n - 1 }. \label{eq:Mm_bounds_gap_c} 
\end{align}
Moreover, 
\begin{align}
\left| \ex{  \mmse_X(Z)}  -  \mmse_X(\tfrac{m}{n} )\right| 
& \le 12 \, \ex{ \left| \frac{ 1}{Z} -\frac{n}{m}  \right|}  \notag\\
& \le 12 \sqrt{ \var(Z^{-1})}  \notag \\
& = \frac{1}{2  } \sqrt{ \frac{2}{ m-2}} \label{eq:Mm_bounds_gap_d} 
\end{align}
where the second and third steps follow from Jensen's inequality. Combining \eqref{eq:Mm_bounds_gap_b}, \eqref{eq:Mm_bounds_gap_c} and  \eqref{eq:Mm_bounds_gap_d} completes the proof of Inequality \eqref{eq:Mm_bounds_gap}. 

\subsection{Proof of Lemma~\ref{lem:IR_boundary}}\label{proof:lem:IR_boundary}

The upper bound in \eqref{eq:IRS_bounds} follows from noting that
\begin{align}
\IR(\delta) & \triangleq \min_{z \ge 0}   R( \delta, z) \le    R( \delta, 0) =  I_X(\delta) . \notag
\end{align}
For the lower bound, observe that the replica-MI function can also be expressed in terms of the replica-MMSE function as
\begin{align}
\IR(\delta) & =  R( \delta, \MR(\delta)).  \notag
\end{align}
Since the term $ [ \log\left( 1+ z\right) - \frac{ z}{ 1 +  z}]$ in the definition of $R(\delta, z)$ is non-negative, we have the lower bound
\begin{align}
\IR(\delta) & \ge I_X\left( \frac{ \delta}{ 1 + \MR(\delta)} \right) . \label{eq:IR_LB_b} 
\end{align}

Next, we recall that the replica-MMSE function satisfies the fixed-point equation
\begin{align}
\MR(\delta)    = \mmse_X\left( \frac{\delta}{1+\MR(\delta) } \right) .   \label{eq:IR_LB_c} 
\end{align}
Also, for any signal distribution $P_X$ with, the MMSE function satisfies the upper bound $\mmse_X(s) \le 1/s$ \cite[Proposition 4]{guo:2011}. Therefore, 
\begin{align}
\MR(\delta)   \le  \frac{ 1 + \MR(\delta)}{\delta} . \notag
\end{align}
For $\delta > 1$, rearranging the terms leads to the upper bound  
\begin{align}
\MR(\delta)  \le \frac{ 1}{  \delta - 1} .  \label{eq:IR_LB_d} 
\end{align}
Combining \eqref{eq:IR_LB_b}  and \eqref{eq:IR_LB_d} with the fact that the mutual information is non-decreasing yields
\begin{align}
\IR(\delta) & \ge I_X\left( \frac{ \delta}{ 1 + \frac{1}{1 -\delta} } \right)  = I_X(\delta - 1).  \notag
\end{align}

Lastly, we consider the bounds in \eqref{eq:MRS_bounds}. Combining \eqref{eq:IR_LB_c}  and \eqref{eq:IR_LB_d} with the fact that the MMSE is non-increasing yields
\begin{align}
\MR(\delta) & \le \mmse_X\left( \frac{ \delta}{ 1 + \frac{1}{1 -\delta} } \right)  = \mmse_X(\delta - 1) . \notag
\end{align}
Alternatively, starting with \eqref{eq:IR_LB_c} and using the non-negativity of $\MR$  leads to the lower bound. This completes the proof of Lemma~\ref{lem:IR_boundary}.

\subsection{Proof of Lemma~\ref{lem:Im_var}}\label{proof:lem:Im_var}

To lighten notation, we will write $I(A)$ in place of $I_{m,n}(A^m)$. By the Gaussian Poincar\'e inequality \cite[Theorem 3.20]{boucheron:2013}, the variance satisfies 
\begin{align}
\var( \ex{ I(A)}) & \le \ex{ \left\| \nabla_{A} I(A) \right\|^2_F}, \label{eq:var_matrix_b} 
\end{align}
where $\nabla_{A}$ is the gradient operator with respect to $A$. Furthermore, by the multivariate I-MMSE relationship \cite[Theorem 1]{palomar:2006}, the gradient of the mutual information with respect to $A$ is given by
\begin{align}
\nabla_{A} I(A)  =  A  K(A), \label{eq:var_matrix_c} 
\end{align}
where $K(A) = \ex{ \cov(X^n | Y^m, A) \mid A}$ is the expected posterior covariance matrix as a function of the matrix $A$.

Next, we recall some basic facts about matrices that will allow us to bound the magnitude of the gradient. 
 For symmetric matrices $U$ and $V$ the notation $U \preceq V$ mean that $V - U$ is positive semidefinite.  If $U$ and $V$ are positive definite with  $U \preceq V$ then $U^2 \preceq V^2$,  by \cite[Theorem~7.7.3]{horn:2012}, and thus
 \begin{align}
 \|S U\|^2_F  = \gtr( S U^2 S^T) \le \gtr( S V^2 S^T)  \notag
 \end{align}
 for every matrix $S$. 

Since the since the conditional expectation minimizes the squared error, the expected covariance matrix can be upper bounded in terms of the MSE matrix associated with the optimal linear estimator:
\begin{align}
K(A) \preceq \left(I_{n \times n}  + \var(X) A^T  A^T\right)^{-1} \var(X). \notag
\end{align}
Combining this inequality with the arguments given above and the fact that $\left(I_{n \times n}  + \var(X) A^T  A^T\right)^{-1} \preceq I_{n \times n}$ leads to
\begin{align}
\left\| A  K(A) \right\|^2_F&  \le  \left( \var(X)\right)^2 \gtr\left( A A^T \right). \label{eq:var_matrix_d} 
\end{align}
Alternatively, for $m > n$, the matrix $A^T A$ is full rank with probability one and thus $K(A) \preceq \left(A^T A\right)^{-1}$ with probability one. 
This leads to the upper bound
\begin{align}
\left\| A  K(A) \right\|^2_F\le
  \gtr\left( \left( A^T\! A\right)^{-1} \right). \label{eq:var_matrix_e} 
\end{align}

To conclude, note that under the assumed Gaussian distribution on $A$,
\begin{align}
 \ex{ \gtr\left( A A^T \right)} & = m  \notag\\
  \ex{ \gtr\left( \left(A^T A\right)^{-1} \right)} & = \frac{n^2}{(m-n-1)_+},  \notag
\end{align}
where the second expression follows from \cite[Theorem~3.3.16]{gupta:1999}. Combining these expectations with \eqref{eq:var_matrix_b} and \eqref{eq:var_matrix_c} and the upper bounds \eqref{eq:var_matrix_d} and \eqref{eq:var_matrix_e} completes the proof.

\subsection{Proof of Lemma~\ref{lem:IMI_var}}\label{proof:lem:IMI_var}

To simplify notation, the mutual information density is denoted by $\IMI = \imath(X^n ;Y^m \gmid A^m)$. Note that $\IMI$ can be expressed as a function of the random tuple $(X^n, W^m, A^m)$ according to
\begin{align}
\IMI &=  \log\left( \frac{f_{Y^m|X^n,A^m}( A^m X^n  +W^m  \gmid   X^n,A^m )}{f_{Y^m|A^m}(A^m X^n  +W^m \gmid  A^m )}    \right). \label{eq:IMI_function} 
\end{align}

Starting with the law of total variance \eqref{eq:law_tot_var}, we see that
\begin{align}
 \var(\IMI ) & = \ex{ \var( \IMI \gmid A^m)} + \var( \ex{ \IMI  \gmid A^m}). \label{eq:var_IMI_decomp} 
\end{align}
The second term on the right-hand side of \eqref{eq:var_IMI_decomp} is variance with respect to the matrix, which is bounded by Lemma~\ref{lem:Im_var}.  

The first term on the right-hand side of \eqref{eq:var_IMI_decomp} is the variance with respect to the signal and the noise. Since the entries of $X^n$ and $W^m$ are independent, the variance can be bounded using the the Efron-Stein inequality  \cite[Theorem 3.1]{boucheron:2013}, which yields
\begin{align}
\ex{ \var( \IMI \gmid  A^m )}  & \le  \sum_{i=1}^n \ex{ \left(  \IMI - \bEx_{X_i}\left[ \IMI  \right] \right)^2}  \notag \\
& \quad  + \sum_{ i = 1}^m \ex{\left( \IMI  - \bEx_{W_i}\left[ \IMI  \right] \right)^2   } . \label{eq:var_IMI_decomp_b} 
\end{align} 

At this point, Lemma~\ref{lem:IMI_var} follows from combining \eqref{eq:var_IMI_decomp}, Lemma~\ref{lem:Im_var},   and \eqref{eq:var_IMI_decomp_b} with the following inequalities 
\begin{align}
\ex{ \left(  \IMI - \bEx_{X_i}\left[ \IMI  \right] \right)^2} & \le 12 \left( 1 +  \left(1 + \sqrt{  \tfrac{m}{n}} \right)  B^{\frac{1}{4}} \right)^4 \label{eq:var_signal}\\
\ex{\left( \IMI  - \bEx_{W_i}\left[ \IMI \right] \right)^2} & \le \sqrt{B}  \label{eq:var_noise}. 
\end{align}
Inequalities \eqref{eq:var_signal} and \eqref{eq:var_noise} are proved in the following subsections.

\subsubsection{Proof of Inequality \eqref{eq:var_signal}}
Observe that the expectation $\ex{ \left(  \IMI - \bEx_{X_i}\left[ \IMI  \right]\right)^2 }$ is identical for all $i \in [n]$ because the distribution on the entries in $X^n$ and the columns of $A^m$ are permutation invariant.  Throughout this proof, we focus on the variance with respect to the last signal entry $X_n$. 

Starting with the chain rule for mutual information density, we obtain the decomposition 
\begin{align}
 \IMI & =\imath(X_n  ;Y^m \gmid A^m )  +  \imath(X^{n-1}  ; Y^m \gmid X_n , A^m  ). \label{eq:var_sig_b}
\end{align}
The second term is independent of $X_n$, and thus does not contribute to the conditional variance. To characterize the first term, we introduce a transformation of the data that isolates the effect of the $n$th signal entry. Let $Q$  be drawn uniformly from the set of $m \times m$ orthogonal matrices whose last row is the unit vector in the direction of the last column of $A^{m}$, and let the \textit{rotated} data be defined according to
\begin{align}
\widetilde{Y}^m = Q Y^m, \qquad \widetilde{A}^m  = Q A^m , \qquad \widetilde{W}^m = Q W^m.  \notag
\end{align} 
By construction,  the last column of $\widetilde{A}^m$ is zero everywhere except for the last entry with
\begin{align}
\widetilde{A}_{i,n} = \begin{cases}
0, &  \text{if $1 \le i \le m-1$}\\
\| A^{m}(n)  \|, & \text{if $i = m$},
\end{cases}  \notag
\end{align}
where $A^m(n)$ denotes the $n$-th column of $A^m$.  Furthermore, by the rotational invariance of the Gaussian distribution of the noise, the rotated noise $\widetilde{W}^m$ has the same distribution as $W^m$ and is independent of everything else. 

Expressing the  mutual information density in terms of the rotated data leads to the further decomposition
\begin{align}
 \imath(X_n ;Y^m \gmid A^m ) & \overset{(a)}{=}   \imath(X_n ;\widetilde{Y}^m \gmid \widetilde{A}^m )  \notag\\
 & \overset{(b)}{=}  \imath(X_n ;\widetilde{Y}_m \gmid \widetilde{Y}^{m-1}, \widetilde{A}^m )  \notag\\
 & \quad  + \imath(X_n ;\widetilde{Y}^{m-1} \gmid \widetilde{A}^m ) \notag \\
 & \overset{(c)}{=}  \imath(X_n ;\widetilde{Y}_m \gmid  \widetilde{Y}^{m-1}, \widetilde{A}^m ), \label{eq:var_sig_c}
\end{align}
where (a) follows from fact that multiplication by $Q$ is a one-to-one transformation, (b)  follows from the chain rule for mutual information density, and (c) follows from fact that the first $m-1$ entries of $\widetilde{Y}^m$ are independent of $X_n$. 

To proceed, we introduce the notation $U = (\widetilde{Y}^{m-1} ,\widetilde{A}^{m-1})$.  Since the noise is independent of the measurements,  the conditional density of $\widetilde{Y}_{m}$ given $(X^n, U, \widetilde{A}_{m})$ obeys
\begin{align}
 \MoveEqLeft f_{\widetilde{Y}_m   \mid X^n,  U, \widetilde{A}_m} (\widetilde{y}_m   \mid x^n,  u, \widetilde{a}_m)   \notag\\
 & =  f_{\widetilde{W}_m} (\widetilde{y}_m - \langle \widetilde{a}_{m} , x^n \rangle)  \notag\\
 & = \frac{1}{\sqrt{ 2 \pi}} \exp\left( - \frac{1}{2} \left(\widetilde{y}_m - \langle \widetilde{a}_{m} , x^n \rangle \right)^2  \right) . \notag
\end{align}
Starting with  \eqref{eq:var_sig_c},  the mutual information density can now be expressed as 
\begin{align}
 \imath(X_n ;Y^m \mid A^m ) 
& = \imath\left(X_n  ; \widetilde{Y}_m   \mid U, \widetilde{A}_{m}  \right)  \notag\\
 & = \log \left( \frac{ f_{\widetilde{Y}_m   \mid X^n,  U, \widetilde{A}_m} (\widetilde{Y}_m   \mid X^n,  U, \widetilde{A}_m)  }{  f_{\widetilde{Y}_m   \mid U, \widetilde{A}_m} (\widetilde{Y}_m   \mid U, \widetilde{A}_m)}\right), \notag\\
  & = g(U, \widetilde{Y}_{m}, \widetilde{A}_{m}) - \frac{1}{2} \widetilde{W}_m^2,  \label{eq:var_sig_d}
\end{align}
where
\begin{align}
g(U, \widetilde{Y}_m, A_{m}) &=   \log\left( \frac{ ( 2 \pi)^{-\frac{1}{2}} }{  f_{\widetilde{Y}_m   \mid U,  \widetilde{A}_m} (\widetilde{Y}_m   \mid U, \widetilde{A}_m)} \right). \notag
\end{align}
Furthermore, by \eqref{eq:var_sig_b} and \eqref{eq:var_sig_d}, 
the difference between $\IMI$ and its conditional expectation with respect to $X_n$ is given by
\begin{align}
 \IMI  - \bEx_{X_n}\left[\IMI  \right]  & = g(U, \widetilde{Y}_m, \widetilde{A}_{m}) -  \bEx_{X_n} \left[ g(U, \widetilde{Y}_m, \widetilde{A}_{m}) \right]. \notag
\end{align}
Squaring both sides and taking the expectation yields 
\begin{align}
\MoveEqLeft \ex{ \left( \IMI  - \bEx_{X_n}\left[\IMI  \right] \right)^2} \notag \\
 & = \ex{\left(  g(U, \widetilde{Y}_m, \widetilde{A}_{m})  -  \bEx_{X_n} \left[  g(U, \widetilde{Y}_m, \widetilde{A}_{m}) \right]\right)^2} \notag \\
& \overset{(b)}{\le} \ex{ g^2(U, \widetilde{Y}_m, \widetilde{A}_{m})  }, \label{eq:var_sig_e}
\end{align}
where (b) follows from the law of total variance \eqref{eq:law_tot_var}.

Next, we bound the function $g(u, \widetilde{y}_m, \widetilde{a}_m)$. Let $X^n_u$ be drawn according to the conditional distribution $P_{X^n \mid U = u}$. Then, the conditional density of $\widetilde{Y}_{m}$ given $(U, \widetilde{A}_{m})$ is given by
\begin{align}
 \MoveEqLeft 
  f_{\widetilde{Y}_m   \mid U, \widetilde{A}_m} (\widetilde{y}_m   \mid  u, \widetilde{a}_m)  \notag\\
 &=  \ex{ f_{\widetilde{Y}_m   \mid X^n,  U, \widetilde{A}_m} (\widetilde{y}_m   \mid X_u^n,  u, \widetilde{a}_m)  }  \notag\\
 & = \ex{  \frac{1}{\sqrt{ 2 \pi}} \exp\left( - \frac{1}{2} \left(\widetilde{y}_m - \langle a_{m} , X_u^n \rangle \right)^2  \right)} . \notag
\end{align}
 Since the conditional density obeys the upper bound
\begin{align}
  f_{\widetilde{Y}_m   \mid U, \widetilde{A}_m} (\widetilde{y}_m   \mid  u, \widetilde{a}_m) \le ( 2 \pi)^{-\frac{1}{2}}, \notag
\end{align}
we see that $g(u, \widetilde{y}_m, \widetilde{a}_m)$ is nonnegative.  Alternatively, by Jensen's inequality, 
\begin{align}
  \MoveEqLeft f_{\widetilde{Y}_m   \mid U, \widetilde{A}_m} (\widetilde{y}_m   \mid  u, \widetilde{a}_m) \notag \\
&   \ge   \frac{1}{\sqrt{ 2 \pi}} \exp\left( - \frac{1}{2} \ex{   \left(\widetilde{y}_m - \langle a_{m} , X_u^n \rangle \right)^2}   \right), \notag
\end{align}
and thus
\begin{align}
g(u, \widetilde{y}_m, \widetilde{a}_m) \le  \frac{1}{2} \ex{   \left(\widetilde{y}_m - \langle a_{m} , X_u^n \rangle \right)^2}.  \notag
\end{align}
Using these facts leads to the following inequality:
 \begin{align}
 g^2(u, \widetilde{y}_m , \widetilde{a}_m) 
&  \le  \frac{1}{4}  \left( \bEx_{X^n_u}\left[  \left (y_m    - \langle \widetilde{a}_m,  X_u^n \rangle  \right )^2  \right] \right)^2 \notag\\
& \overset{(a)}{\le}  \frac{1}{4} \bEx_{X^n_u}\left[  \left (\widetilde{y}_m   - \langle \widetilde{a}_m, X_u^n \rangle  \right )^4 \right] \notag\\
&  \overset{(b)}{\le} 2  \left(\widetilde{y}_m \right)^4 +  2\bEx_{X^n_u}\left[ \left ( \langle \widetilde{a}_m, X_u^n \rangle \right)^4 \right] , \notag
 \end{align}
where (a) follows from Jensen's inequality and (b) follows from \eqref{eq:ab4}.  Taking the expectation with respect to the random tuple $(U, \widetilde{Y}_m , \widetilde{A}_m)$, we can now write
  \begin{align}
 \ex{ g^2(U, \widetilde{Y}_m , \widetilde{A}_m) }
& \le 2  \ex{  \left( \widetilde{Y}_m \right)^4}  + 2 \ex{ \left ( \langle \widetilde{A}_m, X_U^n \rangle \right)^4 } \notag\\
&\overset{(a)}{=} 2  \ex{  \left( \widetilde{Y}_m \right)^4}  + 2 \ex{ \left ( \langle \widetilde{A}_m, X^n \rangle \right)^4 }, \label{eq:var_sig_f}
 \end{align}
where (a) follows from the fact that $X^n_U$ has the same distribution as $X^n$ and is independent of $\widetilde{A}_m$. To upper bound the first term, observe that
\begin{align}
\ex{ \left(  \widetilde{Y}_m\right)^4 } &  =   \ex{ \left(\widetilde{W}_m +  \sum_{i=1}^n \widetilde{A}_{m,i} X_i  \right)^4}  \notag  \\
 & \overset{(a)}{\le}  \left(  \left( \ex{ \widetilde{W}_m^4} \right)^{\frac{1}{4}} +   \left( \ex{\left( \langle \widetilde{A}_m, X^n \rangle \right)^4 }\right)^{\frac{1}{4}}  \right)^4 \notag \\
  & =  \left(  3^{\frac{1}{4}} +   \left( \ex{\left( \langle \widetilde{A}_m, X^n \rangle \right)^4 }\right)^{\frac{1}{4}}  \right)^4  \label{eq:var_sig_g}
 \end{align}
 where (a) follows from Minkowski's inequality \eqref{eq:Minkowski}. To bound the fourth moment of $\langle \widetilde{A}_m, X^n \rangle$, we use the fact that the entries of the rotated measurement vector $\widetilde{A}_m$ are independent with 
\begin{align}
\widetilde{A}_{m,i} \sim
\begin{cases}
 \normal(0,\frac{1}{n} ), &  \text{if $1 \le i \le n-1$}\\
\frac{1}{\sqrt{n}} \chi_m , & \text{if $i  = n$},
\end{cases} \notag
\end{align}
where $\chi_m$ denotes a chi random variable with $m$ degrees of freedom. Thus, we have
\begin{align}
\MoveEqLeft \left( \ex{\left( \langle \widetilde{A}_m,X^n \rangle \right)^4 }\right)^{\frac{1}{4}}   \notag \\
& \overset{(a)}{\le}     \left( \ex{ \left(  \sum_{i=1}^{n-1}  \widetilde{A}_{m,i} X_i \right)^4 }\right)^{\frac{1}{4}} +  \left(\ex{ \left(  \widetilde{A}_{m,n} X_n \right)^4 }\right)^{\frac{1}{4}} \notag \\
& \overset{(b)}{=}   \left( \frac{3}{n^2} \ex{ \left\| X^{n-1} \right \|^4} \right)^{\frac{1}{4}}   + \left( \frac{m(m+2)}{n^2} \ex{X_n^4} \right)^{\frac{1}{4}} \notag \\
& \overset{(c)}{\le}  \left( \frac{3(n-1)^2 }{n^2}  B  \right)^{\frac{1}{4}}   + \left( \frac{m(m+2)}{n^2} B \right)^{\frac{1}{4}}  \notag\\
& \le  \left(  1+   \sqrt{ \frac{m}{n}} \right)  \left( 3 B \right)^{\frac{1}{4}},  \label{eq:var_sig_h}
 \end{align}
where (a) follows from Minkowski's inequality \eqref{eq:Minkowski}, (b) follows from the distribution on $\widetilde{A}^m$ and (c) follows from Assumption 2 and  inequality \eqref{eq:LpBound_1} applied to $\|X^{n-1}\|^4 = \left( \sum_{i=1}^{n-1} X_i^2 \right)^2$.  

Finally, combining \eqref{eq:var_sig_e}, \eqref{eq:var_sig_f}, \eqref{eq:var_sig_g}, and \eqref{eq:var_sig_h} leads to
\begin{align}
\ex{ \left( \IMI  - \bEx_{X_n}\left[\IMI  \right] \right)^2}
& \le  2 \left( 3^\frac{1}{4}  +  \left(  1+   \sqrt{ \frac{m}{n}} \right) (3B)^{\frac{1}{4}} \right)^4 \notag\\
& \quad + 2\left(  1+   \sqrt{ \frac{m}{n}} \right)^4  \left( 3 B \right) \notag\\
& \le 12 \left( 1 +  \left(  1+   \sqrt{ \frac{m}{n}} \right) B^{\frac{1}{4}} \right)^4. \notag
\end{align}
This completes the proof of Inequality \eqref{eq:var_signal}.

\subsubsection{Proof of Inequality \eqref{eq:var_noise}} 
Observe that the expectation $\ex{\left( \IMI_k  - \bEx_{W_m}\left[ \IMI_k \right] \right)^2 }$ is identical for all $k \in [m]$ because the distributions on the entries in $W^n$ and the rows of $A^m$ are permutation invariant. Throughout this proof, we focus on the variance with respect to the last noise entry $W_m$. 

Recall from \eqref{eq:IMI_function},  $\IMI$ can be expressed as a function of $(X^n, W^m, A^m)$. By the Gaussian Poincar\'e inequality \cite[Theorem 3.20]{boucheron:2013}, the variance with respect to $W_m$ obeys
\begin{align}
\ex{\left( \IMI  - \bEx_{W_m}\left[ \IMI \right] \right)^2   } & \le \ex{ \left(  \frac{\partial}{\partial W_m}  \IMI    \right)^2}, \label{eq:var_noise_b}
\end{align} 
where  $\frac{\partial}{\partial W_m}  \IMI$ denotes the partial derivative of the right-hand side of \eqref{eq:IMI_function} evaluated at the point $(X^n, W^m, A^m)$. 

To compute the partial derivative, observe that by the chain rule for mutual information density, 
\begin{align}
\IMI & =\imath(X^n  ;Y^{m-1} \gmid A^m) +  \imath(X^n ; Y_m \gmid Y^{m-1}  , A^m). \notag
\end{align}
In this decomposition, the first term on the right-hand side is independent of $W_m$. The second term can be decomposed as
\begin{align}
\MoveEqLeft  \imath(X^n ; Y_m \gmid Y^{m-1}  , A^m)  \notag\\
& = \log\left( \frac{f_{Y_m \mid X^n, Y^{m-1}, A^m}(Y_m \mid X^n, Y^{m-1}, A^m)}{f_{Y_m \mid Y^{m-1}, A^m}(Y_m \mid  Y^{m-1}, A^m)}    \right) \notag\\
& =   \log\left(\frac{ f_{W_m}(Y_m - \langle A_m, X^n \rangle)}{f_{Y_m \mid Y^{m-1}, A^m}(Y_m \mid  Y^{m-1}, A^m)  } \right) \notag \\
& =   \log\left(\frac{ f_{W_m}(Y_m)}{f_{Y_m \mid Y^{m-1}, A^m}(Y_m \mid  Y^{m-1}, A^m)  } \right)  \notag\\
& \quad + Y_m \langle A_m, X^n \rangle   - \frac{1}{2} \left(\langle A_m, X^n \rangle \right)^2.  \notag
\end{align}
Note that the first term on the right-hand side is the negative of the log likelihood ratio. The partial derivative of this term with respect to $Y_m$ can be expressed in terms of the conditional expectation \cite{esposito:1968}:
\begin{align}
\MoveEqLeft \frac{\partial}{\partial Y_m}   \log\left(\frac{ f_{W_m}(Y_m)}{f_{Y_m \mid Y^{m-1}, A^m}(Y_m \mid  Y^{m-1}, A^m)  } \right) \notag\\
& =  -  \ex{  \langle A^m, X^n  \rangle  \mid Y^{m}, A^m}. \notag
\end{align}
Recalling that $Y_m = \langle A_m , X^n \rangle + W_m$, the partial derivative with respect to $W_m$ can now be computed directly as
\begin{align}
 \frac{\partial}{\partial W_m} \IMI
& =  \frac{\partial}{\partial Y_m}  \imath(X^n ; Y_m \mid Y^{m-1}  , A^m) \notag\\
& =   \langle A_m, X^n  \rangle -  \ex{  \langle A^m, X^n  \rangle  \mid Y^{m}, A^m}  \notag\\
&= \langle A_m , X^n - \ex{X^n \mid Y^m, A^m} \rangle. \notag
\end{align} 
Thus, the expected squared magnitude obeys
\begin{align}
 \ex{ \left(  \frac{\partial}{\partial W_i}  \IMI_m  \right)^2}
  &=   \ex{ \left( \langle A_m , X^n - \ex{X^n \mid Y^m, A^m} \rangle  \right)^2} \notag\\
  & =  \ex{ \var( \langle A_m, X^n \rangle \mid Y^m, A^m)} \notag\\
 &  \overset{(a)}{\le}  \ex{ \var( \langle A_m, X^n \rangle \mid  A_m)} \notag\\
  & =  \ex{ A_m^T \cov(X^n) A_m}   \notag\\
  & \overset{(b)}{=}   \ex{X^2}  \notag\\
  & \overset{(c)}{\le} \sqrt{B},  \label{eq:var_noise_c}
\end{align} 
where (a) follows from the law of total variance \eqref{eq:law_tot_var} and (b) follows from the fact that $\ex{A_m A_m^T} = \frac{1}{n} I_{n \times n}$, and (c) follows from Jensen's inequality and Assumption 2. Combining \eqref{eq:var_noise_b}  and \eqref{eq:var_noise_c} completes the proof of Inequality \eqref{eq:var_noise}.

\section{Proofs of Results in Section~\ref{proof:thm:I_MMSE_relationship} }

\subsection{Proof of Lemma~\ref{lem:moment_bounds}.} \label{proof:lem:moment_bounds}
The squared error obeys the upper bound
\begin{align}
\left|\cE_{m,n} \right|^2  & = \left( \frac{1}{n} \sum_{i=1}^n \left( X_i - \ex{ X_i \gmid Y^m, A^m}\right)^2 \right)^2  \notag\\
& \overset{(a)}{\le}  \frac{1}{n} \sum_{i=1}^n \left( X_i - \ex{ X_i \gmid Y^m, A^m}\right)^4 \notag\\
& \overset{(b)}{\le}   \frac{8}{n} \sum_{i=1}^n  \left[\left| X_i \right|^4  + \left|  \ex{ X_i \gmid Y^m, A^m}\right|^4 \right] \notag\\
&\overset{(c)}{\le}   \frac{8}{n} \sum_{i=1}^n  \left[\left| X_i \right|^4  + \ex{ \left|   X_i \right|^4 \gmid Y^m, A^m} \right], \notag
\end{align}
where (a) follows from Jensen's inequality \eqref{eq:LpBound_1}, (b) follows from  \eqref{eq:ab4}, and (c) follows from Jensen's inequality. Taking the expectation of both sides leads to 
\begin{align}
\ex{ \left|\cE_{m,n} \right|^2 }  &  \le  \frac{16}{n} \sum_{i=1}^n \ex{ \left|X_i\right|^4}   \le 16 B, \notag
\end{align}
where the second inequality follows from Assumption 2.

Next, observe that the conditional distribution $Y_{m+1}$ given $X^n$ is zero-mean Gaussian with variance $1 + \frac{1}{n}\|X^n\|^2$. Thus, 
\begin{align}
\ex{ Y^4_{m+1} } &=  \ex{ \ex{ Y^4 \gmid X^n } } \notag \\
& = 3\ex{  \left(1 + \tfrac{1}{n} \|X^n\|^2\right)^2 } \notag\\
& \overset{(a)}{\le} 6  \left(1 +\tfrac{1}{n^2}\ex{  \|X^n\|^4} \right)  \notag\\
& \overset{(b)} \le 6  (1 + B ), \notag
\end{align}
where (a) follows from \eqref{eq:ab2} and (b) follows from \eqref{eq:LpBound_1} and Assumption 2.  Along similar lines, 
\begin{align}
\MoveEqLeft \ex{\left|  Y_{m+1} - \ex{ Y_{m+1} \gmid Y^m, A^{m+1}}\right|^4 } \notag \\
&\overset{(a)}{ \le} 8 \ex{\left|  Y_{m+1} \right|^4 } +  8 \ex{ \left| \ex{ Y_{m+1} \gmid Y^m, A^{m+1}}\right|^4 }  \notag\\
& \overset{(b)}{\le} 16 \ex{\left|  Y_{m+1} \right|^4 } \notag \\
& \le 96 (1+B), \notag
\end{align} 
where (a) follows from \eqref{eq:ab4} and (b) follows from Jensen's inequality. 

Finally, we note that the proof of \eqref{eq:IMID_moment_2} can be found in the proof of Lemma~\ref{lem:PMID_var}.

\subsection{Proof of Lemma~\ref{lem:post_dist_to_id}}\label{proof:lem:post_dist_to_id}

The starting point for these identities is to observe that the differences between the measurements $Y_i$ and $Y_j$ and their conditional expectations given the data $(Y^m, A^m ,A_i, A_j)$ can be expressed as
\begin{align}
\begin{bmatrix} Y_i - \widehat{Y}_i \\ Y_j - \widehat{Y}_j  \end{bmatrix}  & = \begin{bmatrix} A^T_i \\ A^T_j \end{bmatrix} \left( X^n - \widehat{X}^n\right)  +\begin{bmatrix} W_i \\ W_j \end{bmatrix} \label{eq:dist_YiYj} ,
\end{align}
where $\widehat{X}^n = \ex{ X^n \mid Y^m, A^m}$ is the signal estimate after the first $m$ measurements.

\subsubsection{Proof of Identity~\eqref{eq:id_var}} 
Starting with \eqref{eq:dist_YiYj}, we see that the conditional variance of $Y_i$  can be expressed in terms of the posterior covariance matrix of the signal:
\begin{align}
\var(Y_i \gmid Y^m, A^m, A_i) & = A_i^T \cov(X^n \gmid Y^m, A^m) A_i  + 1. \notag
\end{align}
Taking the expectation of both sides with respect to $A_i$ yields 
\begin{align}
\MoveEqLeft \bEx_{A_i} \left[ \var(Y_i \gmid Y^m, A^m, A_i) \right] \notag\\
 & = \bEx_{A_i} \left[ A_i^T \cov(X^n \gmid Y^m, A^m) A_i \right]  + 1 \notag\\
 & = \bEx_{A_i} \left[ \gtr\left( A_i A_i^T \cov(X^n \gmid Y^m, A^m) \right)  \right]  + 1 \notag\\
 & =\gtr\left( \bEx_{A_i} \left[  A_i A_i^T \right] \cov(X^n \gmid Y^m, A^m)  \right)  + 1 \notag\\
& =\frac{1}{n} \gtr\left(\cov(X^n \mid Y^m, A^m) \right)  + 1, \notag
\end{align}
where we have used the fact that $\ex{ A_i A_i^T}  = \frac{1}{n} I_{m \times m}$. 

\subsubsection{Proof of Identity~\eqref{eq:id_cov}}
Along the same lines as the conditional variance, we see that the conditional covariance between  of $Y_i$ and $Y_j$ is given by
\begin{align}
\cov(Y_i , Y_j \mid Y^m, A^m, A_i, A_j) & = A_i^T \cov(X^n \mid Y^m, A^m) A_j . \notag
\end{align}
Letting $K = \cov(X^n \mid Y^m, A^m)$, the expectation of the squared covariance with respect to $A_i$ and $A_j$ can be computed as follows:
\begin{align}
\MoveEqLeft \bEx_{A_i, A_j} \left[ \left( \cov(Y_i , Y_j \mid Y^m, A^m, A_i, A_j)  \right)^2 \right] \notag\\
& = \bEx_{A_i, A_j} \left[  \left( A_i^T K A_j \right) \left( A_j^T K A_i^T\right) \right] \notag\\
& =  \frac{1}{n} \bEx_{A_i} \left[  A_i^T K^2 A_i  \right] \notag\\
& =  \frac{1}{n} \bEx_{A_i} \left[ \gtr\left( A_i  A_i^T K^2 \right)  \right] \notag\\
& =  \frac{1}{n}\gtr\left(  \bEx_{A_i} \left[  A_i  A_i^T \right] K^2  \right) \notag\\
& =  \frac{1}{n^2}\gtr\left( K^2  \right), \notag
\end{align}
where we have used the fact that $A_i$ and $A_j$ are independent with $\ex{ A_i A_i^T}  = \frac{1}{n} I_{m \times m}$. Noting that $\gtr(K^2) = \| K\|_F^2$ completes the proof of Identity~\eqref{eq:id_cov}.

\subsubsection{Proof of Identity~\eqref{eq:id_post_var}} For this identity, observe that the measurement vectors  $(A_i,A_j)$ and the noise terms $(W_i, W_j)$ in \eqref{eq:dist_YiYj} are Gaussian and independent of the signal error. Therefore, the conditional distribution of  $(Y_i - \widehat{Y}_i, Y_j - \widehat{Y}_j)$ given $(X^n, Y^m, A^m)$ is i.i.d.\ Gaussian with mean zero and covariance 
\begin{align}
\cov\left( \begin{bmatrix} Y_i - \widehat{Y}_i \\ Y_j - \widehat{Y}_j  \end{bmatrix} \; \middle | \;  X^n,  Y^m, A^m\right)  =
 \begin{bmatrix} 1 + \SE_m& 0 \\ 0 & 1 + \SE_m\end{bmatrix}  .  \notag
\end{align}
Using the fact that the expected absolute value of standard Gaussian variable is equal to $\sqrt{2/\pi}$, we see that the conditional absolute moments are given by
\begin{gather}
\ex{ \big| Y_{i} - \widehat{Y}_i  \big | \; \middle | \;   X^n, Y^m ,A^m  }  = \sqrt{ \frac{2}{\pi} }  \sqrt{1 + \SE_m} \notag\\
\ex{ \big | Y_{i} - \widehat{Y}_i \big | \big |  Y_{j} - \widehat{Y}_j \big |  \; \middle |  \;   X^n, Y^m, A^m  }  = \frac{2}{\pi} (1 + \SE_m). \notag
\end{gather}
Taking the expectation of both sides with respect to the posterior distribution of $X^n$ given $(Y^m, A^m)$ leads to 
\begin{gather}
\ex{ \big| Y_{i} - \widehat{Y}_i  \big | \; \middle | \;  Y^m ,A^m  }  = \sqrt{ \frac{2}{\pi} } \ex{  \sqrt{1 + \SE_m} \mid Y^m, A^m } \notag\\
\ex{ \big | Y_{i} - \widehat{Y}_i \big | \big |  Y_{j} - \widehat{Y}_j \big |  \; \middle |  \;   Y^m, A^m  }  = \frac{2}{\pi} (1 + V_m). \notag
\end{gather}
Finally, we see that the conditional covariance is given by
\begin{align}
\MoveEqLeft  \cov\left( \big| Y_{i} - \widehat{Y}_i \big| , \big | Y_{j} - \widehat{Y}_j \big | \;   \middle| \;  Y^m, A^{m}  \right) \notag \\
&  = \frac{2}{\pi} \ex{  \left( \sqrt{ 1 + \SE_m}\right)^2  \; \middle | \;  Y^m, A^m}  \notag\\
&  \quad - \frac{2}{\pi}  \left( \ex{  \sqrt{1 + \SE_m}  \; \middle | \; Y^m, A^m} \right)^2  \notag\\
& = \frac{2}{\pi} \var\left( \sqrt{ 1 + \SE_m}  \mid   Y^m, A^m \right).  \notag
\end{align}
This completes the proof of identity \eqref{eq:id_post_var}.

\subsection{Proof of Lemma~\ref{lem:SE_bound_1}}\label{proof:lem:SE_bound_1}

To simplify  notation, we drop the explicit dependence on the problem dimensions and write $\SE$ and $\PSE$ instead of $\SE_{m,n}$ and $\PSE_{m,n}$. Also, we use $U = (Y^m, A^m)$ to denote the first $m$ measurements.  Using this notation, the posterior distribution is given by $P_{X^n \mid U}$ and the posterior variance is  $\PSE = \ex{ \SE \gmid U}$. 

\subsubsection{Proof of Inequality \eqref{eq:weak_dec_1}} 
To begin, let $\SE'_m$ be a conditionally independent copy of $\SE$ that is drawn according to the posterior distribution $P_{\SE \mid U}$. Starting with the fact that the posterior variance can expressed as $\PSE= \ex{ \SE' \mid U}$, the absolute deviation between $\SE$ and $\PSE$ can be upper bounded using the following series of inequalities:
\begin{align}
\ex{ \left| \SE - \PSE \right| } 
& = \bEx_{U} \Big[ \bEx \big[ \left| \SE -\ex{ \SE' \gmid U} \right| \, \big | \,  U \big] \Big] \notag\\
& \overset{(a)}{\le} \bEx_{U} \Big[ \bEx \big[ \left| \SE - \SE'  \right| \, \big | \,  U \big] \Big] \notag\\
& =  \ex{ \left| \SE - \SE'  \right| } \notag\\
 & =\bEx\bigg[  \bigg |  \left( \sqrt{1 +  \SE}  + \sqrt{1 + \SE'} \right) \notag\\
 & \qquad \times \left( \sqrt{1 +  \SE}  -\sqrt{1 + \SE'} \right)\bigg | \bigg] \notag\\
 & \overset{(b)}{\le} \sqrt{  \ex{ \left|  \sqrt{ 1 +  \SE} +  \sqrt{ 1 +  \SE' } \right|^2} } \notag \\
& \quad \times  \sqrt{ \ex{ \left|\sqrt{1 +  \SE} -\sqrt{1+  \SE'} \right|^2 }} ,\label{eq:SE_bound_1_b} 
\end{align}
where (a)  follows from Jensen's inequality and the fact that $\SE$ and $\SE'$ are conditionally independent given $U$ and (b) follows from the Cauchy-Schwarz inequality.

For the first term on the right-hand side of \eqref{eq:SE_bound_1_b}, observe that
\begin{align}
 \ex{ \left|  \sqrt{ 1 +  \SE} +  \sqrt{ 1 +  \SE' } \right|^2}
 & \overset{(a)}{\le} 2 \ex{ ( 1 +  \SE)} +  2 \ex{ (1 +  \SE' ) }  \notag\\
 & \overset{(b)}{=}  4( 1 + M_{m,n}) \notag\\
  & \le C_B ,\label{eq:SE_bound_1_c} 
\end{align} 
where (a) follows from \eqref{eq:ab2} and (b) follows from the fact that $\SE$ and $\SE'$ are identically distributed. 

For the second term on the right-hand side of \eqref{eq:SE_bound_1_b}, observe that
\begin{align}
\MoveEqLeft  \ex{ \left|\sqrt{1 +  \SE} -\sqrt{1+  \SE'} \right|^2 } \notag\\
& \overset{(a)}{=} 2  \ex{ \var( \sqrt{1 + \SE}  \gmid U)} \notag\\
 & \overset{(b)}{=} \pi \ex{   \cov\left( \left| Z_{m+1} \right| , \left|Z_{m+2}\right| \,  \big | \,  U \right)} , \label{eq:SE_bound_1_d} 
\end{align}
with $Z_{i}    =  Y_{i} - \ex{ Y_{i} \mid Y^m, A^{m}, A_{i}} $. Here,  (a) follows from the conditional variance decomposition \eqref{eq:cond_var_decomp} and (b) follows from Identity~\eqref{eq:id_post_var}. 

To bound the expected covariance, we apply Lemma~\ref{lem:cov_bnd} with $p=1$ to obtain
\begin{align}
 \MoveEqLeft  \ex{ \cov\left( \left| Z_{m+1} \right| , \left|Z_{m+2}\right| \mid U  \right)} \notag\\
 & \le  2 \sqrt{\ex{ \sqrt{   \ex{ Z_{m+1}^4 \mid U} \ex{  Z_{m+2}^4 \mid U }}} }  \notag\\
 & \quad \times  \sqrt{ I\left(|Z_{m+1} |; |Z_{m+2}|  \, \middle  | \,  U  \right) } . \label{eq:SE_bound_1_e} 
\end{align}
For the first term on the right-hand side, observe that
\begin{align}
 \ex{ \sqrt{   \ex{ Z_{m+1}^4 \gmid U} \ex{  Z_{m+2}^4 \gmid U}}} 
& \overset{(a)}{\le} \sqrt{\ex{ Z_{m+1}^4}  \ex{  Z_{m+2}^4}}  \notag \\
& \overset{(b)}{\le} C_B, \label{eq:SE_bound_1_f} 
\end{align}
where (a) follows from the Cauchy-Schwarz inequality and  (b) follows from \eqref{eq:Ybar_moment_4}.

Combining \eqref{eq:SE_bound_1_b}, \eqref{eq:SE_bound_1_c}, \eqref{eq:SE_bound_1_d}, \eqref{eq:SE_bound_1_e}, and \eqref{eq:SE_bound_1_f} yields 
  \begin{align}
  \ex{ \left| \SE - \PSE \right|} 
 \le C_B \cdot \left|  I\left(|Z_{m+1} |; |Z_{m+2}|  \, \middle  | \,   U   \right)  \right|^{\frac{1}{4}}. \notag
  \end{align}
 Thus, in order to complete the proof, we need to show that the mutual information term can be upper bounded in terms of the second order MI difference sequence. To this end, observe that
 \begin{align}
\MoveEqLeft I\left(|Z_{m+1} |; |Z_{m+2}|  \, \big  | \,   U   \right)  \notag\\
& \overset{(a)}{\le} I\left(Z_{m+1} ; Z_{m+2}  \, \middle  | \,   U  \right)  \notag\\
& \overset{(b)}{ \le} I\left(Z_{m+1}, A_{m+1} ; Z_{m+2}, A_{m+2}   \, \middle  | \,   U \right)  \notag\\
&  \overset{(c)}{=}  I\left(Y_{m+1}, A_{m+1} ; Y_{m+2}, A_{m+2}   \, \middle  | \,   U  \right) \notag \\
& \overset{(d)}{ =}  I\left(Y_{m+1}  ; Y_{m+2}   \, \middle  | \,   U, A_{m+1}, A_{m+2}  \right) \notag\\
& \overset{(e)}{=} - I''_m,  \notag
 \end{align}
 where (a) and (b) both follow from the  data processing inequality for mutual information,  (c) follows from the fact that, given $(U, A_{m+i})$, there is a one-to-one mapping between $Z_{m+i}$ and $Y_{m+i}$, (d) follows from the fact that measurements are generated independently of everything else (Assumption~1), and (e) follows from \eqref{eq:Ipp_alt}. This completes the Proof of Inequality \eqref{eq:weak_dec_1}.

\subsubsection{Proof of Inequality \eqref{eq:weak_dec_2}}
The main idea behind this proof is to combine identity \eqref{eq:id_cov} with the covariance bound in Lemma~\ref{lem:cov_bnd}. The only tricky part is that the expectation with respect to  $(A_{m+1}, A_{m+2})$ is taken with respect to the squared Frobenius norm whereas the expectation with respect to $U$ is taken with respect to the square root of this quantity. 

To begin, observe that for each realization $U =u$, we have
\begin{align}
\MoveEqLeft \frac{1}{n^2}  \left \| \cov(X^n \mid U  = u ) \right\|^2_F \notag\\
& \overset{(a)}{=} \ex{ \left| \cov\left(Y_{m+1}, Y_{m+2} \mid U= u , A_{m+1}^{m+2}  \right) \right|^2  } \notag \\
& \overset{(b)}{\le} 4  \sqrt{  \ex{ Y_{m+1}^4 \mid  U = u }\ex{ Y_{m+2}^4 \mid  U =u } } \notag\\
& \quad \times \sqrt{ I(Y_{m+1}; Y_{m+2} \mid U = u, A_{m+1}^{m+2})}, \notag 
\end{align}
where (a) follows from \eqref{eq:id_cov} and (b) follows from Lemma~\ref{lem:cov_bnd} with $p =2$ and the fact that 
\begin{multline}
\ex{  \ex{ Y_{m+1}^4 \mid  U = u, A_{m+1}^{m+2} }\ex{ Y_{m+2}^4 \mid  U =u  , A_{m+1}^{m+2}}  } \notag\\
  =  \ex{ Y_{m+1}^4 \mid  U = u }  \ex{ Y_{m+2}^4 \mid  U =u }. \notag
\end{multline}
Taking the expectation of the square root of both sides with respect to $U$ leads to
\begin{align}
\MoveEqLeft \frac{1}{n} \ex{  \left \| \cov(X^n \mid U ) \right\|_F}  \notag\\
&\overset{(a)}{\le}  4  \left| \ex{ Y_{m+1}^4 }  \ex{ Y_{m+2}^4}   I(Y_{m+1}; Y_{m+2} \mid U , A_{m+1}^{m+2})\right|^{\frac{1}{4}} \notag \\
&\overset{(b)}{\le}  C_B \cdot \left|   I(Y_{m+1}; Y_{m+2} \mid U , A_{m+1}^{m+2})\right|^{\frac{1}{4}}  \notag \\
&\overset{(c)}{=}  C_B\cdot  \left|   I''_m \right|^{\frac{1}{4}},  \notag
\end{align}
where (a) follows from the Cauchy-Schwarz inequality and Jensen's Inequality, (b) follows from \eqref{eq:Y_moment_4}, and (c) follows from  \eqref{eq:Ipp_alt}. This completes the Proof of Inequality \eqref{eq:weak_dec_2}.

\subsection{Proof of Lemma~\ref{lem:Delta_alt}} \label{proof:lem:Delta_alt}

The mutual information difference density can be decomposed as
\begin{align}
\IMID_{m,n} & = \imath(X^n; \Ybar_{m+1} \mid Y^m, A^{m+1}) \notag \\
& =  - \log\left( f_{\Ybar_{m+1} \mid Y^m, A^{m+1}}\left(\Ybar_{m+1}\right) \right)   \notag\\
& \quad -  \frac{1}{2}\log (2 \pi) - \frac{1}{2} W_{m+1}^2, \notag 
\end{align}
where  $f_{\Ybar_{m+1} \mid Y^m, A^{m+1}}(y)$  denotes the conditional density function of the centered measurement evaluated with the random data $(Y^m, A^{m+1})$.  Therefore, for  every $\sigma^2 > 0$, the Kullback--Leibler divergence between $P_{\Ybar_{m+1} \mid Y^m, A^{m+1}}$ and the Gaussian distribution $\normal(0, \sigma^2)$ can be expressed as
\begin{align}
\MoveEqLeft  D_\mathrm{KL}\left( P_{\Ybar_{m+1} \mid Y^m, A^{m+1} } \, \middle \|\, \normal(0, \sigma^2 ) \right)  \notag \\
& =  \int \left( \frac{1}{2 \sigma^2}y^2   + \frac{1}{2} \log\left ( 2\pi \sigma^2 \right ) \right)  f_{\Ybar_{m+1} \mid Y^m, A^{m+1}}(y) \dd y \notag \\
& \quad +  \int  \log\left(  f_{\Ybar_{m+1} \mid Y^m, A^{m+1}}(y)  \right)   f_{\Ybar_{m+1} \mid Y^m, A^{m+1}}(y) \dd y \notag \\
& = \frac{1}{2 \sigma^2} \ex{ \Ybar_{m+1}^2  \gmid  Y^m, A^{m+1} } + \frac{1}{2} \log(2 \pi \sigma^2) \notag\\
& \quad  - \ex{ \IMID_{m,n}  \gmid Y^m, A^{m+1}} -  \frac{1}{2} \log (2 \pi) - \frac{1}{2}. \notag
\end{align}
Taking the expectation with respect to $A_{m+1}$ and rearranging terms leads to 
\begin{align}
\MoveEqLeft  \bEx_{A_{m+1}} \left[ D_\mathrm{KL}\left( P_{\overline{Y}_{m+1} \mid Y^m, A^{m+1} } \, \middle \|\, \normal(0, \sigma^2 ) \right) \right]   \notag\\
& = \ \frac{1}{2} \log( \sigma^2) - \PMID_{m,n}   + \frac{1}{2 } \left( \frac{1 +  \PSE_{m,n}}{\sigma^2}   - 1 \right), \label{eq:Delta_alt_c} 
\end{align}
where we have used the fact that the conditional variance is given by \eqref{eq:Ybar_var_cond}. 

At this point, Identity~\eqref{eq:DeltaP_alt} follows immediately by letting $\sigma^2 =1+ V_{m,n}$. For Identity~\eqref{eq:Delta_alt},  let $\sigma^2 = 1+M_{m,n}$ and note that the expectation of the last term in \eqref{eq:Delta_alt_c} is equal to zero.

\subsection{Proof of Lemma~\ref{lem:DeltaP_bound} }\label{proof:lem:DeltaP_bound} 

The error vector $\bar{X}^n = X^n  - \ex{X^n \gmid Y^m, A^m}$ has mean zero by construction. Therefore, by \eqref{eq:Ybar_alt} and Lemma~\ref{lem:cclt}, the posterior non-Gaussianness satisfies 
\begin{align}
\Delta^P_{m,n} & \le \frac{1}{2} \ex{ \left| \SE_{m,n} - V_{m,n} \right| \, \big | \, Y^m, A^m }  \notag\\
& + C \cdot \left| \tfrac{1}{n}  \| \cov(X^n \gmid Y^m, A^m)\|_F \left(1 +  \widetilde{V}_{m,n}^2 \right)  \right|^\frac{2}{5}, \notag
\end{align}
where $\widetilde{V}_{m,n} = \sqrt{ \ex{ \SE_{m,n}^2 \gmid Y^m, A^m}}$. 
Taking the expectation of both sides and using the Cauchy-Schwarz inequality and Jensen's inequality leads to
\begin{multline}
\ex{ \Delta^P_{m,n}}  \le \frac{1}{2} \ex{ \left| \SE_{m,n} - V_{m,n} \right|} \notag \\
 + C \cdot \left| \tfrac{1}{n} \ex{ \left\|  \cov(X^n \gmid Y^m, A^m)\right\|_F} \big (1 \!+\! \sqrt{ \ex{ \SE_{m,n}^2}} \big)  \right|^\frac{2}{5}.
\end{multline}
Furthermore, combining this inequality with Lemma~\ref{lem:SE_bound_1} and \eqref{eq:SE_moment_2} gives
\begin{align}
\ex{ \Delta^P_{m,n}}  \le  C_B \cdot \left[ \left| I''_{m,n} \right|^\frac{1}{4}  +  \left| I''_{m,n} \right|^\frac{1}{10}\right].   \notag
\end{align} 
Finally, since $| I''_{m,n}|$ can be bounded uniformly by a constant that depends only on $B$, we see that the dominant term on the right-hand side is the one with the smaller exponent. This completes the proof. 

\subsection{Proof of Lemma~\ref{lem:Vdev_to_logVdev}}\label{proof:lem:Vdev_to_logVdev}

To simplify notation we drop the explicit dependence on the problem parameters and write $M$ and $V$ instead of $M_{m,n}$ and $V_{m,n}$.  The first inequality in \eqref{eq:Vdev_to_logVdev}  follows immediately from the fact that the mapping $x \mapsto \log(1 + x)$ is one-Lipschitz on $\reals_+$. 

Next, letting $\PSE'$ be an independent copy of $\PSE$, the absolute deviation of the posterior variance can be upper bounded as follows:
\begin{align}
\ex{ \left| \PSE - M\right|}  & = \ex{ \left| \PSE - \ex{ \PSE'} \right|}  \notag\\
& \overset{(a)}{\le} \ex{ \left| \PSE -  \PSE' \right|}  \notag\\
& \overset{(b)}{\le}  \sqrt{  \ex{ \left(1 + \PSE \right)^2} \ex{\left|  \log\left( \frac{ 1 + \PSE}{ 1 + \PSE'}    \right)\right| }   }, \label{eq:PSE_dev_c_a}
\end{align}
where (a) follows from Jensen's inequality and (b) follows from applying Lemma~\ref{lem:L1_to_Log} with $X = 1 + \PSE_m$ and $Y = 1 +\PSE'_m$.

The first expectation on right-hand side of \eqref{eq:PSE_dev_c_a} obeys:
\begin{align} 
\ex{ \left(1 + \PSE_m \right)^2 } \overset{(a)}{\le} 2(1 + \ex{ \PSE^2}) \overset{(b)}{\le} C_B,  \label{eq:PSE_dev_c_b}
\end{align}
where (a) follows from \eqref{eq:ab2} and (b) follows form Jensen's inequality and \eqref{eq:SE_moment_2}.

For the second expectation on the right-hand side of \eqref{eq:PSE_dev_c_a}, observe that by the triangle inequality, we have, for every $t \in \reals$, 
\begin{align}
\left|  \log\left( \frac{ 1 + \PSE}{ 1 + \PSE'}    \right)\right|  & \le \left|  \log\left( 1 + \PSE  \right) - t \right|  +  \left|  \log\left( 1 + \PSE'  \right) - t \right| .  \notag\notag
\end{align}
Taking the expectation of both sides and minimizing over $t$ yields 
\begin{align}
\ex{\left|  \log\left( \frac{ 1 + \PSE}{ 1 + \PSE'}    \right)\right| }  & \le  \min_{ t \in \reals} 2 \ex{  \left|  \log\left( 1 + \PSE  \right) - t \right| } . \label{eq:PSE_dev_c_c}
\end{align}
Plugging \eqref{eq:PSE_dev_c_a} and \eqref{eq:PSE_dev_c_b} back into \eqref{eq:PSE_dev_c_c} completes the proof of Lemma~\ref{lem:Vdev_to_logVdev}.

\subsection{Proof of Lemma~\ref{lem:post_var_smoothness}}\label{proof:lem:post_var_smoothness}
 Let $U = (Y^m, A^m)$ and $U_k= (Y_{m+1}^{m+k}, A_{m+1}^{m+k})$.  We use the fact that the posterior variance can be expressed in terms of the expected variance of a future measurement. Specifically, by  \eqref{eq:id_var}, it follows that for any integer $i > m + k$, we have
\begin{align}
\PSE_{m}  & =  \bEx_{A_{i}}\left[  \var\left( Y_{i}  \gmid U,  A_{i}\right)  \right] - 1 \notag\\
\PSE_{m+k}  & =  \bEx_{A_{i}}\left[ \var\left( Y_{i}  \gmid U, U_k, A_{i}\right) \right] - 1. \notag
\end{align}
Accordingly, the expectation of the absolute difference can be upper bounded as
\begin{align}
\MoveEqLeft \ex{ \left|  \PSE_{m+k} - \PSE_m  \right|} \notag \\
& =  \ex{ \left |  \bEx_{A_{i}} \left[ \var(Y_{i} \gmid U, U_k,  A_{i}) - \var(Y_{i} \gmid U, A_{i}) \right] \right|} \notag \\
& \overset{(a)}{\le}  \ex{ \left |  \var(Y_{i} \gmid U, U_k, A_{i}) -  \var(Y_{i} \gmid U, A_{i})  \right|}  \notag\\
&  \overset{(b)}{\le}  C \cdot \sqrt{\ex{Y^4_{i}}   \ex{ D_\text{KL}(P_{Y_{i} \mid U, U_k, A_i} \, \middle \| \,  P_{Y_{i} \mid U, A_i}) }}  \notag\\
&  \overset{(c)}{\le}   C_B \cdot  \sqrt{  \ex{ D_\text{KL}(P_{Y_{i} \mid U, U_k, A_i} \, \middle \| \,  P_{Y_{i} \mid U, A_i}) }} ,
 \label{eq:post_var_smoothness_b} 
\end{align}
where (a) follows from Jensen's inequality,  (b) follows from Lemma~\ref{lem:var_bnd} with $p =1$, and (c) follows from \eqref{eq:Y_moment_4}.

Next, the expected Kullback--Leibler divergence can be expressed in terms of a conditional mutual information,
\begin{align}
\MoveEqLeft \ex{ D_\text{KL}(P_{Y_{i} \mid U, U_k, A_i} \, \middle \| \,  P_{Y_{i} \mid U, A_i}) } \notag \\
& = I(Y_{i} ; U_k \mid  U, A_{i}) \notag\\
&\overset{(a)}{=} I(Y_{i} ; Y_{m+1}^{m+k} \gmid Y^m , A^{m+k}, A_i) \notag\\
& \overset{(b)}{=}h(Y_{i}   \gmid  Y^m , A^{m}, A_i) - h(W_i)  \notag \\
& \quad - h(Y_{i}   \gmid  Y^{m+k} , A^{m+k}, A_i) +  h(W_i)  \notag\\
& \overset{(c)}{=}h(Y_{m+1}   \mid  Y^m , A^{m+1}) - h(W_{m+1})  \notag \\
& \quad - h(Y_{m+k + 1}   \gmid  Y^{m+k} , A^{m+k+1}) + h(W_{m+k}) \notag \\
& \overset{(d)}{=}I(X^n ; Y_{m+1}  \gmid Y^m , A^{m+1})  \notag\\
& \quad - I(X^n; Y_{m+k+1}   \gmid  Y^{m+k} , A^{m+k+1})  \notag\\
& \overset{(e)}{=} I'_{m,n} - I'_{m+k,n}   \label{eq:post_var_smoothness_c} 
\end{align}
where (a) follows from the definitions of $U$ and $U_k$ and the fact that the measurements are independent of everything else, (b) follows from expanding the mutual information in terms of the differential entropy of $Y_i$, (c) follows from the fact that future measurements are identically distributed given the past, (d) follows from the fact that $h(W_{m} ) = h(Y_m \gmid X^n, A^m  )$, and (e) follows from \eqref{eq:Ip_alt}.

Combining \eqref{eq:post_var_smoothness_b}  and \eqref{eq:post_var_smoothness_c}, we see that the following inequality holds for all integers $m$ and $k$, 
\begin{align}
\ex{ \left| V_{m,n} - V_{k,n} \right| }  \le C_B \cdot \left|  I'_{m,n} - I'_{k,n} \right|^\frac{1}{2}.  \notag
\end{align}
Moreover, we can now bound the deviation over $\ell$ measurements using
\begin{align}
  \frac{1}{\ell}  \sum_{i=m}^{m+\ell-1}  \ex{ \left | V_m - V_k   \right|  } 
&  \le  C_B \cdot \frac{1}{\ell}  \sum_{k=m}^{m + \ell -1}  \left| I'_m - I'_{k}    \right|^\frac{1}{2}  \notag \\
  &\le  C_B \cdot   \left|  I'_{m} - I'_{m + \ell-1}  \right|^\frac{1}{2},    \notag
\end{align}
where the second inequality follows from the fact that $I'_{m,n}$ is non-increasing in $m$ (see Section~\ref{sec:multivariate_MI_MMSE}). This completes the proof of  Lemma~\ref{lem:post_var_smoothness}.

\subsection{Proof of Lemma~\ref{lem:PMID_var}}\label{proof:lem:PMID_var}

Starting with the triangle inequality, the sum of the posterior MI  difference satisfies, for all $t \in \reals$, 
\begin{align}
 \left| \sum_{i =m}^{m+ \ell-1} \PMID_i - t \right|  \le   \left| \sum_{i =m}^{m+\ell-1} \PMID_{i}  - \IMID_i \right|  +   \left| \sum_{i =m}^{m+\ell-1} \IMID_{i}  - t \right|. \notag
\end{align}
Taking the expectation of both sides and minimizing over $t$ leads to
\begin{multline}
 \inf_{t \in \reals} \ex{  \left| \sum_{i =m}^{m+ \ell-1} \PMID_i - t \right|} 
\le  \inf_{t \in \reals} \ex{ \left| \sum_{i =m}^{m+\ell-1} \IMID_{i}  - t \right|} \\+  \ex{ \left| \sum_{i =m}^{m+\ell-1} \IMID_{i}  - \PMID_i \right|}  \label{eq:PMID_var_b} .
\end{multline}

For the first term in \eqref{eq:PMID_var_b}, observe that
\begin{align}
\inf_{t \in \reals} \ex{ \left| \sum_{i =m}^{m+\ell-1} \IMID_i  - t \right|}  
& \overset{(a)}{\le}  \inf_{t \in \reals}  \sqrt{ \ex{ \left| \sum_{i =m}^{m+\ell-1} \IMID_{i}  - t \right|^2} }   \notag\\
& =   \sqrt{ \var\left(  \sum_{i=m}^{m+\ell-1} \IMID_{i}  \right) } , \label{eq:PMID_var_c}
\end{align}
where (a) follows from Jensen's inequality.  Furthermore, the variance obeys the upper bound
\begin{align}
\var\left(  \sum_{i =m}^{m+\ell-1} \IMID_{i}  \right) & =  \var\left(  \sum_{i = 0}^{m + \ell -1} \IMID_i -  \sum_{i = 0}^{m-1} \IMID_{i}  \right) \notag\\
& \overset{(a)}{\le} 2 \var\left(  \sum_{i = 0}^{m + \ell -1} \IMID_i \right) + 2 \var\left( \sum_{i = 0}^{m-1} \IMID_{i}  \right)  \notag \\
& \overset{(b)}{=} 2 \var\left( \imath\left( X^n ; Y^{m+\ell} \mid A^{m+\ell}  \right) \right) \notag \\
& \quad + 2 \var\left( \imath\left( X^n ; Y^m \mid A^m \right)   \right)  \notag\\
& \overset{(c)}{\le}   C_B \cdot \left( 1  + \tfrac{m+\ell}{n} \right)^2 n   +   C_B \cdot \left( 1  + \tfrac{m}{n} \right)^2 n  \notag\\
& \le  C'_B \cdot \left( 1  + \tfrac{m+\ell}{n} \right)^2 n  \notag
\end{align}
where (a) follows from \eqref{eq:ab2}, (b) follows from the definition of $\IMID_m$, 
and (c) follows from Lemma~\ref{lem:IMI_var}. Plugging this bound back into \eqref{eq:PMID_var_c} gives
\begin{align}
\inf_{t \in \reals} \ex{ \left| \sum_{i =m}^{m+\ell-1} \IMID_i  - t \right|}  
& \le  C_B \cdot \left( 1  + \frac{m+\ell}{n} \right) \sqrt{ n } \label{eq:PMID_var_d}.
\end{align}

Next, we consider the second term in \eqref{eq:PMID_var_b}.  Note that $\IMID_m$ can be expressed explicitly as follows:
\begin{align}
\IMID_m & = \log\left( \frac{f_{Y_{m+1}|X^n , Y^m,A^{m+1}}( Y_{m+1}  \mid  X^n,  Y^m,A^{m+1} )}{f_{Y_{m+1}|Y^m, A^{m+1}}(Y_{m+1} \mid  Y^m,  A^{m+1} )}    \right) \notag\\
& =  - \log\left( f_{Y_{m+1}|Y^m, A^{m+1}}(Y_{m+1} \mid  Y^m,  A^{m+1} )    \right) \notag\\
& \quad  - \frac{1}{2} W^2_{m+1}   - \frac{1}{2} \log( 2\pi) . \notag
\end{align}
To proceed, we define the random variables 
\begin{align}
\cH_m & \triangleq  - \log\left( f_{Y_{m+1}|Y^m, A^{m+1}}(Y_{m+1} \mid  Y^m,  A^{m+1} )   \right) \notag\\
\widehat{\cH}_m & \triangleq \ex{ \cH_m \mid Y^m, A^m} , \notag
\end{align}
and observe that
\begin{align}
\IMID_m & = \cH_m -  \frac{1}{2} W^2_{m+1}   - \frac{1}{2} \log( 2\pi)  \notag\\
\PMID_m & =  \widehat{\cH}_m -  \frac{1}{2}   - \frac{1}{2} \log( 2\pi) . \notag
\end{align}
Using this notation, we can now write 
\begin{align}
\MoveEqLeft \ex{ \left| \sum_{i =m}^{m+\ell-1} \IMID_{i}  - \PMID_i \right|} \notag \\
& =  \ex{ \left| \sum_{i =m}^{m+\ell-1} \cH_i - \widehat{\cH}_i   + \frac{1}{2}(W_{i}^2 - 1)  \right|}  \notag\\
& \overset{(a)}{\le}  \ex{ \left| \sum_{i =m}^{m+\ell-1} \cH_i - \widehat{\cH}_i  \right|} +   \ex{ \left| \sum_{i =m}^{m+\ell-1}\frac{1}{2}(W_{i}^2 - 1)  \right|} \notag \\ 
& \overset{(b)}{\le}  \sqrt{  \ex{ \left( \sum_{i =m}^{m+\ell-1} \cH_i - \widehat{\cH}_i  \right)^2}}  \notag \\
& \quad +  \sqrt{  \ex{ \left( \sum_{i =m}^{m+\ell-1}\frac{1}{2}(W_{i}^2 - 1)  \right)^2}},  \label{eq:PMID_var_e}
\end{align}
where (a) follows from the triangle inequality and (b) follows from Jensen's inequality. For the first term on the right-hand side, observe that square of the sum can be expanded as follows: 
\begin{align}
\MoveEqLeft  \ex{ \left( \sum_{i =m}^{m+\ell-1} \cH_i - \widehat{\cH}_i  \right)^2}  \notag \\
&  =\sum_{i =m}^{m+\ell-1}   \ex{ \left(  \cH_i - \widehat{\cH}_i  \right)^2}  \notag \\
 & \quad + 2 \sum_{i =m}^{m+\ell-1} \sum_{j = i+1}^{m+ \ell-1}  \ex{ \left(  \cH_i - \widehat{\cH}_i  \right)\left(  \cH_j - \widehat{\cH}_j \right)}. \label{eq:PMID_var_f}
\end{align}

To deal with the first term on the right-hand side of \eqref{eq:PMID_var_f}, observe that
\begin{align}
\ex{ \left(  \cH_i - \widehat{\cH}_i  \right)^2}  & \overset{(a)}{ =}\ex{  \var(  \cH_i \mid Y^i, A^i )} \notag \\
&\overset{(b)}{ \le} \ex{  \var(  \cH_i ) }  \notag\\
& \le \ex{ \left(  \cH_i  - \frac{1}{2} \log(2 \pi) \right)^2 }, \notag
\end{align}
where (a) follows from the definition of $\widehat{\cH}_i$ and (b) follows from the law of total variance \eqref{eq:law_tot_var}. To bound the remaining term, let $U = (Y^m, A^m)$ and let $\widetilde{X}^n_{u}$ be drawn according to the posterior distribution of $X^n$ given $U = u$. Then, the density of $Y_{m+1}$ given $(U, A_{m+1})$ can be bounded as follows:
\begin{align}
\frac{1}{\sqrt{2 \pi}} &\ge  f_{Y_{m+1}|U, A_{m+1}}(y_{m+1} \mid  u,  a_{m+1})  \notag\\
& =   \bEx_{\tilde{X}_u^n} \left[ \frac{1}{\sqrt{2 \pi}}\exp\left( - \frac{1}{2} (y_{m+1} - \langle a_{m+1}, \widetilde{X}_u^n \rangle)^2   \right)     \right]  \notag\\
& \overset{(a)}{\ge}  \frac{1}{\sqrt{2 \pi}} \exp\left( - \frac{1}{2} \bEx_{\tilde{X}_u^n} \left[ (y_{m+1} - \langle a_{m+1}, \widetilde{X}_u^n \rangle)^2   \right]  \right)    \notag \\
& \overset{(b)}{\ge}  \frac{1}{\sqrt{2 \pi}} \exp\left( - y_{m+1}^2  - \bEx_{\tilde{X}_u^n} \left[ ( \langle a_{m+1}, \widetilde{X}_u^n \rangle)^2   \right]  \right) ,    \notag
\end{align}
where (a) follows from Jensen's inequality and the convexity of the exponential and (b) follows from \eqref{eq:ab2}. Using these bounds, we obtain\begin{align}
\MoveEqLeft \ex{  \left(  \cH_i  - \frac{1}{2} \log(2 \pi) \right)^2}  \notag\\
& \le \ex{  \left( Y_{m+1}^2  +  \bEx_{\tilde{X}_U^n} \left[ \left( \langle A_{m+1}, \widetilde{X}_U^n \rangle  \right)^2\right]\right)^2}  \notag\\
& \overset{(a)}{\le} 2 \ex{  Y_{m+1}^4}   + 2 \ex{ \left(  \bEx_{\tilde{X}_U^n} \left[ \left( \langle A_{m+1}, \widetilde{X}_U^n \rangle  \right)^2\right]\right)^2} \notag\\
& \overset{(b)}{\le} 2 \ex{  Y_{m+1}^4}   + 2 \ex{  \left( \langle A_{m+1}, X^n \rangle \right)^4} \notag\\
& \le 4\ex{  Y_{m+1}^4} \notag\\
& \overset{(c)}{\le} C_B,  \notag
\end{align} 
where (a) follows from \eqref{eq:ab2}, (b) follows from Jensen's inequality and the fact that $\widetilde{X}^n_U$ as the same distribution as $X^n$ and is independent of $A_{m+1}$, and (c) follows from \eqref{eq:Y_moment_4}.

To deal with the second term on the right-hand side of \eqref{eq:PMID_var_f},  note that  $\cH_i$ and $\widehat{\cH}_i$ are determined by $(Y^{i+1}, A^{i+1})$, and thus, for all $j > i$, 
\begin{align}
\MoveEqLeft \ex{ \left(  \cH_i - \widehat{\cH}_i  \right)\left(  \cH_j - \widehat{\cH}_j \right) \mid Y^{i+1}, A^{i+1} }  \notag\\
& = \left(  \cH_i - \widehat{\cH}_i  \right)  \ex{ \left(  \cH_j - \widehat{\cH}_j \right) \mid Y^{i+1}, A^{i+1} }  = 0 . \notag
\end{align}
Consequently, the cross terms in the expansion of \eqref{eq:PMID_var_f} are equal to zero, and the first term on the right-hand side of \eqref{eq:PMID_var_e} obeys the upper bound
\begin{align}
\sqrt{ \ex{ \left( \sum_{i =m}^{m+\ell-1} \cH_i - \widehat{\cH}_i  \right)^2} }  &\le C_B \cdot \sqrt{\ell}   \label{eq:PMID_var_g}.
\end{align}

As for the second term on the right-hand side of \eqref{eq:PMID_var_e}, note that $\sum_{i =m}^{m+\ell-1}W_{i}^2$ is chi-squared with $\ell$ degrees of freedom, and thus
\begin{align}
\sqrt{  \ex{ \left( \sum_{i =m}^{m+\ell-1}\frac{1}{2}(W_{i}^2 - 1)  \right)^2}} & = \sqrt{ \frac{\ell}{2} }  \label{eq:PMID_var_h}.
\end{align} 

Plugging \eqref{eq:PMID_var_g} and \eqref{eq:PMID_var_h} back in to \eqref{eq:PMID_var_e} leads to
\begin{align}
\ex{ \left| \sum_{i =m}^{m+\ell-1} \IMID_{i}  - \PMID_i \right|} & \le C_B  \cdot \sqrt{\ell} .  \notag
\end{align}
Finally, combining this inequality with \eqref{eq:PMID_var_b} and \eqref{eq:PMID_var_d} gives
\begin{align}
  \inf_{t \in \reals} \ex{  \left| \sum_{i =m}^{m+ \ell-1} \PMID_i - t \right|} & \le C_B  \cdot \left( \left( 1 + \frac{m+\ell}{n}  \right)  \sqrt{n} + \sqrt{\ell} \right) \notag\\
  & \le C'_B  \cdot \left( \left( 1 + \frac{m}{n}  \right)  \sqrt{n} + \frac{\ell}{\sqrt{n}}  \right), \notag
\end{align}
where the last step follows from keeping only the dominant terms.  This completes the proof of Lemma~\ref{lem:PMID_var}.

\subsection{Proof of Lemma~\ref{lem:post_var_dev_bound}} \label{proof:lem:post_var_dev_bound}

%



Fix any $(m,n, \ell) \in \integers^3$. We begin with the following decomposition, which follows from the triangle inequality: 
\begin{align}
 \MoveEqLeft \inf_{t \in \reals } \ex{ \frac{1}{2} \log(1 + V_m) - t } \notag\\
&\le   \ex{ \left |\frac{1}{2}  \log(1 + V_m)  - \frac{1}{\ell} \sum_{k=m}^{m+\ell+1}\frac{1}{2} \log\left(1 +V_k  \right)  \right|  } \notag \\
& \quad + \ex{ \left| \frac{1}{\ell} \sum_{k=m}^{m+ k+1}  \frac{1}{2} \log\left(1 +V_k  \right)   - \frac{1}{k} \sum_{i=m}^{m+\ell+1} \PMID_k  \right| }  \notag\\
& \quad + \inf_{t \in \reals}  \ex{ \left| \frac{1}{\ell} \sum_{i=m}^{m+ \ell+1}  \PMID_k - t  \right| }. \label{eq:PSE_dev} 
\end{align}

The first term on the right-hand side of \eqref{eq:PSE_dev}  can be bounded in terms of the smoothness of the mutual information  given in Lemma~\ref{lem:post_var_smoothness}. We use the following chain of inequalities:
\begin{align}
 \MoveEqLeft \ex{ \left | \log(1 + V_m)  - \frac{1}{\ell} \sum_{k=m}^{m +\ell-1} \log\left(1 +V_k \right)  \right|  } \notag \\
 & \overset{(a)}{\le} \frac{1}{\ell} \sum_{k=m}^{m+ \ell-1}\ex{ \left | \log(1 + V_m)  - \log\left(1 +V_k \right)  \right|  }  \notag\\
 & \overset{(b)}{\le} \frac{1}{\ell}  \sum_{k=m}^{m+\ell-1}  \ex{ \left | V_m - V_k   \right|  }  \notag\\
  &\overset{(c)}{\le}  C_B \cdot   \left|  I'_{m} - I'_{m + \ell -1}  \right|^\frac{1}{2}    \label{eq:PSE_dev_c} 
\end{align}
where (a) follows from Jensen's inequality, (b) follows from the fact that the mapping $\log(1 +x)  \to x$ is one-Lipschitz on $\reals_+$,  and (c) follows from Lemma~\ref{lem:post_var_smoothness}.

The second term on the right-hand side of \eqref{eq:PSE_dev} is bounded by the relationship between the posterior variance and posterior mutual information difference:
\begin{align}
\MoveEqLeft \ex{ \left| \frac{1}{\ell} \sum_{k=m}^{m+\ell-1} \frac{1}{2} \log\left(1 +V_k  \right)   - \frac{1}{\ell} \sum_{k=m}^{m+\ell-1} \PMID_k  \right| }  \notag\\
& \overset{(a)}{=} \frac{1}{\ell} \sum_{k=m}^{m+ \ell -1} \ex{\Delta_k^P }  \notag\\
& \overset{(b)}{\le}  C_B \cdot \frac{1}{\ell} \sum_{k=m}^{m+\ell-1} \left| I''_k \right|^\frac{1}{10} \notag\\
& \overset{(c)}{\le}   C_B \cdot  \left|\frac{1}{\ell} \sum_{k=m}^{m+\ell-1}\left| I''_k \right| \right|^\frac{1}{10}   \notag\\
& =  C_B \cdot  \left|\frac{1}{k} ( I'_{m+k} - I'_{m} )  \right|^\frac{1}{10}   \notag\\
& \overset{(d)}{\le}   C'_B \cdot \ell^{-\frac{1}{10}  } \label{eq:PSE_dev_d} ,
\end{align}
where (a) follows from Identity~\eqref{eq:DeltaP_alt}, (b) follows from Lemma~\ref{lem:DeltaP_bound}, (c) follows from Jensen's inequality and the non-positiviity of $I''_m$ and (d) follows from the fact that $I'_m$ is  bounded by a constant that depends only on $B$ (see Section~\ref{sec:multivariate_MI_MMSE}). 

Finally, the third term  on the right-hand side of \eqref{eq:PSE_dev} is bounded by Lemma~\ref{lem:PMID_var}. Plugging \eqref{eq:PSE_dev_c}  and \eqref{eq:PSE_dev_d} back into \eqref{eq:PSE_dev} leads to
\begin{multline}
 \inf_{t \in \reals } \ex{ \frac{1}{2} \log(1 + V_m) - t } \le C_B \cdot  \left|  I'_{m} - I'_{m + \ell -1}  \right|^\frac{1}{2} \\ 
+ C_B \cdot   \left[\ell^{-\frac{1}{10}} +  \left(1 + \frac{m  }{n} \right)  \frac{ \sqrt{n}}{\ell}  + \frac{1}{ \sqrt{n}} \right]. \notag
\end{multline}
Combining this inequality with Lemma~\ref{lem:Vdev_to_logVdev} completes the proof of Lemma~\ref{lem:post_var_dev_bound}.

\subsection{Proof of Lemma~\ref{lem:Delta_bound}}\label{proof:lem:Delta_bound} 

Fix any $(m,n, \ell) \in \integers^3$. Combining  Identity \eqref{eq:Delta_decomp}, with Lemmas~\ref{lem:DeltaP_bound}  and \ref{lem:post_var_dev_bound} yields 
\begin{align}
\Delta_{m,n} & =  \ex{ \Delta^P_{m,n}}   + \frac{1}{2}  \ex{\log\left( \frac{  1  +  M_{m,n}}{ 1 + V_{m,n}} \right)} \notag\\
& \begin{multlined}[b] 
\le  C_B \cdot  \Big[ \left| I''_{m,n} \right|^\frac{1}{10}    +  \left | I'_{m,n} - I'_{m+ \ell + 1, n}  \right|^\frac{1}{4} \\
+ \ell^{-\frac{1}{20}}+ \left(1 + \tfrac{m}{n} \right)^\frac{1}{2}   n^\frac{1}{4} \ell^{-\frac{1}{2}}  + n^{-\frac{1}{4}} \Big]. \label{eq:Delta_bound}
\end{multlined}
\end{align}
For the specific choice of $\ell = \lceil n^\frac{5}{6} \rceil$, we have
\begin{align}
 \ell^{-\frac{1}{20}}+ \left(1 + \tfrac{m}{n} \right)^\frac{1}{2}   n^\frac{1}{4} \ell^{-\frac{1}{2}} &  \le n^{-\frac{1}{24}} +   \left(1 + \tfrac{m}{n} \right)^\frac{1}{2}   n^{-\frac{1}{6} }  \notag
\\
 & \le  2\left(1 + \tfrac{m}{n} \right)^\frac{1}{2}   n^{-\frac{1}{24} }. \notag
\end{align}
%
%
Plugging this inequality back into \eqref{eq:Delta_bound} completes the proof of Lemma~\ref{lem:Delta_bound}.

\section{Proofs of Results in Section~\ref{proof:thm:MMSE_fixed_point}}

\subsection{Proof of Lemma~\ref{lem:M_to_M_aug}}\label{proof:lem:M_to_M_aug}

Recall that $M_m$ is the expectation of the posterior variance $V_m$. Therefore, the difference between $M_{m+1}$ and $M_m$ can bounded as follows: 
\begin{align}
\left| M_{m+1}  - M_m \right|  & = \left| \ex{ V_{m+1}  - V_m}  \right|  \notag\\
& \overset{(a)}{\le} \ex{ \left | V_{m+1} - V_{m} \right|} \notag\\
& \overset{(b)}{\le} C_B \cdot \sqrt{I_m''}, \label{eq:MMSE_smooth}
\end{align}
where (a) follows from Jensen's inequality and (b) follows from Lemma~\ref{lem:post_var_smoothness}. Combining \eqref{eq:MMSE_smooth} with the sandwiching relation \eqref{eq:tildeM_sandwich} leads to \eqref{eq:M_to_M_aug}. 
%
%
%
%

\subsection{Proof of Lemma~\ref{lem:MMSE_aug_alt}}\label{proof:lem:MMSE_aug_alt} 

 Let $Q$ be a random matrix distributed uniformly on the set of $(m+1) \times (m+1)$ orthogonal matrices and define the rotated augmented measurements:
\begin{align}
\widetilde{Y}^{m+1} &= Q  \begin{bmatrix} Y^{m} \\ Z_{m+1}  \end{bmatrix} , \qquad \widetilde{A}^{m+1} = Q \begin{bmatrix} A^{m} & \bm{0}_{m \times 1}  \\ A_{m+1} & \sqrt{G_{m+1}} \end{bmatrix}. \notag
\end{align}
Since multiplication by $Q$ is a one-to-one transformation, the augmented MMSE can be expressed equivalently in terms of the rotated measurements:
\begin{align}
\widetilde{M}_m & \triangleq  \frac{1}{n}  \mmse(X^n    \mid Y^{m},  A^m, \cD_{m+1} ) \notag\\
& = \frac{1}{n}  \mmse(X^n \mid \widetilde{Y}^{m+1}, \widetilde{A}^{m+1}). \label{eq:tildeM_alt_a} 
\end{align}

\begin{lemma}\label{lem:tildeA}
The entries of the $(m+1) \times (n+1)$ random matrix $\widetilde{A}^{m+1}$ are i.i.d.\ Gaussian $\normal(0,1/n)$. 
\end{lemma}
\begin{proof}
The first $n$ columns are i.i.d.\ Gaussian $\normal(0, \frac{1}{n} I_{m+1})$ and independent of $Q$ because of the rotational invariance of the i.i.d.\ Gaussian distribution on $A^{m+1}$. The last column of $ \widetilde{A}^{m+1}$ is equal to the product of $\sqrt{G_{m+1}}$ and the last column of $Q$. Since $G_{m+1}$ is proportional to a chi random variable with $m+1$ degrees of freedom and $Q$ is distributed uniformly on the Euclidean sphere of radius one, the last column is also Gaussian $\normal(0, \frac{1}{n} I_{m+1})$; see e.g.\ \cite[Theorem~2.3.18]{gupta:1999}. 
\end{proof}

The key takeaway from Lemma~\ref{lem:tildeA}, is that the distribution on the columns of $\widetilde{A}^{m+1}$ is permutation invariant. Since the distribution on the entries of $X^{n+1}$ is also permutation invariant, this means that the MMSEs of the signal entries are identical, i.e.,
\begin{align}
 \mmse(X_{i} \mid \widetilde{Y}^{m+1}, \widetilde{A}^{m+1}) =  \mmse(X_{j} \mid \widetilde{Y}^{m+1}, \widetilde{A}^{m+1}), \notag
\end{align}
for all $i,j \in [n+1]$. Combining this fact with \eqref{eq:tildeM_alt_a}, we see that the augmented MMSE can be expressed equivalently as
\begin{align}
\widetilde{M}_m
& = \mmse(X_{n+1}  \mid \widetilde{Y}^{m+1}, \widetilde{A}^{m+1}) \notag\\
& = \mmse(X_{n+1} \mid   Y^{m},  A^m, \cD_{m+1}) , \notag
\end{align}
where the last step follows, again, from the fact that multiplication by $Q$ is a one-to-one transformation of the data. This completes the proof of  Lemma~\ref{lem:MMSE_aug_alt}.

\subsection{Proof of Lemma~\ref{lem:MMSE_aug_bound}}\label{proof:lem:MMSE_aug_bound}

This proof is broken into two steps. First, we show that the augmented MMSE satisfies the inequality, 
\begin{align}
 \left|  \widetilde{M}_m  - \ex{  \mmse_X\left( \frac{G_{m+1}}{ 1+ M_m}  \right) }   \right| & \le  C_B \cdot \sqrt{\Delta_m}. \label{eq:M_aug_bound_2} 
\end{align}
Then,  we use the smoothness of the of the single-letter MMSE function $\mmse_X(s)$ to show that
\begin{align}
\MoveEqLeft  \left|  \ex{  \mmse_X\left( \frac{G_{m+1}}{ 1+ M_m}  \right) }  -\mmse_X\left( \frac{m/n}{ 1+ M_m}  \right)   \right|  \notag\\
& \le  C_B  \frac{1 + \sqrt{m}}{n}. \label{eq:M_aug_bound_3}
\end{align}
The proofs of  Inequalities \eqref{eq:M_aug_bound_2} and \eqref{eq:M_aug_bound_3} are given in the following subsections.

\subsubsection{Proof of Inequality~\eqref{eq:M_aug_bound_2}} 
The centered augmented measurement  $\Zbar_{m+1}$ is defined by
\begin{align}
\Zbar_{m+1}  = Z_{m+1} - \ex{ Z_{m+1} \mid  Y^m, A^{m+1} }. \notag
\end{align}
Since $G_{m+1}$ and $X_{n+1}$ are independent of the first $m$ measurements, $\Zbar_{m+1}$ can also be expressed as
\begin{align}
\Zbar_{m+1}  = \sqrt{G_{m+1}} X_{n+1}   + \Ybar_{m+1}, \notag
\end{align}
where $\Ybar_{m+1} = Y_{m+1} - \ex{Y_{m+1} \mid Y^m, A^{m+1}}$ is the centered measurement introduced in Section~\ref{sec:Gaussiannness}. Starting with \eqref{eq:M_aug_alt},  we see that the augmented MMSE can expressed as 
\begin{align}
\widetilde{M}_m & = \mmse( X_{n+1} \gmid Y^m\!, A^m\!,  \Zbar_{m+1} , A_{m+1}, G_{m+1}),  \label{eq:mmse_aug_b} 
\end{align}
where we have used the fact that there is a one-to-one mapping between  $\Zbar_{m+1}$ and $Z_{m+1}$.  

The next step of the proof is to address the extent to which the MMSE in \eqref{eq:mmse_aug_b} would differ if the `noise' term $\Ybar_{m+1}$ were replaced by an independent Gaussian random variable with the same mean and variance. To make this comparison precise, recall that $\ex{ \Ybar_{m+1}} = 0$ and $\var(\Ybar_{m+1}) = 1+ M_m$, and let $Z^*_{m+1}$ be defined by
\begin{align}
Z^*_{m+1}  = \sqrt{G_{m+1}} X_{n+1}  +  Y_{m+1}^*, \notag
\end{align} 
where $Y_{m+1}^* \sim \normal(0,1+M_m)$ is independent of everything else. Note that the MMSE of $X_{n+1}$  with $\Zbar_{m+1}$ replaced by $Z^*_{m+1}$ can be characterized explicitly  in terms of the single-letter MMSE function:
\begin{align}
\MoveEqLeft \mmse( X_{n+1} \mid Y^m, A^m,  Z^*_{m+1} ,A_{m+1}, G_{m+1})  \notag\\
& \overset{(a)}{=} \mmse( X_{n+1} \mid Z^*_{m+1} , G_{m+1})  \notag\\
& = \ex{  \mmse_X\left( \frac{G_{m+1}}{ 1+ M_m}  \right) } \label{eq:mmse_aug_c} 
\end{align}
where (a) follows from the fact that $(Y^m, A^{m+1})$ is independent of $(X_{n+1}, Z^*_{m+1}, G_{m+1})$.

The next step is to bound the difference between \eqref{eq:mmse_aug_b} and \eqref{eq:mmse_aug_c}. To proceed, we introduce the notation
\begin{align}
\cF & = (Y^m, A^m, \Zbar_{m+1} ,  A_{m+1}, G_{m+1} ) \notag\\
\cF^* & =  (Y^m, A^m, Z^*_{m+1} ,  A_{m+1}, G_{m+1}). \notag
\end{align}
Then, using Lemma~\ref{lem:mmse_diff} yields 
\begin{align}
\MoveEqLeft \left| \mmse(X_{n+1} \mid \cF) - \mmse(X_{n+1} \mid \cF^*) \right| \notag \\
& \le 2^\frac{5}{2} \sqrt{ \ex{ X^4_{n+1}} D_\mathrm{KL} \left(  P_{\cF , X_{n+1}}  \, \middle \|\,   P_{\cF^*, X_{n+1}}  \right)}. \notag
\end{align}
By Assumption 2, the fourth moment of $X_{n+1}$ is upper bounded by $B$. The last step is to show that the Kullback--Leibler divergence  is equal to the non-Gaussianenesss $\Delta_m$. To this end, observe that
\begin{align}
\MoveEqLeft D_\mathrm{KL} \left(  P_{\cF , X_{n+1}}  \, \middle \|\,   P_{\cF^*, X_{n+1}}  \right)  \notag\\
& \overset{(a)}{=} \bEx_{X_{n+1}} \left[  D_\mathrm{KL} \left(  P_{\cF \mid X_{n+1}}  \, \middle \|\,   P_{\cF^*\mid X_{n+1}}  \right)  \right] \notag\\
& \overset{(b)}{=}  D_\mathrm{KL} \left(  P_{\Ybar_{m+1} , Y^m, A^{m+1} }  \, \middle \|\,   P_{ Y^*_{m+1}, Y^m, A^{m+1}}  \right)  \notag\\
& \overset{(c)}{=} \bEx_{Y^m, A^{m+1}} \left[  D_\mathrm{KL} \left(  P_{\Ybar_{m+1} \mid Y^m, A^{m+1} }  \, \middle \|\,   P_{ Y^*_{m+1}}  \right)  \right] \notag\\
& = \Delta_{m} , \notag
\end{align}
where (a) follows from the chain rule for Kullback--Leibler divergence, (b) follows from the fact that both $\Ybar_{m+1}$ and $Y^*_{m+1}$ are independent of $(X_{n+1}, G_{m+1})$, and (c) follows from the chain rule for Kullback--Leibler divergence and the fact that $Y^*_{m+1}$ is independent of $(Y^m, A^{m+1})$. This completes the proof of Inequality~\eqref{eq:M_aug_bound_2}.

\subsubsection{Proof of Inequality~\eqref{eq:M_aug_bound_3}} 
Observe that
\begin{align}
\MoveEqLeft \left| \ex{  \mmse_X\left( \frac{G_{m+1}}{ 1+ M_m}  \right) }  - \mmse_X\left( \frac{ m/n }{ 1+ M_m}  \right)  \right|  \notag\\
&\overset{(a)}{\le} \ex{   \left|  \mmse_X\left( \frac{G_{m+1}}{ 1+ M_m}  \right)   - \mmse_X\left( \frac{ m/n }{ 1+ M_m}  \right)  \right| }  \notag\\
& \overset{(b)}{\le}  4 B \ex{ \left |G_{m+1} - \frac{m}{n}  \right |} \notag \\
& \overset{(c)}{\le}  4 B \left(  \ex{ \left |G_{m+1} - \ex{ G_{m+1}} \right|} +  \frac{1}{n} \right) \notag\\
& \overset{(d)}{\le}  4 B \left(  \sqrt{ \var(G_{m+1}) }+ \frac{1}{n} \right)  \notag\\
& = 4  B \Big( \frac{ \sqrt{2(m+1)}}{n} + \frac{1}{n}  \Big)  , \notag
\end{align}
%
where (a) follows from Jensen's inequality, (b) follows from Lemma~\ref{lem:mmseX_bound} and Assumption 2, (c) follows from the triangle inequality and the fact that $\ex{G_{m+1}} = \frac{m+1}{n}$, and (d) follows from Jensen's inequality. 

\section*{Acknowledgment}
G.\ Reeves is grateful to D. Donoho for many helpful conversations as well as the inspiration behind the sandwiching argument that lead to the fixed-point relationships in Section~\ref{sec:MMSE_fixed_point}.

\bibliographystyle{ieeetr}
\bibliography{replica_mmse.bbl}


\end{document}